\documentclass[11pt]{article}%
\usepackage{amsfonts}
\usepackage{amsmath}%
\usepackage{amssymb}%
\usepackage{graphicx}
\usepackage{mathtools}

\usepackage[colorlinks=true,linkcolor=blue,citecolor=red,plainpages=false,pdfpagelabels]%
{hyperref}%

\providecommand{\U}[1]{\protect\rule{.1in}{.1in}}
\newtheorem{theorem}{Theorem}

\newtheorem{lemma}[theorem]{Lemma}

\newtheorem{remark}[theorem]{Remark}

\newenvironment{proof}[1][Proof]{\noindent\textbf{#1.} }{\ \rule{0.5em}{0.5em}}

\numberwithin{equation}{section}

\def\e{{ e}}
\newcommand{\Mod}[1]{\ (\mathrm{mod}\, #1)}

\allowdisplaybreaks

\begin{document}

\title{\vspace{-1.1in}Wave Matrix Lindbladization II: General Lindbladians, Linear Combinations, and Polynomials}

\author{Dhrumil Patel\thanks{Department of Computer Science, Cornell University, Ithaca, New York, 14850, USA, Email: djp265@cornell.edu}\and 
        Mark M. Wilde\thanks{School of Electrical and Computer Engineering, Cornell University, Ithaca, New York 14850, USA, Email: wilde@cornell.edu} }

\maketitle
\begin{abstract}
     In this paper, we investigate the problem of simulating open system dynamics governed by the well-known Lindblad master equation. 
     In our prequel paper, we introduced an input model in which Lindblad operators are encoded into pure quantum states, called program states, and we also introduced a method, called wave matrix Lindbladization, for simulating Lindbladian evolution by means of interacting the system of interest with these program states. Therein, we focused on a simple case in which the Lindbladian consists of only one Lindblad operator and a Hamiltonian. 
     Here, we extend the method to simulating general Lindbladians and other cases in which a Lindblad operator is expressed as a linear combination or a polynomial of the operators encoded into the program states.
     We propose quantum algorithms for all these cases and also investigate their sample complexity, i.e., the number of program states needed to simulate a given Lindbladian evolution approximately. 
     Finally, we demonstrate that our quantum algorithms provide an efficient route for simulating Lindbladian evolution relative to full tomography of encoded operators, by proving that the sample complexity for tomography is dependent on the dimension of the system, whereas the sample complexity of wave matrix Lindbladization is dimension independent.
\end{abstract}

\medskip
\begin{quote}
\textit{We dedicate our paper to the memory of G\"oran Lindblad (July~9, 1940--November~30, 2022), whose profound contributions to quantum information science, in the form of the Lindblad master equation~\cite{Lindblad1976OnSemigroups} and the data-processing inequality for quantum relative entropy \cite{Lin75}, will never be forgotten.}
\end{quote}
\vspace{.5in}

\tableofcontents

\section{Introduction}
Quantum simulation involves simulating or modeling the behavior of a complex quantum system using a quantum computer, allowing researchers to further investigate its properties in detail. 
The problem of simulating a closed quantum system, also known as Hamiltonian simulation, is well-studied, and hitherto, many quantum algorithms have been proposed to solve it \cite{Lloyd1996UniversalSimulators,Berry2013ExponentialHamiltonians,Berry2012Gate-efficientAlgorithms,Lloyd2014QuantumAnalysis,Berry2014SimulatingSeries,Berry2015HamiltonianParameters,Low2016OptimalProcessing,Kimmel2017HamiltonianComplexity}. 
However, many practical systems are not closed; rather, they interact with their environment, resulting in more complex dynamics that are well captured by the Lindblad master equation if the system under consideration is Markovian in nature \cite{Lindblad1976OnSemigroups, Gorini1976CompletelySystems}. 
This equation is critical to understanding the behavior of many open quantum systems~\cite{breuer2002theory,Weiss2021}, in condensed matter~\cite{Prosen2011OpenTransport,Manzano2011QuantumLaw,Olmos2012FacilitatedGlasses}, quantum chemistry~\cite{nitzan2006chemical,may2008charge},  quantum optics~\cite{Plenio1998,Gardiner2004QuantumBooks}, entanglement preparation~\cite{Kraus2008,Kastoryano2011,Reiter2016}, 
 thermal state preparation~\cite{Kastoryano2014QuantumCase}, quantum state engineering~\cite{verstraete2009quantum},
and the effects of noise on quantum computers~\cite{Magesan2012ModelingCircuits}.

In this paper, we address the problem of simulating Lindbladian evolution of a finite-dimensional quantum system. Starting from an initial state $\rho$, we aim to simulate its dynamics over a period of time $t\geq0$, as governed by the Lindblad master equation:
\begin{equation}
\label{eq:lindbladmaster}
    \frac{\partial \rho}{\partial t} =  \mathcal{L} (\rho) \coloneqq  -i [H, \rho] + \sum_{k=1}^{K} L_{k}\rho L_{k}^{\dagger} - \frac{1}{2} \left \{L_{k}^{\dagger}L_{k}, \rho \right\}.
\end{equation}
Here, the first term $-i [H, \rho]$ accounts for the unitary evolution of the system under the system Hamiltonian $H$. 
The second term $\sum_{k=1}^{K} L_{k}\rho L_{k}^{\dagger} - \frac{1}{2}  \{L_{k}^{\dagger}L_{k}, \rho \}$ captures non-unitary dynamics from interactions with the environment, described by Lindblad operators $\{L_{k}\}_{k=1}^K$.
These operators are not necessarily Hermitian and in fact have no constraints on them. 
The superoperator $\mathcal{L}$ in the equation above is called a Lindbladian, and the notation $\{A, B\}$ refers to the anti-commutator of operators $A$ and $B$, i.e., $\{A, B\} \coloneqq AB + BA$.

By simulating the evolution mentioned above for time $t$, we are referring to implementing its corresponding quantum channel~$\e^{\mathcal{L}t}$, which is the solution of \eqref{eq:lindbladmaster}, where
\begin{equation}
    \e^{\mathcal{L}t}(\rho) = \sum_{s=0}^\infty \frac{\mathcal{L}^s(\rho)t^s}{s!} ,
    \label{eq:Lind-expand}
\end{equation}
and $\mathcal{L}^s$ denotes $s$ sequential applications of the Lindbladian $\mathcal{L}$. 
The equality above simply comes from the Taylor series expansion of the exponential. 
For small $t$, we have the expansion $\e^{\mathcal{L}t}(\rho) = \rho +  \mathcal{L}(\rho) t + O(t^2)$, and we make use of it in what follows.

There has been an increasing interest in developing efficient quantum algorithms for simulating the dynamics of open quantum systems in recent years \cite{Childs2016EfficientDynamics, Cleve2016EfficientEvolution, KSMM22, Schlimgen2022QuantumOperators, Suri2022Two-UnitarySimulation} (see \cite{Miessen2022QuantumDynamics} for a review). 
These works are based on an assumption that a succinct representation of or black-box access to a set of Lindblad operators is provided beforehand. 
One such succinct representation is a list of non-zero coefficients when writing these operators as a linear combination of Pauli strings \cite{Cleve2016EfficientEvolution}.

As introduced in our prequel paper \cite{Patel2023}, we take a different approach to the above problem. 
Specifically, our approach differs from the methods mentioned above in how the Lindblad operators are provided as input. 
Let us suppose that each Lindblad operator $L_{k}$ is encoded into a pure state $|\psi_{k}\rangle$ in the following way:
\begin{equation}
    |\psi_{k}\rangle \coloneqq (L_{k}\otimes I) |\Gamma\rangle,
    \label{eq:L-encode}
\end{equation}
where
\begin{equation}
|\Gamma\rangle \coloneqq \sum_{j} |j\rangle |j\rangle    
\label{eq:max-ent-vec-def}
\end{equation}
is a maximally entangled vector. Suppose further that we have access to multiple copies of this state.
We refer to such a state as a program state because it can encode any unit-norm linear operator.
By unit-norm, we mean $\left\Vert L_{k}\right\Vert_{2} = 1$, where $\left\Vert A \right\Vert_{2} \coloneqq \sqrt{\operatorname{Tr}[A^{\dagger}A]}$ is the {Schatten-2} norm of a matrix $A$ (also known as the Hilbert--Schmidt norm).
This constraint on~$L_{k}$ is due to the requirement that $|\psi_{k}\rangle$ should be a quantum state (i.e., normalized). 
In light of this, when encoding a Lindblad operator $L'$ with an arbitrary norm, we suppose that its normalized version, i.e., $L'/\left\Vert L' \right\Vert_2$, is encoded in a state. 
Furthermore, for the Hamiltonian term in \eqref{eq:lindbladmaster}, we follow the input model of~\cite{Kimmel2017HamiltonianComplexity} and suppose that the Hamiltonian $H$ is given as a linear combination of program states $\{\sigma_{j}\}_{j=1}^{J}$, i.e., 
\begin{align}
    H  & \coloneqq \sum_{j=1}^{J}c_{j}\sigma_{j},
\end{align}
where $c_{j}\in\mathbb{R}$, and that we have access to multiple copies of $\sigma_{j}$ for all $j \in \{1,\ldots, J\}$.

In essence, we are interested in answering the following question: given one copy of an unknown quantum state $\rho$, some number $n_{j}$ copies of the program state~$\sigma_{j}$ for all $j \in \{1,\ldots, J\}$, and $m_{k}$ of copies of the program state~$\psi_{k} \coloneqq |\psi_{k}\rangle\! \langle\psi_{k}|$ for all $k \in \{1,\ldots, K\}$, can we approximately simulate the quantum channel~$\e^{\mathcal{L}t}$ up to an error $\varepsilon$? In other words,
can we realize the following transformation?
\begin{equation}\label{eq:task-map}
    \rho \otimes \sigma_{1}^{\otimes n_{1}} \otimes \cdots \otimes \sigma_{J}^{\otimes n_{J}} \otimes \psi_{1}^{\otimes m_{1}} \otimes \cdots \otimes \psi_{K}^{\otimes m_{K}} \overset{\overset{\varepsilon}{\approx}}{\longrightarrow} \e^{\mathcal{L}t}(\rho).
\end{equation}
As done in our prequel paper \cite{Patel2023}, we call this modified problem sample-based  Lindbladian simulation, and it can be seen as a natural analogue to sample-based Hamiltonian simulation \cite{Lloyd2014QuantumAnalysis,Kimmel2017HamiltonianComplexity}. In our prequel paper, we introduced an approach for solving a simpler version of this problem, called  
\textit{Wave Matrix Lindbladization} (WML), which is an analogue to  \textit{Density Matrix Exponentiation} \cite{Lloyd2014QuantumAnalysis}. That is,
in~\cite{Patel2023}, we focused on a relatively simple case in which the Lindbladian consists of only one Lindblad operator and a Hamiltonian term. 
We did this so that the reader could grasp the intuition behind the techniques introduced. 
That being said, in this sequel paper, we extend wave matrix Lindbladization to general Lindbladians and beyond.

\subsection{Summary of Main Results}

In this paper, we present two quantum algorithms that implement the quantum channel $\e^{\mathcal{L}t}$ with some desired accuracy $\varepsilon \in (0,1)$. We call these the sampling-based approach and the Trotter-like approach.
We then investigate the sample complexity of these algorithms (see Section~\ref{sec:gen-lindbladians}). 
By sample complexity, we mean the number of copies of the program states $\{\sigma_{j}\}_{j=1}^J$ and~$\{\psi_{k}\}_{k=1}^K$ needed to achieve the above task; i.e., this number is equal to~$\sum_{j=1}^J n_{j} + \sum_{k=1}^K m_{k}$.

Furthermore, we provide another, different extension of the method introduced in \cite{Patel2023}.
Previously, we considered a Lindbladian with only one Lindblad operator and a Hamiltonian, and this entire Lindblad operator was given encoded in a single program state.
We extend this simple case to that in which this Lindblad operator is provided as a linear combination or a polynomial of the operators encoded in the program states.
We propose quantum algorithms for these cases as well under the assumption that we have access to unitaries that prepare these program states (see Sections~\ref{sec:single-op-lin-comb} and~\ref{sec:single-op-poly}).
We then investigate the query complexity of these algorithms, i.e., the number of times these state-preparation unitaries are queried to approximate a given evolution. 

Finally, we show that our quantum algorithms offer an efficient approach for simulating Lindbladian evolution compared to full tomography of the encoded operators. 
We prove in Section~\ref{sec:WML-vs-WMT}~that the sample complexity of the tomography approach scales with the dimension of the system being simulated, while that of our algorithms does not.

\textbf{Key Idea ---} Prior to delving into the specifics of each algorithm, we first provide a high-level overview of the first approach that leverages sampling to implement the channel $\e^{\mathcal{L}t}$ (i.e., sampling-based approach).  We note here that this approach builds upon that outlined in the proof of \cite[Theorem~9]{Kimmel2017HamiltonianComplexity} and is conceptually similar to random compiler approaches in Hamiltonian simulation \cite{campbell2019}.
This sampling-based algorithm primarily consists of three steps that we repeat $n \coloneqq \sum_{j=1}^J n_{j} + \sum_{k=1}^K m_{k}$ times. 
Recall that $\rho$ is the initial state on which we would like to simulate this channel.
\begin{enumerate}
\item The first step is a sampling step, where we sample
\begin{itemize}
    \item $\sigma_{j} \otimes \tau \otimes |0\rangle\!\langle0|$ with probability $c_{j}/c$ for $c_{j} > 0$,
    \item $\sigma_{j} \otimes \tau \otimes |1\rangle\!\langle1|$ with probability $(-c_{j})/c$ for $c_{j} < 0$,
    \item $\psi_{k} \otimes |2\rangle\!\langle2|$ with probability $\left\Vert L_{k}\right\Vert_{2}^{2}/c$,
\end{itemize} 
with $\tau$ an arbitrary state.
Here, we define the normalization constant~$c$ as
\begin{equation}
c\coloneqq \sum_{j=1}^{J}\left\vert c_{j}\right\vert +\sum_{k=1}^{K}\left\Vert
L_{k}\right\Vert_{2}^{2}.
\end{equation}

\item Suppose that $\rho$ is in register 1 and the sampled state obtained from the above step is in registers 2, 3, and 4. Set $\Delta = ct/n$.
The second step of our algorithm involves applying a quantum channel $\mathcal{E}$ on all registers conditioned on the state of register 4:
\begin{itemize}
    \item If $|0\rangle\!\langle0|$, then apply a unitary $\e^{-i\mathsf{SWAP} \Delta}$ on registers 1 and 2, where   $\mathsf{SWAP}$ is the swap operator (defined explicitly in~\eqref{eq:swap-def}).
    \item if $|1\rangle\!\langle1|$, then apply a unitary $\e^{i\mathsf{SWAP} \Delta}$ on registers 1 and 2.
    \item if $|2\rangle\!\langle2|$, then apply a channel $\e^{\mathcal{M}\Delta}$ on registers 1, 2, and 3, where $\mathcal{M}$ is the Lindbladian:
    \begin{equation}\label{eq:intro-M}
        \mathcal{M} (\cdot) \coloneqq M(\cdot) M^{\dagger} - \frac{1}{2} \left \{M^{\dagger}M, (\cdot) \right\},
    \end{equation}
    with Lindblad operator
    \begin{equation}
        M \coloneqq \frac{1}{\sqrt{d}}\left(I_1\otimes |\Gamma\rangle\! \langle \Gamma |_{23}\right) \left( \mathsf{SWAP}_{12} \otimes I_3\right),
        \label{eq:M-op-initial}
\end{equation}    
where we recall \eqref{eq:max-ent-vec-def} for the definition of $|\Gamma\rangle$.
\end{itemize}
\item Finally, in the third step, we trace out registers 2, 3, and 4. 
\end{enumerate}

\begin{figure}
    \centering
    \includegraphics[scale=0.19]{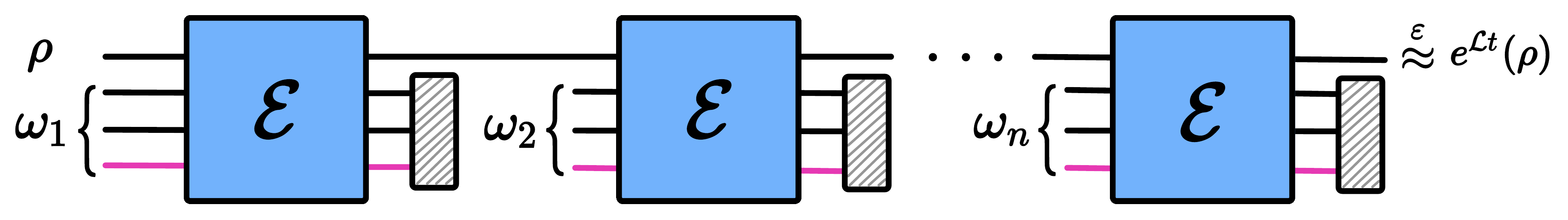}
    \caption{Our quantum algorithm, repeated $n = O(c^2 t^2/\varepsilon)$ times, for approximating the target quantum channel, i.e., $\e^{\mathcal{L}t}$, with an approximation error of~$\varepsilon$. Each hatched rectangle represents the trace-out operation. At the very left, the state $\rho$ is in register~1, and the state $\omega_{i}$, which is obtained after the sampling step of each iteration of our algorithm, is in registers~2, 3 and 4. The quantum channel $\mathcal{E}$ is conditioned on the state of the register~4 (pink line). }
    \label{fig:n-times}
\end{figure}

Now, the key step here is to prove the following equality:
\begin{align}
    \operatorname{Tr}_{234}\left[ \mathcal{E}(\rho \otimes \omega) \right] & = \rho + \frac{1}{c} \mathcal{L}(\rho) \Delta + O(\Delta^2)\\
& = \e^{(\mathcal{L}/c) \Delta}(\rho) + O(\Delta^2),
\end{align}
where $\omega$ is the state resulting from the aforementioned sampling processing (formally defined in \eqref{eq:omega-mul-op-plus-H} later on) and the Lindbladian $\mathcal{L}$ is defined in \eqref{eq:lindbladmaster}.
If the above equality holds and we repeat the aforementioned three steps of our algorithm $n = O(c^2t^2/\varepsilon)$ times, then we can approximate the target channel~$\e^{\mathcal{L}t}$ with an approximation error of $O(\varepsilon)$. 
We prove this statement in detail in Theorem~\ref{thm:mul-op-plus-H}.
This further shows that the algorithm above uses
\begin{equation}
n_j = |c_{j}|/c \cdot O(c^2t^2/\varepsilon) = O(|c_{j}|ct^2/\varepsilon)    
\end{equation}
copies of the program state $\sigma_{j}$ on average and uses
\begin{equation}
m_k = \left\Vert
L_{k}\right\Vert _{2}^{2}/c \cdot O(c^2t^2/\varepsilon) = O(\left\Vert L_{k}\right\Vert _{2}^{2} ct^2/\varepsilon)    
\end{equation}
copies of the program state~$\psi_{k}$ on average.

\subsection{Notation}

We employ the same notation used in our prequel paper \cite{Patel2023} but recall it here for convenience. Let $\mathcal{H}_{S}$ denote a $d$-dimensional Hilbert space associated with a quantum system $S$. 
We denote the set of quantum states acting on~$\mathcal{H}_{S}$ by~$\mathcal{D}(\mathcal{H}_{S})$. 
Let $\operatorname{Tr}[X]$ denote the trace of a matrix $X$, i.e., the sum of its diagonal elements. 
Also, let $X^{\dagger}$ denote the Hermitian conjugate (or adjoint) of the matrix $X$.
The Schatten $p$-norm of a matrix $X$ is defined for $p \in [1, \infty)$ as follows:
\begin{equation}
    \left\Vert X\right\Vert_{p} \coloneqq \left(\operatorname{Tr}\!\left[\left(X^{\dagger}X\right)^{\frac{p}{2}}\right]\right)^{\frac{1}{p}}.\label{eq:schatten-norm-def}
\end{equation}
For the purpose of this paper, we use Schatten norms with $p =1$ (also called trace norm) and $p=2$ (Hilbert--Schmidt norm). Furthermore, let  $[X, Y] \coloneqq XY - YX$ and $\{X, Y\} \coloneqq XY + YX$ denote the commutator and anti-commutator of the operators $X$ and~$Y$, respectively. 

The diamond distance between two quantum channels $\mathcal{N}$ and $\mathcal{M}$ is defined as follows \cite{Kitaev1997QuantumCorrection}:
\begin{equation}
     \left\Vert \mathcal{N} - \mathcal{M} \right\Vert_{\diamond} \coloneqq \sup_{\rho \in \mathcal{D}(\mathcal{H}_{R} \otimes \mathcal{H}_{S})}  \left \Vert (\mathcal{I}_{R} \otimes \mathcal{N}  )(\rho) - (\mathcal{I}_{R} \otimes \mathcal{M}  )(\rho)  \right \Vert_{1},
\end{equation}
where $R$ is a reference system and $\mathcal{I}_{R}$ is the identity channel acting on the system $R$. An important point to note here is that, in the above definition, the dimension of $R$ is arbitrarily large. However, it is known that it suffices to perform the optimization over pure bipartite states with the dimension of~$R$ equal to the dimension of $S$. Furthermore, the quantity in the objective function of the above optimization is the trace distance, defined as $\left \Vert \rho -\sigma \right\Vert_{1}$ for two quantum states $\rho, \sigma \in \mathcal{D}(\mathcal{H_{S}})$. In what follows, we employ the normalized diamond distance $\frac{1}{2} \left\Vert \mathcal{N} - \mathcal{M} \right\Vert_{\diamond}$ to measure approximation error---the normalization factor of $\frac{1}{2}$ guarantees that $\frac{1}{2} \left\Vert \mathcal{N} - \mathcal{M} \right\Vert_{\diamond} \in [0,1]$ for quantum channels $\mathcal{N}$ and $\mathcal{M}$.

\begin{figure}
    \centering
    \includegraphics[scale=0.31]{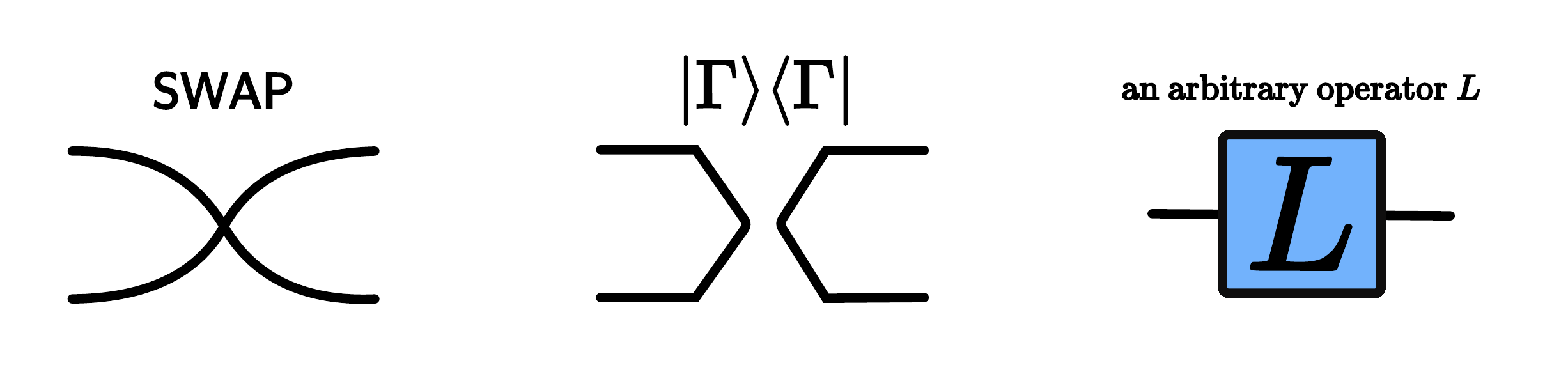}
    \caption{Tensor-network diagrams of operators $\mathsf{SWAP}$, $|\Gamma\rangle\!\langle\Gamma|$, and $L$.}
    \label{fig:lindblad-symb}
\end{figure}

For $\rho$ a state in $\mathcal{D}(\mathcal{H}_{R} \otimes \mathcal{H}_{S})$, we denote the partial trace over the Hilbert space $\mathcal{H}_{R}$ by $\operatorname{Tr}_{R}[\rho]$. We also sometimes use a different notation for partial trace; i.e., given a multipartite state $\rho$, we use the notation $\operatorname{Tr}_{k}[\rho]$ to denote the action of tracing out the $k^{\text{th}}$ party. Furthermore, we define the maximally entangled vector in $\mathcal{H}_{R} \otimes \mathcal{H}_{S}$ as
\begin{equation}
|\Gamma\rangle_{RS} \coloneqq \sum_{i} |i\rangle_{R} |i\rangle_{S},   
\label{eq:max-ent-vec}
\end{equation}
and we denote its normalized version, i.e., the maximally entangled state, by~$|\Phi\rangle$.
 We also define the unitary swap operation and unitary cyclic permutation operation (cyclic swap) in the following way:
\begin{align}
    \mathsf{SWAP} & \coloneqq \sum_{i, j} |i\rangle\! \langle j| \otimes |j\rangle\! \langle i|,
    \label{eq:swap-def}\\
    \mathsf{CYCSWAP} &\coloneqq \sum_{i_{1},i_{2}, \ldots, i_{n} } |i_{n}\rangle\!\langle i_{1}| \otimes |i_{1}\rangle\!\langle i_{2}| \otimes \cdots \otimes |i_{n-1}\rangle\!\langle i_{n}|,
    \label{eq:cyc-swap-def}
\end{align}
where $n$ is the number of systems on which the cyclic swap is performed.

In our paper, we make extensive use of tensor-network diagrams. Figure~\ref{fig:lindblad-symb} depicts tensor-network diagrams for some basic operators defined above, such as $\mathsf{SWAP}$ and $|\Gamma\rangle\!\langle\Gamma|$. For more background on tensor-network diagrams, please refer to \cite{Biamonte2017TensorNutshell}. Throughout this paper, we sometimes suppress system labels for ease of notation; however, they should be clear from the context.

\section{Quantum Algorithms for Simulating Markovian Dynamics}

\subsection{Simulating General Lindbladians}\label{sec:gen-lindbladians}

In what follows, we consider the case in which the Lindbladian consists of multiple Lindblad operators, as well as a Hamiltonian. 
To be more precise, we are interested in simulating the action of the dynamics specified by \eqref{eq:lindbladmaster} and the following Lindbladian:
\begin{equation}\label{eq:lindbladmaster-gen}
    \mathcal{L} (\rho) \coloneqq  -i [H, \rho] + \sum_{k=1}^{K} L_{k}\rho L_{k}^{\dagger} - \frac{1}{2} \left \{L_{k}^{\dagger}L_{k}, \rho \right\},
\end{equation}
on a quantum state~$\rho$ for time $t$.
Here, for the Hamiltonian term, we follow the input model of \cite{Kimmel2017HamiltonianComplexity}, assuming that $H$ is provided as a linear combination of program states $\{\sigma_{j}\}_{j=1}^{J}$, i.e.,
\begin{align}
    H  & \coloneqq \sum_{j=1}^{J}c_{j}\sigma_{j},
\end{align}
where $c_{j}\in\mathbb{R}$.
Additionally, let us suppose that each Lindblad operator~$L_{k}$ is encoded in a pure state $\psi_{k}$ as described in~\eqref{eq:L-encode}.
While $L_{k}$ can have an arbitrary Hilbert–Schmidt norm in general, we can only encode its normalized version~$L_{k}/\left\Vert L_{k}\right\Vert_{2}$ into a pure state, so that we set
\begin{equation}
    |\psi_k\rangle \coloneqq \frac{(L_k \otimes I)|\Gamma\rangle}{\left \Vert L_k \right\Vert_2},
    \label{eq:psi_k-L_k-def}
\end{equation}
which implies that $\left\Vert L_{k}\right\Vert_{2}^{2} \psi_k = \left\Vert L_{k}\right\Vert_{2}^{2} |\psi_k\rangle\!\langle \psi_k| =  (L_k \otimes I)|\Gamma\rangle\!\langle\Gamma |(L_k \otimes I)^\dag   $.

\subsubsection{Sampling-based Approach}

Given access to $\{(c_{j}, \sigma_{j})\}_{j=1}^{J}$ and $\{\left(\Vert L_{k}\right\Vert^2_{2}, \psi_{k})\}_{k=1}^{K}$, we prepare the following state:
\begin{equation}
    \omega  \coloneqq \frac{1}{c} \sum_{x=0}^2  \omega^{(x)} \otimes |x\rangle\!\langle x|,
    \label{eq:omega-mul-op-plus-H}
\end{equation}
where
\begin{align}
    \omega^{(0)} & \coloneqq \sum_{j:c_{j}>0}c
_{j}\sigma_{j}\otimes\tau ,\\
    \omega^{(1)} & \coloneqq \sum_{j:c_{j}%
<0}(-c_{j})\sigma_{j}\otimes\tau, \\
\omega^{(2)} & \coloneqq \sum_{k=1}%
^{K}\left\Vert L_{k}\right\Vert _{2}^{2}\psi_{k},
\end{align}
% \begin{multline}
%         \omega  \coloneqq \frac{1}{c}\Bigg[  \underbrace{\sum_{j:c_{j}>0}c
% _{j}\sigma_{j}\otimes\tau}_{\eqqcolon~\omega^{(0)}}  \otimes |0\rangle\!\langle0| + \underbrace{\sum_{j:c_{j}%
% <0}(-c_{j})\sigma_{j}\otimes\tau}_{\eqqcolon~\omega^{(1)}}  \otimes |1\rangle\!\langle1|+ \\
% \underbrace{\sum_{k=1}%
% ^{K}\left\Vert L_{k}\right\Vert _{2}^{2}\psi_{k}}_{\eqqcolon~\omega^{(2)}}\otimes |2\rangle\!\langle2|\Bigg],\label{eq:omega-mul-op-plus-H}
% \end{multline}
$\tau$ is an arbitrary state, and 
\begin{equation}
c\coloneqq \sum_{j=1}^{J}\left\vert c_{j}\right\vert +\sum_{k=1}^{K}\left\Vert
L_{k}\right\Vert_{2}^{2}.
\label{eq:normalization-c-gen}
\end{equation}
Specifically, we prepare $\omega$ by sampling the state $\sigma_{j} \otimes \tau \otimes |0\rangle\!\langle0|$ with probability $c_{j}/c$ when $c_{j} > 0$, sampling the state $\sigma_{j} \otimes \tau \otimes |1\rangle\!\langle1|$ with probability $(-c_{j})/c$ when $c_{j} < 0$, and sampling the state $\psi_{k} \otimes |2\rangle\!\langle2|$ with probability $\left\Vert L_{k}\right\Vert_{2}^{2}/c$.

We now present a quantum algorithm for simulating the channel $\e^{\mathcal{L}t}$ corresponding to the Lindbladian in \eqref{eq:lindbladmaster-gen}, up to error $\varepsilon$ in diamond distance, using $n$ copies of the state $\omega$. 
Given that we are interested in providing an analysis related to the diamond distance, let us suppose that the channel acts on one system~$S$ of a joint state of two systems.
So, let $\rho \in \mathcal{D}(\mathcal{H}_{R} \otimes \mathcal{H}_{S})$ be an unknown quantum state given as input over the joint system $RS$, where $R$ is a reference system. 
Also, let the $i^{\text{th}}$ copy of the state~$\omega$ be a quantum state of the joint system $P_{i}Q_{i}C_{i}$, with the system ordering specified as $\omega_{PQC} = \frac{1}{c} \sum_{x=0}^2 \omega^{(x)}_{PQ} \otimes |x\rangle\!\langle x|_C$. 
For brevity, let us use $(PQC)_i$ as a shorthand for $P_{i}Q_{i}C_{i}$.

\textbf{Algorithm 1 --} Set $n\in \mathbb{N}$, with a particular choice specified later. Set $i=1$. Given the $i^{\text{th}}$ copy of $\omega$, i.e., $\omega_{(PQC)_{i}}$, perform the following three steps:
\begin{enumerate} 
    \item Adjoin the state $\omega_{(PQC)_{i}}$ to the input state $\rho$.
    \item Let $\mathcal{E}$ be the following channel that acts nontrivially on systems $S(PQC)_i$:
    \begin{equation}
        \mathcal{E}(\cdot) \coloneqq \sum_{x=0}^2 \mathcal{E}^{(x)}_{S(PQ)_i}(\langle x|_{C_i} (\cdot) |x\rangle_{C_i}),
        \label{eq:E-channel-alg-1}
    \end{equation}
    where each channel $\mathcal{E}^{(x)}$ is defined for $x \in \{0,1,2\}$ as
    \begin{align}
        \mathcal{E}^{(0)}(\cdot) & \coloneqq \e^{-i\mathsf{SWAP}_{SP_{i}}\Delta}(\cdot) \e^{i\mathsf{SWAP}_{SP_{i}}\Delta}, \\
        \mathcal{E}^{(1)}(\cdot) & \coloneqq \e^{i\mathsf{SWAP}_{SP_{i}}\Delta}(\cdot) \e^{-i\mathsf{SWAP}_{SP_{i}}\Delta}, \\
        \mathcal{E}^{(2)}(\cdot) & \coloneqq \e^{\mathcal{M}\Delta}(\cdot).
    \end{align}
    In the above, $\Delta$ is some small duration of time; specifically, we set~$\Delta \coloneqq ct/n$.
    Intuitively, the channel $\mathcal{E}$ first measures the classical register~$C_i$, and if outcome $x$ is obtained, it discards $C_i$ and applies the channel~$\mathcal{E}^{(x)}$ to registers $S(PQ)_i$. Here and in what follows, we apply the identity operator $I$ on all systems that are not explicitly mentioned.  
    Moreover, $\e^{\mathcal{M}\Delta}$ is the quantum channel with the Lindbladian $\mathcal{M}$ defined as
\begin{equation}\label{eq:general-fixed-M-master}
        \mathcal{M}\left (\cdot\right) \coloneqq M(\cdot) M^{\dagger} - \frac{1}{2} \left \{M^{\dagger}M, (\cdot) \right\}.
    \end{equation}
    In the above, the Lindblad operator $M$ acts on the joint system $S(PQ)_{i}$, and we define it as
    \begin{equation}
        M \coloneqq \frac{1}{\sqrt{d}} \left(I_S \otimes |\Gamma\rangle\!\langle\Gamma |_{(PQ)_{i}}\right) \left( \mathsf{SWAP}_{SP_{i}} \otimes I_{Q_i}\right).
    \end{equation}

    After applying the channel $\mathcal{E}$ defined in \eqref{eq:E-channel-alg-1} to the joint state $\rho_{RS}\otimes \omega_{(PQC)_{i}}$, the resulting state is as follows:
    \begin{multline}
        \mathcal{E}\left (\rho_{RS}\otimes \omega_{(PQC)_{i}} \right ) \\
        \coloneqq \frac{1}{c}\Bigg[  \left (\e^{-i\mathsf{SWAP}_{SP_{i}}\Delta} \otimes I \right) \left (\rho_{RS}\otimes \omega^{(0)}_{(PQ)_{i}} \right)\left (\e^{i\mathsf{SWAP}_{SP_{i}}\Delta} \otimes I \right) \\
+ \left (\e^{i\mathsf{SWAP}_{SP_{i}}\Delta} \otimes I \right) \left (\rho_{RS}\otimes \omega^{(1)}_{(PQ)_{i}} \right)\left (\e^{-i\mathsf{SWAP}_{SP_{i}}\Delta} \otimes I \right)  \\
+ \e^{\mathcal{M}\Delta}\left (\rho_{RS}\otimes \omega^{(2)}_{(PQ)_{i}} \right) \Bigg].\label{eq:general-cond-channel}
    \end{multline}

    \item Trace out the systems $(PQ)_{i}$.
\end{enumerate}
We repeat the above procedure for each copy of $\omega$, i.e., for all $i \in \{1,\ldots, n\}$.

Theorem~\ref{thm:mul-op-plus-H} below states that the algorithm mentioned above employs $n = O(c^2t^2/\varepsilon)$ copies of $\omega$ for simulating the Lindbladian evolution of $\rho_{RS}$ for time~$t$ according to the Lindbladian in \eqref{eq:lindbladmaster-gen}. 
This results in the normalized trace distance between the final state and the desired target state $\left(\mathcal{I}_{R} \otimes \e^{\mathcal{L}t} \right)\left(\rho_{RS}\right)$ being no larger than $\varepsilon$, for an arbitrary input state $\rho_{RS}$. 
The aforementioned bound on $n$ further implies that on average, the above algorithm uses 
\begin{equation}
    n_{j} = |c_{j}|/c \cdot O(c^2t^2/\varepsilon) = O(|c_{j}|ct^2/\varepsilon)
\end{equation}
copies of the program state $\sigma_{j}$ and 
\begin{equation}
    m_{k} = \left\Vert
L_{k}\right\Vert _{2}^{2}/c \cdot O(c^2t^2/\varepsilon) = O(\left\Vert L_{k}\right\Vert _{2}^{2} ct^2/\varepsilon)
\end{equation}
copies of the program state $\psi_{k}$.

\begin{theorem}\label{thm:mul-op-plus-H}
 Given access to $n$ copies of the state $\omega \in  \mathcal{D}\!\left(\mathcal{H}_{PQC}\right)$, which is defined in~\eqref{eq:omega-mul-op-plus-H} and encodes the Lindblad operators $L_{1}, \ldots, L_{K}$ and the Hamiltonian $H$, there exists a quantum algorithm $\mathcal{A}$ such that the following holds:
    \begin{equation}
        \frac{1}{2}\left \Vert \e^{\mathcal{L}t} - \mathcal{A} \right \Vert_{\diamond} \leq \varepsilon,
    \end{equation}
    with only $n = O(c^2t^2/\varepsilon)$ copies of $\omega$, where $c$ is defined in \eqref{eq:normalization-c-gen}.
    In other words, $\mathcal{A}$ uses only $n = O(c^2t^2/\varepsilon)$ copies of $\omega$ to approximate the channel~$ \e^{\mathcal{L}t}$, defined from  \eqref{eq:Lind-expand} and \eqref{eq:lindbladmaster-gen}, up to $\varepsilon$ error in diamond distance. 
\end{theorem}

\begin{proof}
For clarity and simplicity, we omit system labels, and let us suppose that the input state~$\rho$ does not have a reference system $R$.

We first expand the target state $\e^{\mathcal{L}t}(\rho)$ at $t=0$ using its Taylor series, as in \eqref{eq:Lind-expand}:
\begin{equation}\label{eq:general-desired-state}
    e^{\mathcal{L}t}(\rho) = \rho + \mathcal{L}(\rho) t +  \frac{1}{2} (\mathcal{L}\circ  \mathcal{L})(\rho) t^2 + \cdots \, .
\end{equation}
In the first step of Algorithm~1, we apply the channel $\mathcal{E}$ defined in \eqref{eq:E-channel-alg-1} to~$\rho \otimes \omega$, and then trace out the systems corresponding to~$\omega$. 
The output state obtained after this step is given by tracing out systems 2 and 3 of the state in \eqref{eq:general-cond-channel}:
\begin{align}
& \operatorname{Tr}_{23}\!\left [ \mathcal{E}(\rho \otimes \omega)\right] \notag \\
& = \operatorname{Tr}_{23}\Bigg [ \frac{1}{c}\Bigg[  \left (e^{-i\mathsf{SWAP} \Delta} \otimes I \right) \left (\rho \otimes \omega^{(0)}\right)\left (e^{i\mathsf{SWAP} \Delta} \otimes I \right)  \notag \\
& \hspace{3cm}+ \left (e^{i\mathsf{SWAP}\Delta} \otimes I \right) \left (\rho \otimes \omega^{(1)} \right)\left (e^{-i\mathsf{SWAP} \Delta} \otimes I \right)  \notag\\
& \hspace{7cm} + e^{\mathcal{M}\Delta}\left (\rho \otimes \omega^{(2)} \right) \Bigg] \Bigg]\\
& = \frac{1}{c} \operatorname{Tr}_{23} \Bigg[  \left (e^{-i\mathsf{SWAP} \Delta} \otimes I \right) \left (\rho \otimes \left (\sum_{j:c_{j}>0}c
_{j}\sigma_{j}\otimes\tau \right) \right)\left (e^{i\mathsf{SWAP} \Delta} \otimes I \right)  \notag \\
& \hspace{1cm}+ \left (e^{i\mathsf{SWAP}\Delta} \otimes I \right) \left (\rho \otimes \left (\sum_{j:c_{j}<0} (-c_{j})\sigma_{j}\otimes\tau \right) \right)\left (e^{-i\mathsf{SWAP} \Delta} \otimes I \right)  \notag\\
& \hspace{5cm} + e^{\mathcal{M}\Delta}\left (\rho \otimes \left ( \sum_{k=1}%
^{K}\left\Vert L_{k}\right\Vert _{2}^{2}\psi_{k} \right) \right) \Bigg] \\
& = \frac{1}{c}\Bigg[ \sum_{j:c_{j}>0} c_{j}\operatorname{Tr}_{23}\!\left [  \left (e^{-i\mathsf{SWAP} \Delta} \otimes I \right) \left (\rho \otimes \sigma_{j}\otimes\tau \right)\left (e^{i\mathsf{SWAP} \Delta} \otimes I \right)  \right] \notag \\
& \hspace{1cm}+ \sum_{j:c_{j}<0} (-c_{j})\operatorname{Tr}_{23}\!\left [ \left (e^{i\mathsf{SWAP}\Delta} \otimes I \right) \left (\rho \otimes \sigma_{j}\otimes\tau \right)\left (e^{-i\mathsf{SWAP} \Delta} \otimes I \right)  \right ] \notag\\
& \hspace{5cm} + \sum_{k=1}%
^{K} \left\Vert L_{k}\right\Vert _{2}^{2} \operatorname{Tr}_{23}\!\left [e^{\mathcal{M}\Delta}\left (\rho \otimes \psi_{k}\right) \Bigg] \right]\\
& = \frac{1}{c}\Bigg[ \sum_{j:c_{j}>0} c_{j}\operatorname{Tr}_{2}\left [  e^{-i\mathsf{SWAP} \Delta} \left (\rho \otimes \sigma_{j} \right)e^{i\mathsf{SWAP} \Delta}\right] \notag \\
& \hspace{2cm}+ \sum_{j:c_{j}<0} (-c_{j})\operatorname{Tr}_{2}\left [ e^{i\mathsf{SWAP}\Delta}  \left (\rho \otimes \sigma_{j} \right)e^{-i\mathsf{SWAP} \Delta} \right ] \notag\\
& \hspace{5cm} + \sum_{k=1}%
^{K} \left\Vert L_{k}\right\Vert _{2}^{2} \operatorname{Tr}_{23}\!\left [e^{\mathcal{M}\Delta}\left (\rho \otimes \psi_{k}\right) \right] \Bigg]\\
& = \frac{1}{c}\Bigg[ \sum_{j:c_{j}>0} c_{j} \left(\rho - i \operatorname{Tr}_{2}\left [ \left [\mathsf{SWAP}, \left (\rho \otimes \sigma_{j} \right) \right]\right] \Delta  + O(\Delta^2) \right) \notag \\
& \hspace{2cm}+ \sum_{j:c_{j}<0} (-c_{j})\left(\rho + i \operatorname{Tr}_{2}\left [ \left [\mathsf{SWAP}, \left (\rho \otimes \sigma_{j} \right) \right]\right] \Delta  + O(\Delta^2) \right) \notag\\
& \hspace{3.5cm} + \sum_{k=1}%
^{K} \left\Vert L_{k}\right\Vert _{2}^{2} \left ( \rho + \operatorname{Tr}_{23}\!\left [\mathcal{M}\left (\rho \otimes \psi_{k}\right) \right] \Delta + O(\Delta^2)\right)\Bigg]\tag{Taylor Expansion} \\
& = \rho + \frac{1}{c}\Bigg[ \sum_{j:c_{j}>0} c_{j} \left(- i \operatorname{Tr}_{2}\left [ \left [\mathsf{SWAP}, \left (\rho \otimes \sigma_{j} \right) \right]\right]\right) \notag \\
& \hspace{2cm}+ \sum_{j:c_{j}<0} c_{j}\left(-i \operatorname{Tr}_{2}\left [ \left [\mathsf{SWAP}, \left (\rho \otimes \sigma_{j} \right) \right]\right]\right) \notag\\
& \hspace{3.5cm} + \sum_{k=1}%
^{K} \left\Vert L_{k}\right\Vert _{2}^{2} \operatorname{Tr}_{23}\!\left [\mathcal{M}\left (\rho \otimes \psi_{k}\right) \right] \Bigg] \Delta + O(\Delta^2).
\end{align}
Next, we expand the Lindbladian $\mathcal{M}$ using its definition given by \eqref{eq:general-fixed-M-master} and then employ tensor-network diagrams for further simplifying all the partial-trace terms above. 
Please refer to Figures~3, 5, 6, and 7 of \cite{Patel2023} to simplify these terms (or alternatively, Appendix~A of \cite{Patel2023}) and obtain the following:
\begin{align}
    & \rho + \frac{1}{c}\Bigg[ \sum_{j:c_{j}>0} c_{j} \left(- i  \left [\sigma_{j}, \rho \right] \right) + \sum_{j:c_{j}<0} c_{j}\left(-i \left [\sigma_{j}, \rho \right]\right) \notag\\
& \hspace{1cm} + \sum_{k=1}%
^{K} \left\Vert L_{k}\right\Vert_{2}^{2} \left( \frac{1}{\left\Vert L_{k}\right\Vert_{2}^{2}} \right) \left(L_{k}\rho L_{k}^{\dagger} - \frac{1}{2} \left \{L_{k}^{\dagger}L_{k}, \rho \right\} \right) \Bigg] \Delta + O(\Delta^2)\\
& = \rho + \frac{1}{c}\Bigg( - i  \left [H, \rho \right]  + \sum_{k=1}%
^{K} L_{k}\rho L_{k}^{\dagger} - \frac{1}{2} \left \{L_{k}^{\dagger}L_{k}, \rho \right\}  \Bigg) \Delta + O(\Delta^2)\\
& = \rho + \frac{1}{c} \mathcal{L}(\rho) \Delta + O(\Delta^2)\\
& = e^{(\mathcal{L}/c) \Delta}(\rho) + O(\Delta^2).
\end{align}
The last equality follows from \eqref{eq:general-desired-state}. Substituting $\Delta = c t/n$ and repeating Algorithm~1 for $n = O(c^2 t^2 / \varepsilon)$ times produces a quantum state that is $\varepsilon$-close to the ideal target state $e^{\mathcal{L}t}(\rho)$ in normalized trace distance. 
A detailed error analysis of this claim is similar to that provided in Appendix~B of~\cite{Patel2023}.
\end{proof}

\subsubsection{Trotter-like Approach}

Here, we present another quantum algorithm for simulating the quantum channel $e^{\mathcal{L}t}$ with the Lindbladian $\mathcal{L}$ given by \eqref{eq:lindbladmaster-gen}. 
Before delving into the specifics of this algorithm, we first provide a brief overview of it in order to gain an intuition behind its inner workings.  
Conceptually, one can think of the Lindbladian $\mathcal{L}$ as a linear combination of single-operator Lindbladians~$\mathcal{L}_{1}, \ldots, \mathcal{L}_{r}$, i.e., $\sum_{i=1}^{r}\mathcal{L}_{i}$. 
By single-operator, we mean that each $\mathcal{L}_{i}$ consists of only one term, and this term can be either a Hamiltonian term or a non-Hamiltonian term with a single Lindblad operator. 
Now, if the quantum channels $e^{\mathcal{L}_1 t}, \ldots, e^{\mathcal{L}_r t}$ associated with the single-operator Lindbladians can be simulated efficiently, then the intuition behind this approach is to approximate $e^{\mathcal{L}t}$ by sequentially applying these simple channels for short time steps and repeating this composite sequence multiple times. 
Due to its similarity with Trotter-based approaches in Hamiltonian \cite{Lloyd1996UniversalSimulators} and Lindbladian~\cite{Childs2016EfficientDynamics} simulation, we call it the Trotter-like approach.

With this high-level intuition established, we now present the full algorithm.
Please note that the input state $\rho$ is in registers 1 and 2, where register~1 is a reference system and register~2 is the system of interest on which the Lindbladian is being simulated.
This is to enable analysis related to the diamond distance.

\textbf{Algorithm 2 --} 
\begin{enumerate}
    \item To begin with, repeat the following three steps for all $k$ ranging from $K$ to 1:
\begin{enumerate}
    \item Evolve the joint state $\rho \otimes \psi_{k}$ according to the dynamics realized by the following Lindbladian $\mathcal{M}$ acting on systems 2, 3, and 4, for some small duration of time~$\Delta_{k} \coloneqq \frac{\left \Vert L_{k} \right \Vert_{2}^{2}t}{\left \Vert \mathcal{L} \right \Vert_{\max} n} $:
    \begin{equation}\label{eq:lindblad-fixed-M-master}
        \mathcal{M}\left (\cdot\right) \coloneqq M(\cdot) M^{\dagger} - \frac{1}{2} \left \{M^{\dagger}M, (\cdot) \right\}.
    \end{equation}
    In the above, we define the Lindblad operator $M$ as
    \begin{equation}
        M \coloneqq \frac{1}{\sqrt{d}}\left(I_{2}\otimes |\Gamma\rangle\! \langle \Gamma |_{34}\right) \left( \mathsf{SWAP}_{23} \otimes I_{4}\right),
        \label{eq:lindblad-orig-M-def}
    \end{equation}
    and we define $\left \Vert \mathcal{L} \right \Vert_{\max}$ as
    \begin{equation}
        \left \Vert \mathcal{L} \right \Vert_{\max} \coloneqq \max\left \{|c_{1}|, \ldots, |c_{J}|, \left \Vert L_{1} \right \Vert_{2}^{2}, \ldots, \left \Vert L_{K} \right \Vert_{2}^{2} \right \}.
        \label{eq:L-max-def}
    \end{equation}
    \item Trace out registers 3 and 4.
    \item Set $\rho$ to be the state obtained after performing the above two steps.
\end{enumerate}
\item Then, repeat the following three steps for all $j$ ranging from $J$ to 1:
\begin{enumerate}
    \item Evolve the joint state $\rho \otimes \sigma_{j}$ according to the dynamics realized by the following Lindbladian $\mathcal{N}$ acting on systems 2 and 3, for some small duration of time~$\Delta'_{j} \coloneqq \frac{|c_{j}|t}{\left \Vert \mathcal{L} \right \Vert_{\max} n} $:
    \begin{equation}\label{eq:lindblad-fixed-N-master}
        \mathcal{N}\left (\cdot\right) \coloneqq -i \left [  \mathsf{sign}(c_{j}) \mathsf{SWAP}_{23}, (\cdot) \right ].
    \end{equation}
    \item Trace out register 3.
    \item Set $\rho$ to be the state obtained after performing the above two steps.
\end{enumerate}
\item Repeat Steps~2(a)--2(c) for all $j$ ranging from $1$ to $J$.
\item Repeat Steps~1(a)--1(c) for all $k$ ranging from $1$ to $K$.
\end{enumerate}

\begin{figure}
    \centering
    \includegraphics[scale=0.13]{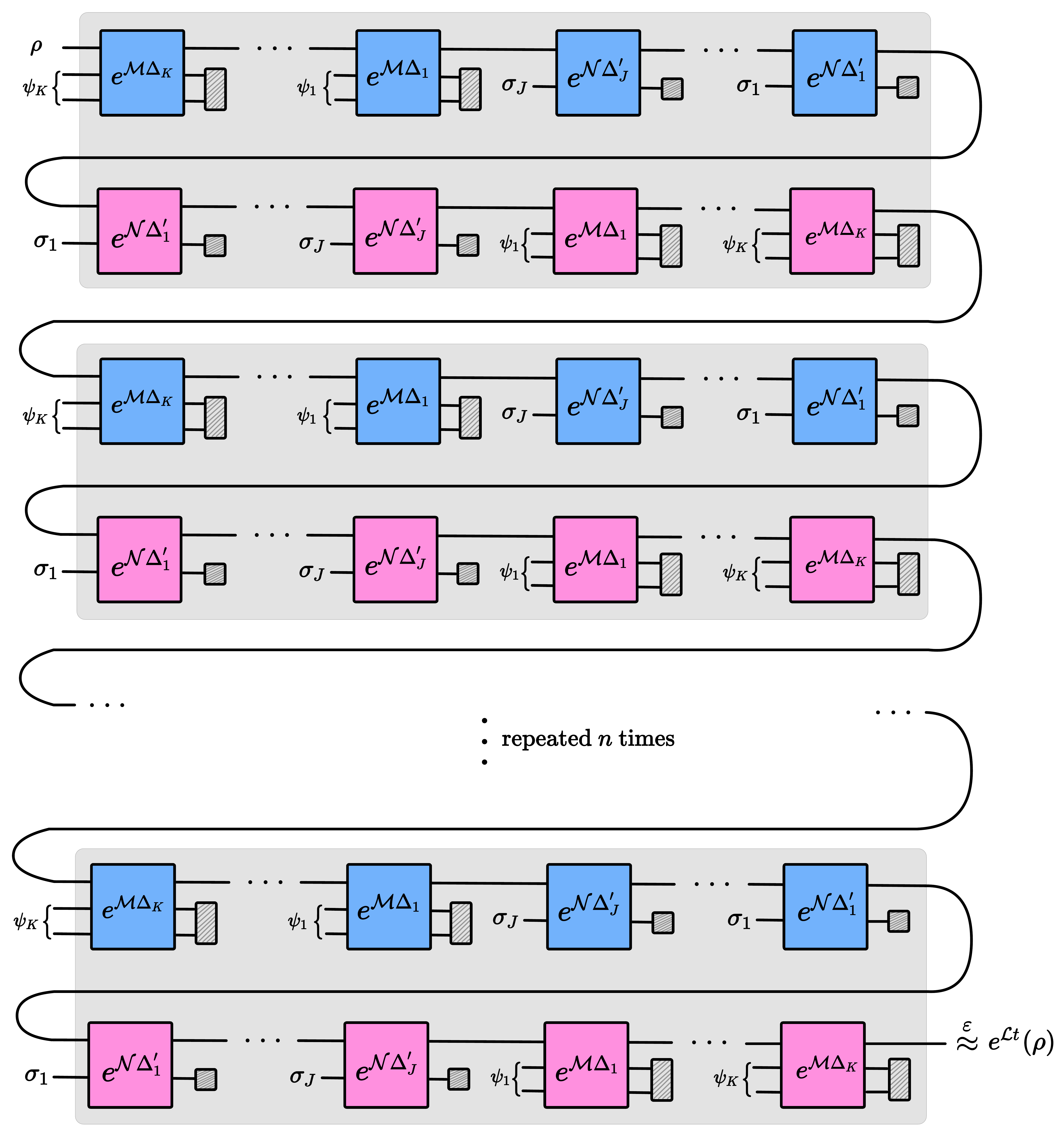}
    \caption{The gray strip at the very top represents the flow of Algorithm~2 and is repeated $n = O((J+K)\left\Vert \mathcal{L}\right\Vert_{\max}^2 t^2/\varepsilon)$ times. The blue boxes represent the flow of Steps 1 and 2 of Algorithm~2, while the pink boxes represent that of Steps 3 and 4 (reverse of Steps 1 and 2). The algorithm  approximates the target quantum channel~$e^{\mathcal{L}t}$ with an approximation error of~$\varepsilon$. Each hatched rectangle represents the trace-out operation. Initially, the state $\rho$ is in register~1. The program states $\psi_{1}, \ldots, \psi_{K}$ are in registers~2 and 3, while the program states $\sigma_{1}, \ldots, \sigma_{J}$ are in register~2. The channels $e^{\mathcal{M}\Delta_{i}}$ and $e^{\mathcal{N}\Delta'_{i}}$ correspond to Lindbladian evolutions with Lindbladians $\mathcal{M}$ and $\mathcal{N}$  defined in \eqref{eq:lindblad-fixed-M-master} and~\eqref{eq:lindblad-fixed-N-master}, respectively.}
    \label{fig:trotter-lindblad}
\end{figure}

We repeat the above procedure $n$ times. The step-by-step flow of this algorithm is illustrated in Figure~\ref{fig:trotter-lindblad}. Let us note that the particular ordering of steps used in Algorithm~2 is inspired by that used in \cite[Proposition~2]{Childs2016EfficientDynamics}, and it has the effect of reducing the dependence of the algorithm's sample complexity on $J+K$.

Theorem~\ref{thm:trotter-approach} below states that the algorithm presented above approximates the channel $e^{\mathcal{L}t}$ up to $\varepsilon$ error in diamond distance using $n = O( (J+K)\left \Vert \mathcal{L} \right \Vert_{\max}^2 t^2/\varepsilon)$ copies of each of the program states $\sigma_{1}, \ldots, \sigma_{J}, \psi_{1}, \ldots, \psi_{K}$, so that the total number of program states used is $O((J+K)^2\left \Vert \mathcal{L} \right \Vert_{\max}^2 t^2/\varepsilon)$. 
% However, we are interested in simulating the channel $e^{\mathcal{L}t}$; this can be accomplished by simulating the channel $e^{\frac{\mathcal{L}}{\left \Vert \mathcal{L} \right \Vert_{\max}} t'}$ for time $t' = \left \Vert \mathcal{L} \right \Vert_{\max} t$. This implies that we need $n = O(\left \Vert \mathcal{L} \right \Vert_{\max}^2 t^2/\varepsilon)$ copies of the program states for simulating the channel $e^{\mathcal{L}t}$ up to $\varepsilon$ error in diamond distance.

\begin{theorem}
\label{thm:trotter-approach}
    Given access to $n$ copies of the program states $\sigma_{1}, \ldots, \sigma_{J} \in \mathcal{D}(\mathcal{H}_{H})$ and $\psi_{1}, \ldots, \psi_{K} \in \mathcal{D}(\mathcal{H}_{P} \otimes \mathcal{H}_{Q})$, there exists an algorithm $\mathcal{A}$ satisfying:
    \begin{equation}
        \frac{1}{2}\left \Vert e^{\mathcal{L}t} - \mathcal{A} \right \Vert_{\diamond} \leq \varepsilon.
    \end{equation}
    This approximation to within $\varepsilon$ in diamond distance is achieved using only $n = O((J+K)\left \Vert \mathcal{L} \right \Vert_{\max}^2 t^2/\varepsilon)$ copies of each provided program state, where~$\left \Vert \mathcal{L} \right \Vert_{\max}^2$ is defined in \eqref{eq:L-max-def}. 
    The total number of program states used is then 
    \begin{equation}
        O\!\left((J+K)^2\left \Vert \mathcal{L} \right \Vert_{\max}^2 t^2/\varepsilon\right ).
    \end{equation}
 \end{theorem}

\begin{proof}
The proof is given in Appendix~\ref{app:error-analysis-gen}.
\end{proof}

\begin{remark}
To compare the sample complexity of Algorithms~1 and 2, consider that the following inequality holds:
\begin{equation}
    c \leq (J+K)\left \Vert \mathcal{L} \right \Vert_{\max},
\end{equation}
where $c$ is defined in \eqref{eq:normalization-c-gen}.
As a consequence, it follows that the sample complexity of Algorithm~1 never exceeds that of Algorithm~2. Regardless, for completeness and comparison purposes, we have provided a full analysis of the Trotter-like approach in Appendix~\ref{app:error-analysis-gen}.
\end{remark}

\subsection{Simulating Linear Combinations}

\label{sec:single-op-lin-comb}

In this section, we consider the case of simulating a Lindbladian evolution with a single Lindblad operator $L$ and no Hamiltonian. 
To be more specific, we want to simulate the action of the dynamics on a quantum state $\rho$ for time~$t$ according to the following Lindbladian:
\begin{align}
    \mathcal{L} (\rho) = L\rho L^{\dagger} - \frac{1}{2} \left \{L^{\dagger}L, \rho \right\}.
    \label{eq:lindbladian-single}
\end{align}
Furthermore, in this case, let us suppose that the Lindblad operator $L$ is provided beforehand as a linear combination of encoded operators $L_{1}, \ldots, L_{K}$:
\begin{equation}\label{eq:linear-comb-L}
    L \coloneqq  \sum_{k=1}^{K} c_{k} L_{k},
\end{equation}
where $c_{k} > 0$ and $\left\Vert L_k\right\Vert_2 = 1$, for all $k$ ranging from 1 to $K$, and the operator~$L_{k}$ is encoded in the program state $\psi_{k}$ in the same way that was considered previously in~\eqref{eq:L-encode} (i.e., $|\psi_k\rangle \coloneqq (L_k \otimes I) |\Gamma\rangle$).
Let us further suppose that we have access to unitaries $U_{1}, \ldots, U_{K}$ that prepare these program states:
\begin{equation}
    U_{k}|0\rangle = |\psi_{k}\rangle \qquad \forall k \in \{1,\ldots, K\}.
\end{equation}

Given efficient implementations of the above unitaries, we can efficiently implement a unitary $\text{select-}U$ defined in the following way:
\begin{equation}
    \text{select-}U \coloneqq \sum_{k=1}^{K} |k\rangle\!\langle k| \otimes U_{k}.
\end{equation}
Let us additionally suppose that we have access to a unitary $U_{A}$ that can efficiently prepare the state:
\begin{equation}
    U_A |0\rangle = |A\rangle \coloneqq \frac{1}{\sqrt{\sum_{k=1}^{K} c_{k}}} \sum_{k=1}^{K}\sqrt{c_{k}}|k\rangle.
\end{equation}

Given access to the unitaries $\{U_{k}\}_{k=1}^{K}$ and their Hermitian conjugates $\{U_{k}^{\dagger}\}_{k=1}^{K}$ (for realizing the unitaries select-$U$ and select-$U^{\dagger}$), as well as unitaries $U_{A}$ and $U_{A}^{\dagger}$, we can use the standard linear combination of unitaries (LCU) technique \cite{childs2012} to prepare the following state:
\begin{equation}
    |\phi\rangle \coloneqq \frac{1}{\sqrt{c}}\sum_{k=1}^{K} c_{k}|\psi_{k}\rangle = \frac{1}{\sqrt{c}} (L \otimes I)|\Gamma\rangle,
    \label{eq:program-state-linear}
\end{equation}
where $c \coloneqq \left \Vert L \right \Vert^2_{2}$.
Due to the fact that the LCU technique is inherently probabilistic, we use the  amplitude amplification (AA) method~\cite{brassard2002quantum} to boost the probability of success to one. 
We do not go into specifics here because LCU and AA are both fairly common quantum algorithmic primitives. 
However, for completeness, we show in Figure~\ref{fig:linear-comb} the exact quantum circuit that implements this method for preparing our desired state given by~\eqref{eq:program-state-linear}. 
Please refer to \cite[Section 2.2]{kothari2014efficient} and \cite[Section~II-B]{Chakraborty2023} for more precise information on the LCU method and using AA to boost the success probability. 
Overall, we prepare $n$ copies of this state for our purposes. 

Now, we propose a quantum algorithm for simulating the quantum channel~$e^{\mathcal{L}t}$, which corresponds to the Lindbladian in \eqref{eq:lindbladian-single}, using $n$ copies of the program state $\phi \coloneqq |\phi\rangle\!\langle \phi|$ up to error $\varepsilon$ in diamond distance. 
As before, for providing an analysis related to the diamond distance, the input is an unknown quantum state $\rho \in \mathcal{D}(\mathcal{H}_{R} \otimes \mathcal{H}_{S})$ on joint system $RS$, with $R$ as a reference system. 
Furthermore, let the $i^{\text{th}}$ copy of the program state $\phi$ be a quantum state of the joint system $P_{i}Q_{i}$.

\textbf{Algorithm~3 ---} Set $n\in \mathbb{N}$, with the particular choice specified later. Set $i=1$. Given the $i^{\text{th}}$ copy of $\phi$, i.e., $\phi_{P_{i}Q_{i}}$, perform the following two steps:\label{algo:single-op}
\begin{enumerate}
    \item Evolve the joint quantum state $\rho_{RS}\otimes \phi_{P_{i}Q_{i}}$ under the following Lindbladian $\mathcal{M}$ acting on $SP_i Q_i $, for some small duration of time~$\Delta = ct/n $:
    \begin{equation}\label{eq:single-fixed-L-master-linear}
        \mathcal{M}\left (\cdot\right) \coloneqq M(\cdot) M^{\dagger} - \frac{1}{2} \left \{M^{\dagger}M, (\cdot) \right\}.
    \end{equation}
    In the above, the Lindblad operator $M$ acts on the joint system $SP_{i}Q_{i}$, and we define it as
    \begin{equation}
        M \coloneqq \frac{1}{\sqrt{d}}\left(I_{S}\otimes |\Gamma\rangle\! \langle \Gamma |_{P_{i}Q_{i}}\right) \left( \mathsf{SWAP}_{SP_{i}} \otimes I_{Q_{i}}\right).
        \label{eq:orig-M-def-linear}
    \end{equation}
    \item Trace out systems $P_{i}Q_{i}$.
\end{enumerate}
We repeat the above procedure using each copy of $\phi$, i.e., for all $i$ ranging from $1$ to $n$.

\begin{figure}
    \centering
    \includegraphics[scale=0.1]{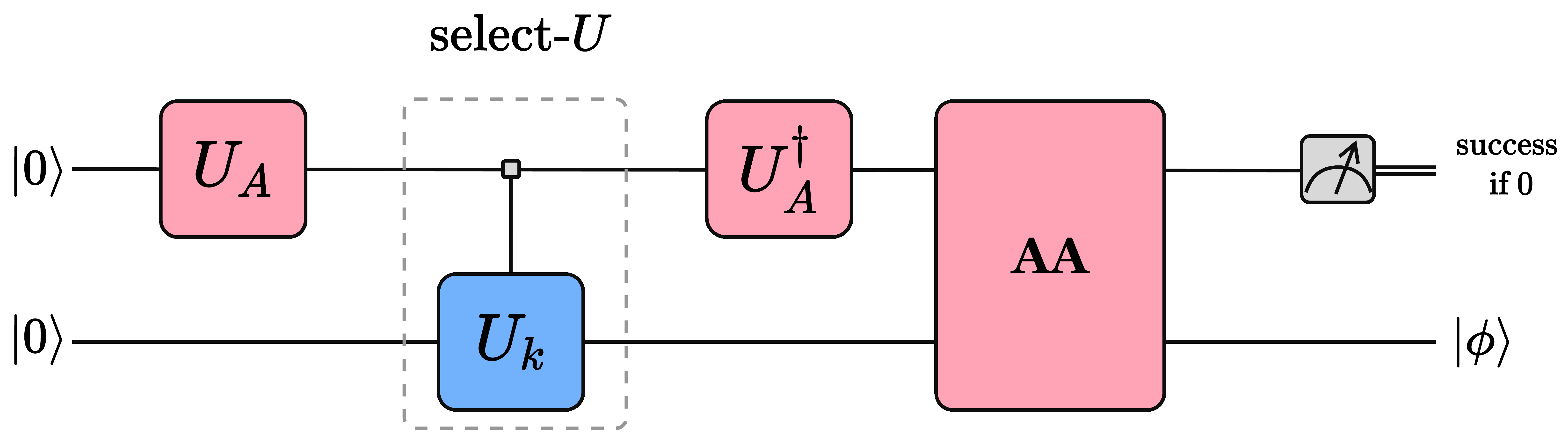}
    \caption{A quantum circuit for preparing $|\phi\rangle$ using the technique of linear combination of unitaries. The rectangle labelled ``AA'' represents multiple rounds of  amplitude amplification.}
    \label{fig:linear-comb}
\end{figure}

The following theorem states that the above algorithm uses $n = O(c^2t^2/\varepsilon)$ copies of $\phi$ to simulate the Lindbladian evolution of $\rho_{RS}$, according to the Lindbladian in \eqref{eq:lindbladian-single}, for time $t$, such that the final state is $\varepsilon$-close in normalized trace distance to the ideal target state $\left(\mathcal{I}_{R} \otimes \e^{\mathcal{L}t} \right)\left(\rho_{RS}\right)$, for an arbitrary input state $\rho_{RS}$. For preparing one copy of $\phi$, we need to make 
\begin{equation}
    O\!\left(\frac{1 }{\sqrt{c}}\sum_{k=1}^{K} c_{k}\right)
\end{equation}
queries to the unitaries $\{U_{k}\}_{k=1}^{K}$, $\{U_{k}^{\dagger}\}_{k=1}^{K}$, $U_{A}$, and $U_{A}^{\dagger}$ \cite[Section~II-B]{Chakraborty2023}. Therefore, for preparing $O(c^2t^2/\varepsilon)$ copies of $\phi$, we need to make
\begin{equation}
    O\!\left(\frac{1}{\varepsilon}\left(\sum_{k=1}^{K} c_{k} \right)c^{3/2}t^2\right)
\end{equation}
queries to these unitaries.

\begin{theorem}\label{thm:lin-comb}
Given access to $n$ copies of the program state $\phi \in \mathcal{D}\!\left(\mathcal{H}_{P} \otimes \mathcal{H}_{Q}\right)$, defined in \eqref{eq:program-state-linear}, each of which can be prepared using unitaries $\{U_{k}\}_{k=1}^{K}$, $\{U_{k}^{\dagger}\}_{k=1}^{K}$, $U_{A}$, and $U_{A}^{\dagger}$, there exists a quantum algorithm~$\mathcal{A}$ such that the following error bound holds:
    \begin{equation}
        \frac{1}{2}\left \Vert \e^{\mathcal{L}t} - \mathcal{A} \right \Vert_{\diamond} \leq \varepsilon,
    \end{equation}
with only $n = O(c^2 t^2/\varepsilon)$ copies of $\phi$, where $c = \left \Vert L \right \Vert^2_{2}$. In other words, $\mathcal{A}$ uses only $n = O(c^2 t^2/\varepsilon)$ copies of $\phi$ to approximate the channel $ \e^{\mathcal{L}t}$ up to $\varepsilon$ error in normalized diamond distance.
\end{theorem}

\begin{proof}
In what follows, we give a brief sketch of the proof. We will not go into detail because the proof follows a similar line of reasoning as the proof of Theorem~1 of \cite{Patel2023}. 
We leave out system labels for clarity and simplicity, and let us also suppose that the input state $\rho$ does not have a reference system~$R$ for the purposes of this proof sketch.

Let us begin by expanding the target state $\e^{\mathcal{L}t}(\rho)$ at $t=0$ using its Taylor series:
\begin{equation}\label{eq:single-desired-state-linear}
    \e^{\mathcal{L}t}(\rho) = \rho + \mathcal{L}(\rho) t +  \frac{1}{2} (\mathcal{L}\circ  \mathcal{L})(\rho) t^2 + \ldots .
\end{equation}
In the first step of Algorithm~3, we simulate the Lindbladian evolution of $\rho \otimes \phi$, given by \eqref{eq:single-fixed-L-master-linear}, with the Lindbladian $\mathcal{M}$ for some small duration of time $\Delta$, and we then trace out $\phi$. The output state obtained after this step is
\begin{equation}\label{eq:single-op-proof-sketch}
    \operatorname{Tr}_{23}\!\left[\e^{\mathcal{M}\Delta}(\rho \otimes \phi)\right] = \rho + \operatorname{Tr}_{23}\!\left[\mathcal{M}(\rho \otimes \phi) \right] \Delta + O(\Delta^2).
\end{equation}
We first expand the second term and rewrite it in the following way:
\begin{multline}
    \operatorname{Tr}_{23}\!\left[\mathcal{M}\left (\rho \otimes \phi \right) \right ] = \operatorname{Tr}_{23}\! \left [M(\rho \otimes \phi) M^{\dagger} \right]\\  - \frac{1}{2}\operatorname{Tr}_{23}
    \!\left [M^{\dagger}M \left(\rho \otimes  \phi\right) \right] - \frac{1}{2}\operatorname{Tr}_{23}
    \!\left [\left(\rho \otimes  \phi\right)M^{\dagger}M  \right].
\end{multline}
Substituting $\phi =  \frac{1}{c}(L \otimes I)|\Gamma\rangle\!\langle \Gamma |(L^{\dag} \otimes I)$ into the above equation, we get
\begin{align}\notag
    \operatorname{Tr}_{23}\!\left[\mathcal{M}\left (\rho \otimes \phi \right) \right ]  
    & = \frac{1}{c}  \Bigg ( \operatorname{Tr}_{23}\! \left [M \left (\rho \otimes  (L \otimes I)|\Gamma\rangle\!\langle \Gamma |(L^{\dag} \otimes I)   \right) M^{\dagger} \right] \notag \\ 
    &\qquad - \frac{1}{2}\operatorname{Tr}_{23}
    \left [M^{\dagger}M \left(\rho \otimes  (L \otimes I)|\Gamma\rangle\!\langle \Gamma |(L^{\dag} \otimes I) \right) \right] \notag \\
    & \qquad - \frac{1}{2}\operatorname{Tr}_{23}
    \left [\left(\rho \otimes   (L \otimes I)|\Gamma\rangle\!\langle \Gamma |(L^{\dag} \otimes I)\right)M^{\dagger}M  \right] \Bigg) \label{eq:thm-2-proof-M}\\
    & = \frac{1}{c} \left ( L \rho  L^{\dagger} - \frac{1}{2} L^{\dagger}L \rho - \frac{1}{2} \rho  L^{\dagger}L \right)
    \label{eq:thm-2-proof-L} \\
    & = \frac{1}{c} \mathcal{L}(\rho).\label{eq:Tr[M]-L-linear}
\end{align}
We use tensor-network diagrams, as shown in Figures 3, 4, 5, and 6 of \cite{Patel2023}, to obtain \eqref{eq:thm-2-proof-L} from \eqref{eq:thm-2-proof-M} (alternatively, see Appendix~A of \cite{Patel2023}).

Using \eqref{eq:Tr[M]-L-linear}, we rewrite \eqref{eq:single-op-proof-sketch} as
\begin{align}
    \operatorname{Tr}_{23}\!\left[\e^{\mathcal{M}\Delta}(\rho \otimes \phi)\right] & = \rho + \frac{1}{c} \mathcal{L}(\rho)\Delta + O(\Delta^2) \\
    & = \e^{(\mathcal{L}/c) \Delta}(\rho) + O(\Delta^2).
\end{align}
Substituting $\Delta = c t/n$ and repeating Algorithm~3 for $n = O(c^2 t^2 / \varepsilon)$ times produces a quantum state that is $\varepsilon$-close to the ideal target state~$\e^{\mathcal{L}t}(\rho)$ in normalized trace distance.
\end{proof}

\subsection{Simulating Lindbladian Polynomials}

\label{sec:single-op-poly}

In this section, we consider the case of simulating a Lindbladian evolution with a single Lindblad operator $L$ and no Hamiltonian term, as in the previous section (see \eqref{eq:lindbladian-single}). Here, we suppose that the Lindblad operator $L$ is represented as a polynomial of linear operators encoded in program states.

To be more specific, let us suppose that $L$ can be decomposed as follows:
\begin{equation}\label{eq:poly-L}
    L \coloneqq \sum_{s\in \mathcal{S}} c_{s} T_{s},
\end{equation}
where $c_{s} > 0$ and each term $T_{s}$, defined as follows, is of degree $|s|$:
\begin{equation}\label{eq:poly-exp}
    T_{s} = L_{s[1]} L_{s[2]} \cdots L_{s[|s|]}.
\end{equation}
Here, we define $\mathcal{S}$ as a set of strings over an alphabet $\{1, 2, \ldots, K\}$. Let the notation $s[i]$ denote the $i^{\text{th}}$ character of the string $s$. Furthermore, we use the notation $|s|$ to denote the length of the string $s$, and we use the notation $D$ to denote the degree of the polynomial given by \eqref{eq:poly-L}, i.e., $D \coloneqq \max_{s \in \mathcal{S}}\{|s|\}$. Finally, let us suppose that the set $\mathcal{S}$ consists of polynomially many (in the number of qubits) strings, i.e., $|\mathcal{S}| = \log d$, where $|\mathcal{S}|$ denotes the cardinality of the set $\mathcal{S}$.
The operators $L_{1}, \ldots, L_{K}$ are encoded in the program states $\psi_{1}, \ldots, \psi_{K}$, respectively, in the same way that was considered previously in~\eqref{eq:L-encode}, so that $\left \Vert L_1 \right \Vert_2 = \cdots = \left \Vert L_K \right \Vert_2 = 1$. For this case as well, let us suppose that we have access to unitaries $U_{1}, \ldots, U_{K}$ that prepare these program states by acting on the maximally entangled state~$|\Phi\rangle$ as follows:
\begin{equation}
    U_{k}|\Phi\rangle = |\psi_{k}\rangle.
\end{equation}

Given efficient implementations of the above unitaries, we can efficiently implement a unitary select-$W$ defined as:
\begin{equation}
    \text{select-}W \coloneqq \sum_{i = 0}^{|\mathcal{S}| - 1} |i\rangle\!\langle i| \otimes W_{s_{i}},
\end{equation}
where $s_{i}$ is the $i^{\text{th}}$ string of the set $\mathcal{S}$ (the order is irrelevant for indexing purposes) and
\begin{align}
    W_{s_{i}} \coloneqq U_{s_{i}[1]} \otimes U_{s_{i}[2]} \otimes \cdots \otimes U_{s_{i}[|s|]} \otimes I^{\otimes D-|s|}.
\end{align}
Moreover, let us suppose that we have access to a unitary $U_{A}$ that can efficiently prepare a pure quantum state defined as follows:
\begin{equation}
    |A\rangle \coloneqq \frac{1}{\sqrt{\sum_{i=0}^{|\mathcal{S}| - 1} c_{s_i}}} \sum_{i=0}^{|\mathcal{S}| - 1}\sqrt{c_{s_i}}|i\rangle.
\end{equation}

Given access to the set $\{U_{k}\}_{k=1}^{K}$ of unitaries  and their Hermitian conjugates~$\{U_{k}^{\dagger}\}_{k=1}^{K}$ (for realizing unitaries select-$W$ and select-$W^{\dagger}$), as well as unitaries $U_{A}$ and $U_{A}^{\dagger}$, we can prepare the following state using the LCU technique:
\begin{align}
    |\phi\rangle \coloneqq \frac{1}{\sqrt{c}}\sum_{s\in \mathcal{S}} c_{s} |\phi_{s}\rangle,
    \label{eq:program-state-poly}
\end{align}
where
\begin{align}
    |\phi_{s}\rangle & \coloneqq |\psi_{s[1]}\rangle |\psi_{s[2]}\rangle\cdots |\psi_{s[|s|]} \rangle\underbrace{|\Phi \rangle  \cdots |\Phi \rangle }_\text{$D-|s|$ times},\\
    c & \coloneqq \left \Vert \sum_{s\in \mathcal{S}} c_{s} |\phi_{s}\rangle \right \Vert_{2}^2 \label{eq:c-def-poly}.
\end{align}
As stated before in the previous section, the LCU technique is inherently probabilistic. 
As a result, we employ the amplitude amplification method to boost its probability of success to nearly one. Figure~\ref{fig:poly} depicts the quantum circuit that implements this technique for preparing our desired state given by \eqref{eq:program-state-poly}. 
Overall, we prepare $n$ copies of this state for our purposes.

We are now in a position to propose a quantum algorithm for simulating the quantum channel $\e^{\mathcal{L}t}$, which corresponds to the Lindbladian with a single Lindblad operator given by \eqref{eq:poly-L}, using $n$ copies of the program state $\phi \coloneqq |\phi\rangle\!\langle \phi|$ up to error $\varepsilon$ in diamond distance. 
As previously stated, in order to provide an analysis of the diamond distance, let $\rho \in \mathcal{D}(\mathcal{H}_{R} \otimes \mathcal{H}_{S})$ be an unknown quantum state given as input over the joint system $RS$, where the system $R$ acts as a reference system. 
Also, let the $i^{\text{th}}$ copy of the program state $\phi$ be a quantum state over a joint system $P^{1}_{i}Q^{1}_{i}\cdots P^{D}_{i}Q^{D}_{i}$, i.e., $\phi \in \mathcal{D}(\mathcal{H}_{P^1_{i}} \otimes \mathcal{H}_{Q^1_{i}} \otimes \cdots \otimes \mathcal{H}_{P^D_{i}} \otimes \mathcal{H}_{Q^D_{i}})$. 
For brevity, we use $(PQ)^{D}_{i}$ as a shorthand for $P^{1}_{i}Q^{1}_{i}\cdots P^{D}_{i}Q^{D}_{i}$.

\begin{figure}
    \centering
    \includegraphics[scale=0.09]{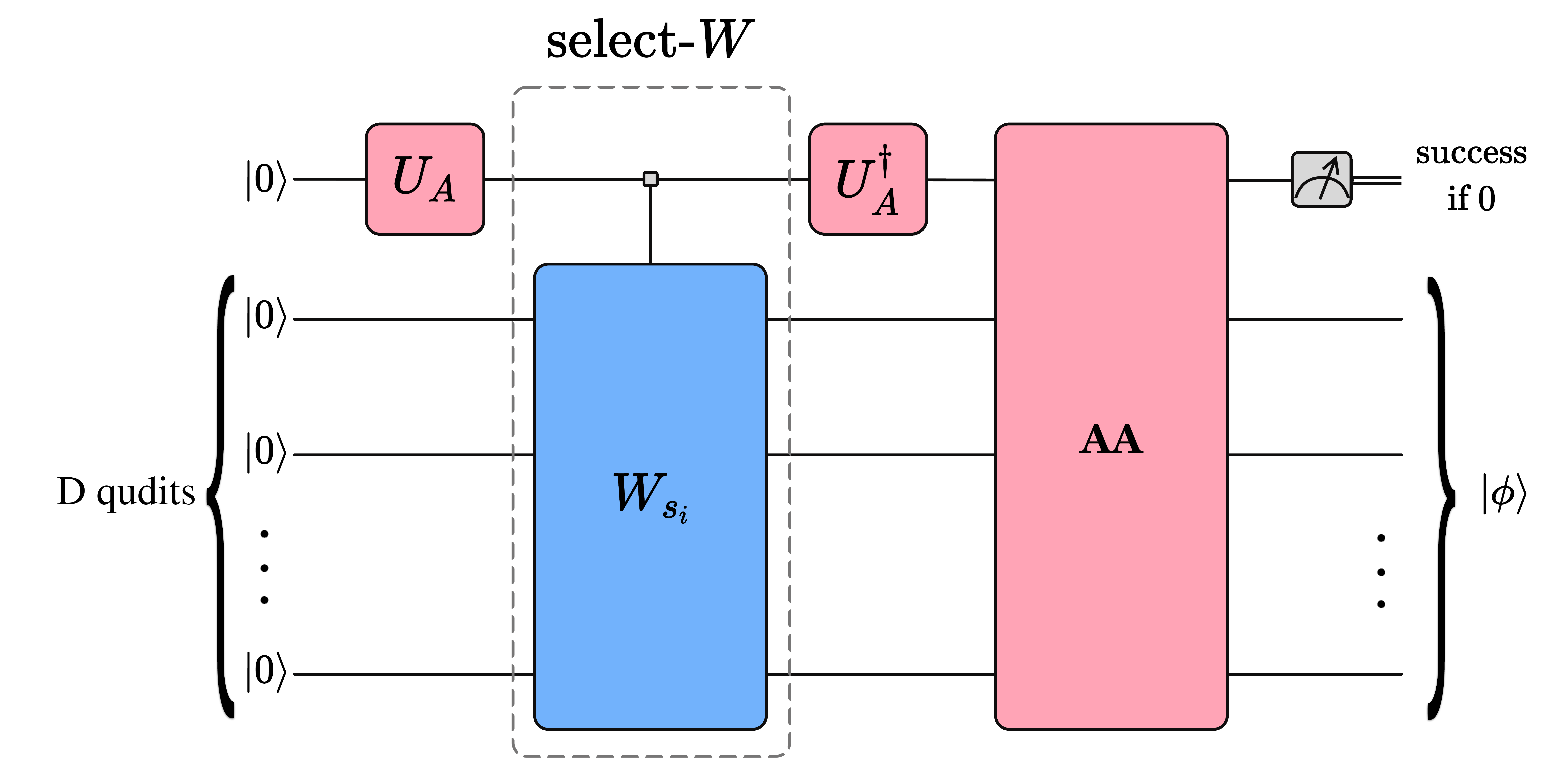}
    \caption{A quantum circuit for preparing $|\phi\rangle$ using the technique of linear combination of unitaries. The rectangle labelled ``AA'' represents multiple rounds of  amplitude amplification.}
    \label{fig:poly}
\end{figure}

\textbf{Algorithm 4 ---} Set $n\in \mathbb{N}$, with a particular choice specified later. Set~$i=1$. Given the $i^{\text{th}}$ copy of $\phi$, denoted by $\phi_{(PQ)^{D}_{i}}$, perform the following two steps:
\begin{enumerate}
    \item Evolve the joint quantum state $\rho_{RS}\otimes \phi_{(PQ)^{D}_{i}}$ according to the following Lindbladian $\mathcal{M}$ acting on systems $S(PQ)_{i}^D$, for some small duration of time $\Delta = ct/n$:
    \begin{equation}\label{eq:poly-fixed-L-master}
        \mathcal{M}(\cdot) \coloneqq M(\cdot) M^{\dagger} - \frac{1}{2} \left \{M^{\dagger}M, (\cdot) \right\},
    \end{equation}
    where the Lindblad operator $M$ acts on the joint system $S(PQ)^{D}_{i}$, and we define it as
    \begin{equation}
        M \coloneqq  \frac{1}{d^{D/2}} \left(I_{S}\otimes\left(|\Gamma\rangle\!\langle\Gamma |^{\otimes D}\right)_{(PQ)^{D}_{i}} \right) \left(  \mathsf{CYCSWAP}_{SP^{1}_{i} \cdots P^{D}_{i}} \otimes I_{Q^{1}_{i} \cdots Q^{D}_{i}}\right).
        \label{eq:M-cyc-swap}
    \end{equation}
    \item Trace out the systems $(PQ)^{D}_{i}$.
\end{enumerate}
We repeat the above procedure with each copy of $\phi$, for $i$ ranging from $1$ to~$n$.

According to the following theorem, the algorithm above uses $n = O(c^2t^2/\varepsilon)$ copies of $\phi$ to simulate the Lindbladian evolution of $\rho_{RS}$, given by \eqref{eq:lindbladian-single}, for time $t$, such that the final state is $\varepsilon$-close to the target state $\left(\mathcal{I}_{R} \otimes \e^{\mathcal{L}t} \right)\left(\rho_{RS}\right)$ in normalized trace distance, for an arbitrary input state $\rho_{RS}$.
For preparing one copy of $\phi$, we need to make 
\begin{equation}
    O\!\left(\frac{1 }{\sqrt{c}}\sum_{i=0}^{|\mathcal{S}| - 1} c_{s_i}\right)
\end{equation} queries to the unitaries $\{U_{k}\}_{k=1}^{K}$, $\{U_{k}^{\dagger}\}_{k=1}^{K}$, $U_{A}$, and $U_{A}^{\dagger}$ \cite[Section~II-B]{Chakraborty2023}. Therefore, for preparing $O(c^2t^2/\varepsilon)$ copies of $\phi$, we need to make
\begin{equation}
O\!\left(\frac{1}{\varepsilon} \left(\sum_{i=0}^{|\mathcal{S}| - 1} c_{s_i} \right)c^{3/2} t^2 \right)
\end{equation} queries to these unitaries.

\begin{figure}
    \centering
    \includegraphics[scale=0.115]{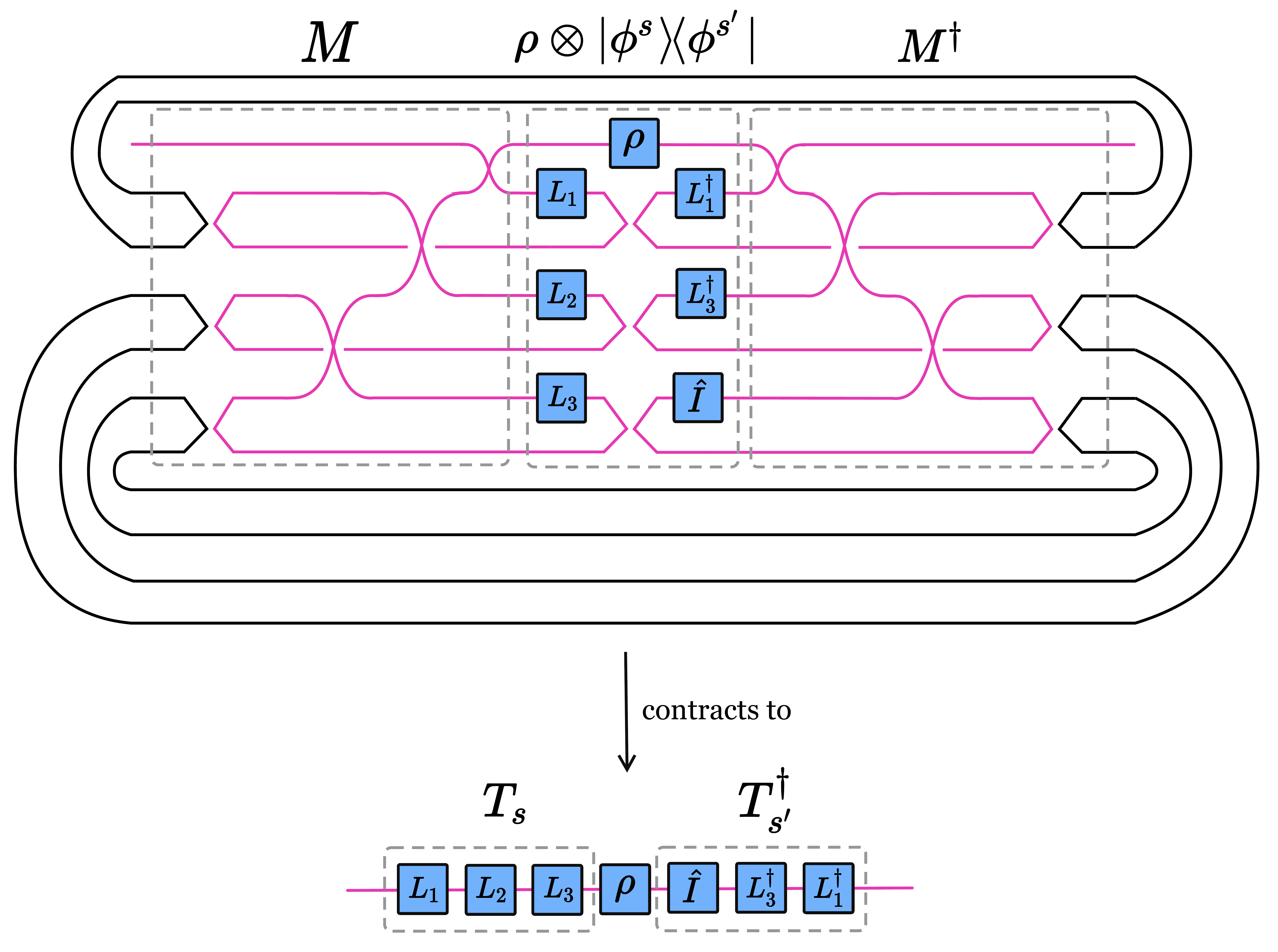}
    \caption{Tensor-network diagram of $\operatorname{Tr}_{2,\ldots,7}\!\left [M \left (\rho \otimes  |\phi^{s}\rangle\!\langle \phi^{s'}|  \right) M^{\dagger}\right]$, where~$s=123$ and $s'=13$. The whole network on the top contracts to the network on the bottom as a result of tracing out registers $2,\ldots,7$. The pink line flowing from the left end of the first system to the right end illustrates how the networks are connected after the partial trace operation.}
    \label{fig:lindblad-first-term}
\end{figure}

\begin{figure}
    \centering
    \includegraphics[scale=0.11]{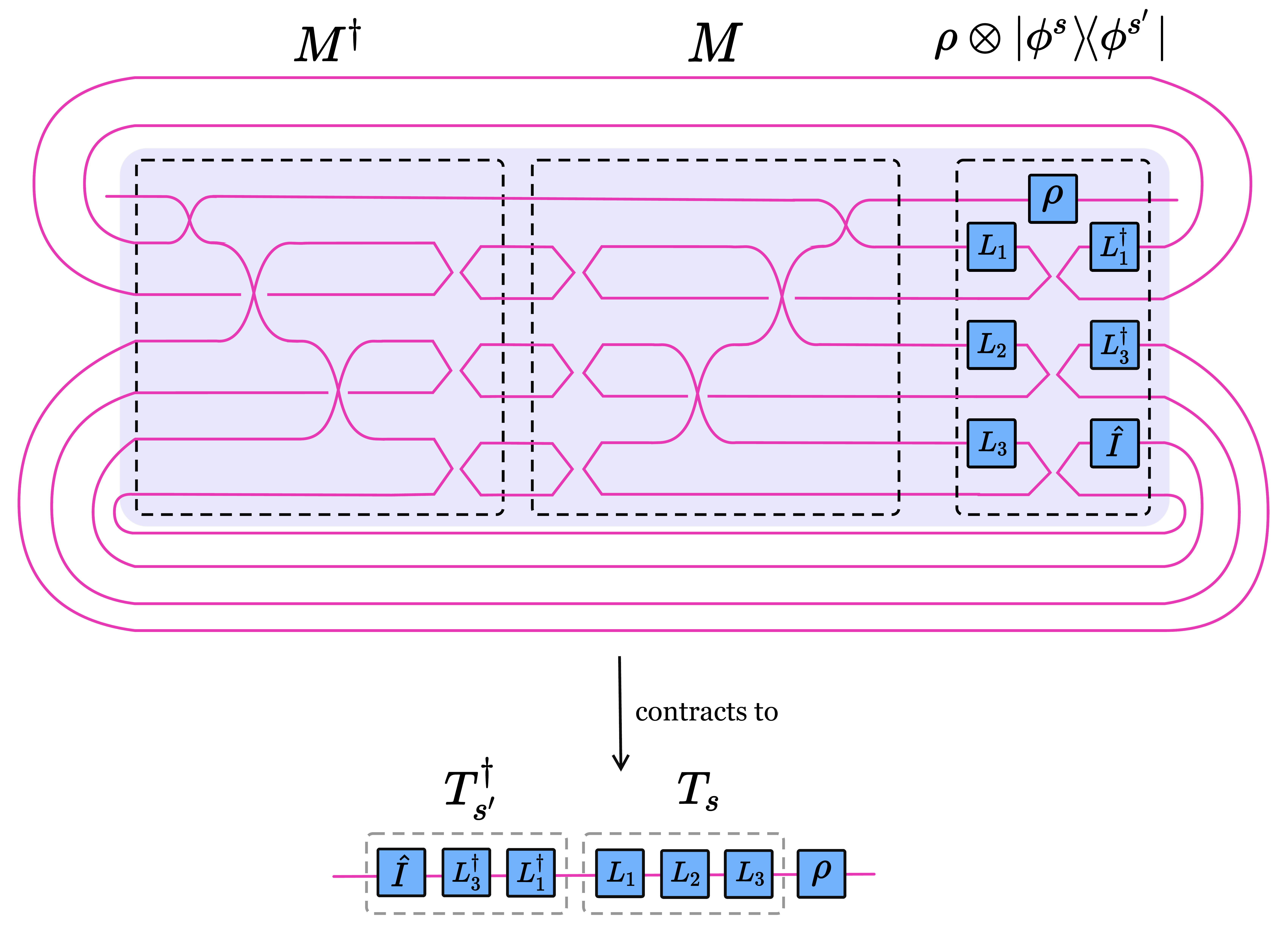}
    \caption{Tensor-network diagram of $\operatorname{Tr}_{2,\ldots,7}\!\left [M^{\dagger}M\left (\rho \otimes  |\phi^{s}\rangle\!\langle \phi^{s'}|  \right) \right]$, where $s=123$ and $s'=13$. 
    The whole network on the top contracts to the network on the bottom as a result of tracing out registers $2,\ldots,7$. 
    The tensor-network diagram enclosed in the light-purple box depicts $M^{\dagger} M\left (\rho \otimes  |\phi^{s}\rangle\!\langle \phi^{s'}|  \right)$. 
    The lines flowing out and back in represent the trace-out operation being performed on this network.
    Finally, in order to visualize how the network on the top contracts to that on the bottom after the trace-out operation, simply follow the pink line flowing from the left end of the first system to its right end.}
    \label{fig:lindblad-second-term}
\end{figure}

\begin{figure}
    \centering
    \includegraphics[scale=0.11]{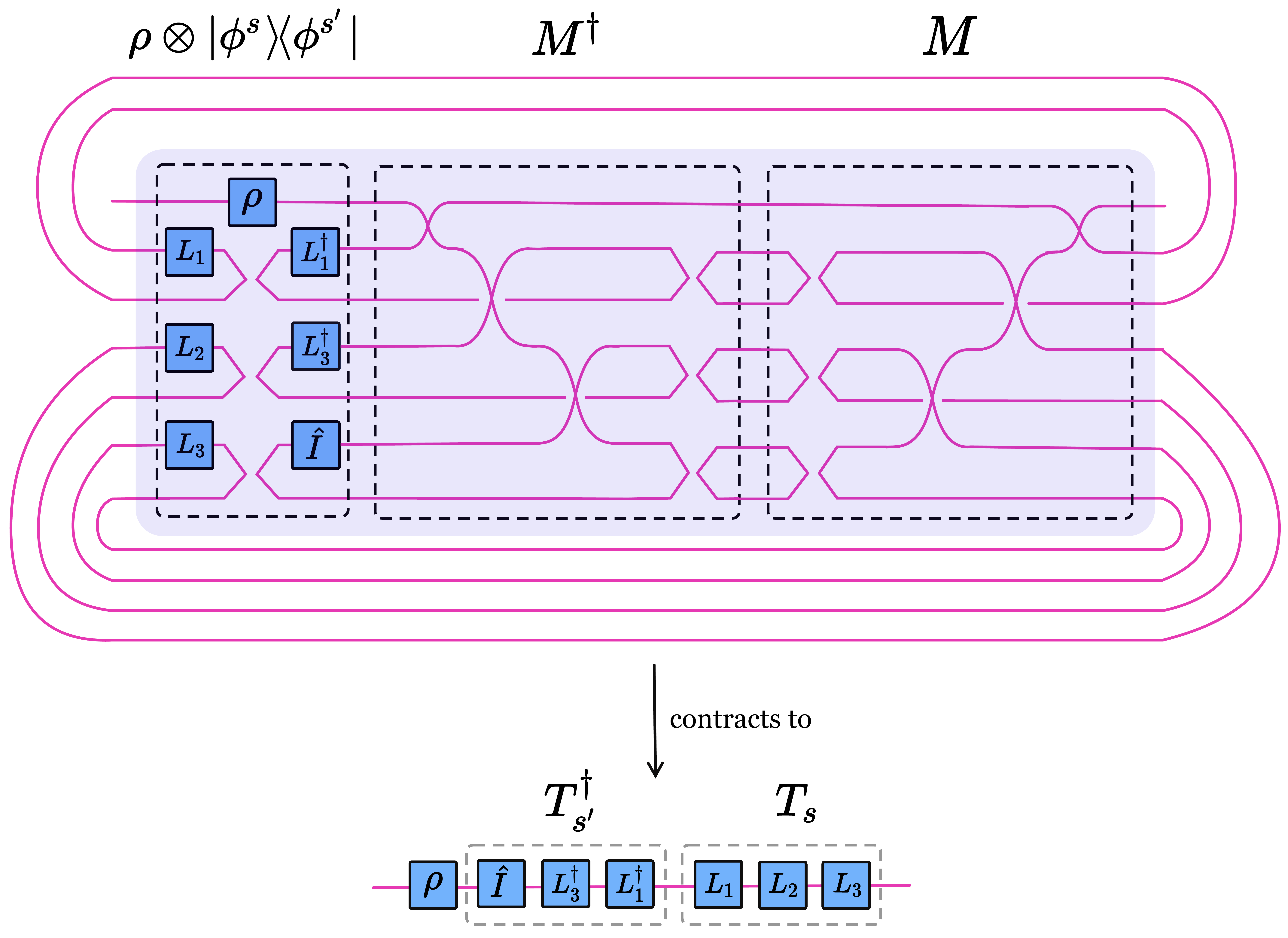}
    \caption{Tensor-network diagram of $\operatorname{Tr}_{2,\ldots,7}\!\left [\left (\rho \otimes  |\phi^{s}\rangle\!\langle \phi^{s'}|  \right) M^{\dagger}M \right]$, where $s=123$ and $s'=13$.
    The whole network on the top contracts to the network on the bottom as a result of tracing out registers $2,\ldots,7$. 
    The tensor-network diagram enclosed in the light-purple box depicts $\left (\rho \otimes  |\phi^{s}\rangle\!\langle \phi^{s'}|  \right) M^{\dagger} M$. 
    The lines flowing out and back in represent the trace-out operation being performed on this network.
    Finally, in order to visualize how the network on the top contracts to that on the bottom after the trace-out operation, simply follow the pink line flowing from the left end of the first system to its right end.}
    \label{fig:lindblad-thirdterm}
\end{figure}

\begin{theorem}
Given access to $n$ copies of the program state $\phi \in \mathcal{D}(\mathcal{H}_{P^1_{i}} \otimes \mathcal{H}_{Q^1_{i}} \otimes \cdots \otimes \mathcal{H}_{P^D_{i}} \otimes \mathcal{H}_{Q^D_{i}})$, each of which can be prepared using unitaries $\{U_{k}\}_{k=1}^{K}$, $\{U_{k}^{\dagger}\}_{k=1}^{K}$, $U_{A}$, and $U_{A}^{\dagger}$, there exists a quantum algorithm~$\mathcal{A}$ such that the following error bound holds:
    \begin{equation}
        \frac{1}{2}\left \Vert \e^{\mathcal{L}t} - \mathcal{A} \right \Vert_{\diamond} \leq \varepsilon,
    \end{equation}
with only $n = O(c^2 t^2/\varepsilon)$ copies of $\phi$, where $c$ is defined in \eqref{eq:c-def-poly}. In other words, $\mathcal{A}$ uses only $n = O(c^2 t^2/\varepsilon)$ copies of $\phi$ to approximate the channel $ \e^{\mathcal{L}t}$ up to $\varepsilon$ error in normalized diamond distance.
\end{theorem}

\begin{proof}
For ease of notation, we do not explicitly write system labels, and let us suppose that the input state $\rho$ does not have a reference system $R$. 

The key point here is to demonstrate that the following three equalities hold for all $s, s' \in \mathcal{S}$:
\begin{align}
    \operatorname{Tr}_{2,3,\ldots, 2D+1}\left [M \left (\rho \otimes  |\phi^{s}\rangle\!\langle \phi^{s'}|  \right) M^{\dagger}\right] & =  T_{s} \rho  T_{s'}^{\dagger} , \\
    \operatorname{Tr}_{2,3,\ldots, 2D+1}
    \left [M^{\dagger}M \left(\rho \otimes   |\phi^{s}\rangle\!\langle \phi^{s'}| \right) \right] & =  T_{s}^{\dagger}T_{s'} \rho , \\
    \operatorname{Tr}_{2,3,\ldots, 2D+1}
     \left [\left(\rho \otimes  |\phi^{s}\rangle\!\langle \phi^{s'}| \right)M^{\dagger}M  \right] & =  \rho  T_{s}^{\dagger}T_{s'},
\end{align}
for $M$ as defined in \eqref{eq:M-cyc-swap}. Lemma~\ref{lem:cyclic-swap-lem} in Appendix~\ref{app:cyclic-swap-lem} provides a full statement and proof. One can alternatively employ tensor-network diagrams to establish these equalities.
For visualizing these diagrams for a simple case with $s=123$ and $s'=13$, please refer to Figures~\ref{fig:lindblad-first-term}, \ref{fig:lindblad-second-term}, and \ref{fig:lindblad-thirdterm}.

With the above equalities holding for all $s, s' \in \mathcal{S}$, then we can essentially use the same type of reasoning as in the proof of Theorem~\ref{thm:lin-comb} to conclude that
\begin{equation}
\operatorname{Tr}_{2,3,\ldots, 2D+1}\!\left[\mathcal{M}\left (\rho \otimes \phi \right) \right ] = \frac{1}{c} \mathcal{L}(\rho).\label{eq:Tr[M]-L-poly}
\end{equation}
This further implies that
\begin{align}
    \operatorname{Tr}_{2,3,\ldots ,2D+1}\!\left[\e^{\mathcal{M}\Delta}(\rho \otimes \phi)\right] & = \rho + \operatorname{Tr}_{2,3,\ldots, 2D+1}\!\left[\mathcal{M}\left (\rho \otimes \phi \right) \right ]\Delta + O(\Delta^2)\\
    & = \rho + \frac{1}{c} \mathcal{L}(\rho)\Delta + O(\Delta^2) \\
    & = \e^{(\mathcal{L}/c) \Delta}(\rho) + O(\Delta^2).
\end{align}

Substituting $\Delta = c t/n$ and repeating Algorithm~4 for $n = O(c^2 t^2 / \varepsilon)$ times produces a quantum state that is $O(\varepsilon)$-close to the ideal target state~$\e^{\mathcal{L}t}(\rho)$ in normalized trace distance.
\end{proof}

\subsection{Wave Matrix Lindbladization Improves upon Tomography}

\label{sec:WML-vs-WMT}

In this section, we provide a comparison between two different methods of simulating the channel $\e^{\mathcal{L}t}$, either by wave matrix Lindbladization, or by performing tomography to determine a classical description of the Lindbladian~$\mathcal{L}$, followed by simulation from this description. The main finding, summarized in Figure~\ref{fig:wml-tomo}, is that wave matrix Lindbladization is superior to the tomography method, in the sense that its sample complexity has no dimension dependence, while the sample complexity of the tomography method does. Let us note that a similar comparison was made between density matrix exponentiation and state tomography in \cite[Section~2]{Kimmel2017HamiltonianComplexity}, demonstrating the advantage of the former over the latter in terms of the sample complexity required for sample-based Hamiltonian simulation.

To conduct this comparison, let us consider a Lindbladian $\mathcal{L}$ with just one Lindblad operator $L$, similar to what we did in the previous two sections (see~\eqref{eq:lindbladian-single}). 
Our goal is to simulate the quantum channel $\e^{\mathcal{L}t}$ corresponding to~$\mathcal{L}$ for a time $t$. 
Let us suppose that the Lindblad operator $L$ is given upfront by encoding it in a program state $|\psi\rangle$, as in \eqref{eq:L-encode}, and that we have access to $n$~copies of this state.

As mentioned above, one naive way to implement the channel $\e^{\mathcal{L}t}$ is to first obtain the full classical description of the encoded Lindblad operator $L$ by using multiple copies of the state $|\psi\rangle$. 
By classical description, we mean the complete matrix representation of $L$. Once we have such a description of~$L$, we can then use known algorithms \cite{Childs2016EfficientDynamics, Cleve2016EfficientEvolution, Miessen2022QuantumDynamics, Schlimgen2022QuantumOperators, Suri2022Two-UnitarySimulation} to actually implement the channel $\e^{\mathcal{L}t}$ on a quantum computer. 
However, an important point to note here is that obtaining an exact classical description of an arbitrary operator encoded in a pure state may not be computationally efficient in general. 
Furthermore, obtaining an operator that is ``close" to the encoded operator is sufficient for many practical purposes.
As a result, we look into a slightly relaxed version of the above problem, in which the task is to output the full classical description of an operator $\tilde{L}$ that is $\delta$-close to $L$ in the Hilbert--Schmidt distance measure, given $n$ copies of $|\psi\rangle$.  We further require that the operator $\tilde{L}$ satisfies the constraint $\left\Vert \tilde L\right\Vert_{2} = 1$.  We refer to this modified relaxed problem as \textit{wave-matrix tomography}. More formally, the task is to output the full classical description of an operator $\tilde{L}$ such that $\left\Vert \tilde L\right\Vert_{2} = 1$ and
\begin{equation}
    \left \Vert \tilde{L} - L \right \Vert_{2} \leq \delta,
    \label{eq:tildeL-L}
\end{equation}
by using multiple copies of the state $|\psi\rangle$.

Having said that, a natural question in the context of Lindbladian simulation is how well the quantum channel $\e^{\tilde{\mathcal{L}}t}$ approximates the channel $\e^{\mathcal{L}t}$, where~$\tilde{\mathcal{L}}$ is the Lindbladian with a single Lindblad operator $\tilde{L}$. 
The following theorem states that if the two Lindblad operators $\tilde L$ and $L$ are $\delta = O(\varepsilon/t)$-close in Hilbert--Schmidt distance, then the quantum channel $\e^{\tilde{\mathcal{L}}t}$ approximates the target channel $\e^{\mathcal{L}t}$ up to $O(\varepsilon)$ error in diamond distance.

\begin{theorem}
\label{thm:comp-main-thm}
    Let $L$ and $\tilde L$ be two $d \times d$-dimensional linear operators such that $\left \Vert L \right \Vert_{2} = \left \Vert \tilde L \right \Vert_{2} = 1$ and 
    \begin{equation}
    \left \Vert \tilde{L} - L \right \Vert_{2} = O(\varepsilon/t).
    \label{eq:thm-tildeL-L}
\end{equation}
Then, the following holds:
\begin{equation}
     \frac{1}{2}\left\Vert \e^{\tilde{\mathcal{L}}t} -  \e^{\mathcal{L}t}\right\Vert_{\diamond} = O(\varepsilon),
\end{equation}
where $\tilde{\mathcal{L}}$ and $\mathcal{L}$ are Lindbladians with Lindblad operators $\tilde{L}$ and $L$, respectively.
\end{theorem}
\begin{proof}
    The proof is provided in Appendix~\ref{app:D}.
\end{proof}

\medskip
We now turn to examining the sample complexity of wave-matrix tomography in order to compare it with our WML approach.
The key question here is: how many copies of the state $|\psi\rangle$ are needed to solve the wave-matrix tomography problem?

To begin with, let us first look at the sample complexity of a related problem: pure-state tomography. 
In pure-state tomography, we are given $n$~copies of a state $|\psi\rangle$, and we want to output a  state $|\tilde{\psi}\rangle$ such that
\begin{equation}
    \frac{1}{2}\left \Vert |\tilde \psi\rangle\!\langle \tilde \psi| - |\psi\rangle\!\langle \psi| \right \Vert_{1} \leq \delta,
    \label{eq:tilde_psi-psi}
\end{equation}
where $\delta \in (0,1]$.
In order to solve the above problem, \cite[Theorem 3]{Haah2017} establishes that it is necessary to use $n$ copies of $|\psi\rangle$, where
\begin{equation}
\label{eq:numofsamplestomo}
    n = \Omega\!\left( \frac{d^2(1-\delta)^2}{\delta^{2}\log(d^2/\delta)}\right).
\end{equation}
In the above expression, $d^2$ is the dimension of the underlying Hilbert space of $|\psi\rangle$, and so we see that the sample complexity is dimension dependent.

Let us now connect pure-state tomography to wave-matrix tomography by reducing the former to the latter. 
In other words, we need to show that we can solve the pure-state tomography problem if we can solve the wave-matrix tomography problem.
The primary reason for showing this reduction is to obtain a lower bound on the sample complexity of the wave-matrix tomography problem.
Note that the input for both problems is the same, i.e., some number of copies of~$|\psi\rangle$.
On the contrary, their outputs are different.
The output of the wave-matrix tomography problem is an operator $\tilde L$ that is $\delta$-close to $L$ in Hilbert--Schmidt distance (see \eqref{eq:tildeL-L}),
while the output of pure-state tomography is a pure state $|\tilde{\psi}\rangle$ that is $\delta$-close to $|\psi\rangle$ in  normalized trace distance (see~\eqref{eq:tilde_psi-psi}).
We now relate these outputs as follows. Let $|\tilde \psi\rangle$ be the following state:
\begin{equation}
    |\tilde \psi\rangle = (\tilde L\otimes I) |\Gamma\rangle.
\end{equation}
Then Lemma~\ref{lem:HS-distance-lem} below states that $|\tilde{\psi}\rangle$ is $\delta$-close to $|\psi\rangle$ if $L$ is $\delta$-close to $\tilde{L}$ in Hilbert--Schmidt distance. As such, $|\tilde{\psi}\rangle$ is a solution to the pure-state tomography problem when the input is $|\psi\rangle$.
This concludes the reduction of pure-state tomography to wave-matrix tomography. 
This further implies that the lower bound on the sample complexity of the pure-state tomography problem, given by \eqref{eq:numofsamplestomo}, is also a lower bound on the sample complexity of the wave-matrix tomography problem.

\begin{lemma}
\label{lem:HS-distance-lem}
 Let $L$ and $\tilde L$ be two $d \times d$-dimensional linear operators such that $\left \Vert L \right \Vert_{2} = \left \Vert \tilde L \right \Vert_{2} = 1$ and 
    \begin{equation}
    \left \Vert \tilde{L} - L \right \Vert_{2} \leq \delta.
    \label{eq:lemma-tildeL-L}
\end{equation}
Then,
\begin{equation}
    \frac{1}{2}\left \Vert |\tilde \psi\rangle\!\langle \tilde \psi| - |\psi\rangle\!\langle \psi| \right \Vert_{1} \leq \delta,
\end{equation}
where $|\tilde \psi\rangle = (\tilde{L} \otimes I) |\Gamma\rangle$ and $|\psi\rangle =  (L \otimes I) |\Gamma\rangle$ are  pure states that encode $\tilde L$ and~$L$, respectively.
\end{lemma}

\begin{proof}
    See Appendix~\ref{app:HS-distance-lem}.
\end{proof}
\medskip 

\begin{figure}
    \centering
    \includegraphics[scale=0.16]{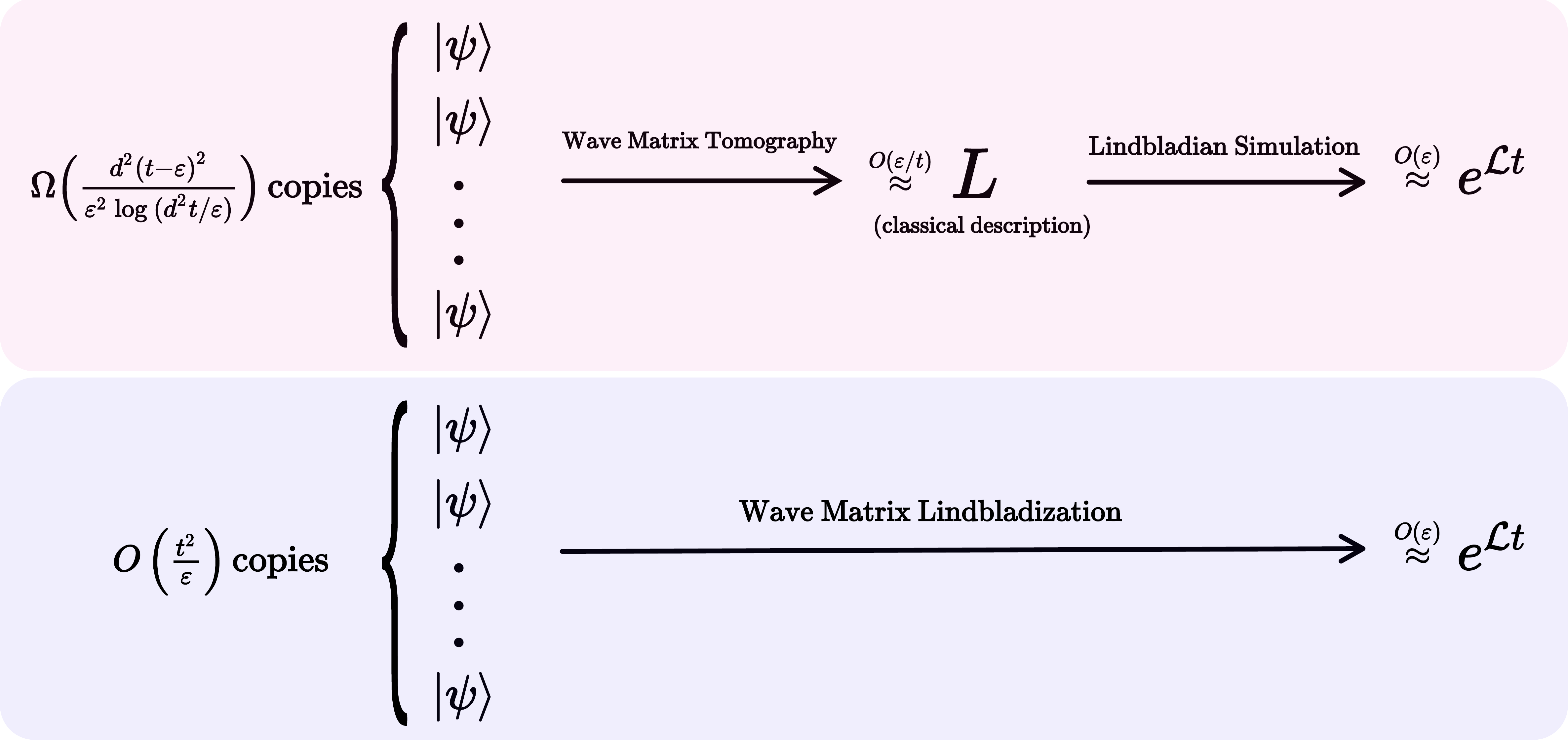}
    \caption{Comparing sample complexities of two different approaches for approximately simulating the channel $e^{\mathcal{L}t}$: (pink box) performing wave-matrix tomography to first obtain the classical description of the Lindblad operator~$L$, followed by simulation from this description and (blue box) wave matrix Lindbladization.}
    \label{fig:wml-tomo}
\end{figure}

According to Theorem~\ref{thm:comp-main-thm}, the quantum channel $\e^{\tilde{\mathcal{L}}t}$ approximates the target channel $\e^{\mathcal{L}t}$ up to $O(\varepsilon)$ error in diamond distance if $\delta = O(\varepsilon/t)$.
Substituting this into \eqref{eq:numofsamplestomo}, we get
\begin{equation}
    n = \Omega\!\left( \frac{d^2(t-\varepsilon)^{2}}{\varepsilon^{2}\log{(d^2t/\varepsilon})}\right).
\end{equation}
It is evident from above that wave-matrix tomography requires more copies of~$|\psi\rangle$ to implement the channel $\e^{\mathcal{L}t}$ than wave matrix Lindbladization. 
As shown in \cite[Theorem 1]{Patel2023}, our approach needs only $O(t^2/\varepsilon)$ copies of $|\psi\rangle$, with no dependence on the dimension $d$. 
Please refer to Figure~\ref{fig:wml-tomo} for a visual representation of the comparison between these two approaches.
This favorable scaling highlights an advantage over wave-matrix tomography, where sample complexity grows with increasing~$d$. 
Our quantum algorithm thus provides an efficient route for implementing~$\e^{\mathcal{L}t}$ without needing full tomography of~$L$.

The development above indicates that a party could send sufficiently many program states to another party in order to simulate the Lindblad dynamics, without revealing details about the encoded operator.
This provides a form of quantum copy-protection, as introduced in~\cite{Aaronson2009}, in which the quantum operation is the Lindbladian evolution of a quantum state.
Our results also offer a potential resolution to the open question posed in \cite{Kimmel2017HamiltonianComplexity} regarding what other quantum operations, besides Hamiltonian simulation, could be encoded in quantum states and executed privately.

\section{Conclusion and Open Problems}

In this paper, we focused on the problem of simulating open system dynamics governed by the Lindblad master equation when the Lindblad operators are provided beforehand through encoding in program states. 
In our prequel paper \cite{Patel2023}, we investigated a relatively simple case in which the Lindbladian consists of a single Lindblad operator and a Hamiltonian. 
In the present sequel paper, we extended this case to simulating general Lindbladians and other situations in which the Lindblad operator can be expressed as a linear combination or polynomial of the operators encoded across multiple program states. 
We proposed quantum algorithms for all these cases and investigated their sample complexities -- the number of copies of program states needed.
Finally, we showed that our quantum algorithms provide a more efficient route to simulating Lindbladian evolution when compared to full tomography of the encoded operators. We demonstrated this by proving that the sample complexity for tomography is dependent on the system dimension, while that for wave matrix Lindbladization is not.

In our prequel paper \cite{Patel2023}, we listed some open questions that we then managed to answer in this sequel. However, some open problems still remain unaddressed. For completeness, we recap those outstanding questions from our prequel paper, which also apply to the algorithms presented here:
\begin{itemize}
    \item An important step common to the quantum algorithms proposed in this paper and our prequel involves simulating Lindbladian evolution for a short time period using a predefined Lindbladian. 
    For instance, this is seen in Algorithm~1 where the Lindbladian $\mathcal{M}$ is simulated in the first loop's Step 1. 
    It is essential to investigate whether efficient implementations can be developed for these Lindbladian channels that need to be simulated as subroutines. 
    This will undoubtedly resolve the time or gate complexities of our algorithms.
    \item Here we analyzed the sample complexities of the proposed quantum algorithms, quantifying the number of copies of program states needed. 
    A natural next step would be investigating whether these sample complexities are optimal, meaning whether there exist matching lower bounds.
\end{itemize}

\section{Acknowledgements}

We are indebted to Andrew Childs, Tongyang Li, and Marina Radulaski for a helpful email exchange regarding \cite[Proposition~2]{Childs2016EfficientDynamics}. MMW is especially grateful to Prof.~Ingemar Bengtsson for the opportunity to have met and discussed research with Prof.~Lindblad in Stockholm, Sweden, during April~2019.

\appendix

\section{Error Analysis for Trotter-like Approach}

\label{app:error-analysis-gen}

The Lindbladian $\mathcal{L}$ given by \eqref{eq:lindbladmaster-gen} can be written as a linear combination of single-operator Lindbladians $\mathcal{L}_{1}, \ldots, \mathcal{L}_{r}$ with $r \coloneqq J+K$. In other words, we have:
\begin{multline}
    \mathcal{L} (\rho) = \underbrace{-i [c_{1}\sigma_{1}, \rho]}_{\mathcal{L}_{1}(\rho)} + \cdots +  \underbrace{(-i [c_{J}\sigma_{J}, \rho])}_{\mathcal{L}_{J}(\rho)} + \underbrace{L_{1}\rho L_{1}^{\dagger} - \frac{1}{2} \left \{L_{1}^{\dagger}L_{1}, \rho \right\}}_{\mathcal{L}_{J+1}(\rho)} + \cdots \\ + \underbrace{L_{K}\rho L_{K}^{\dagger} - \frac{1}{2} \left \{L_{K}^{\dagger}L_{K}, \rho \right\}}_{\mathcal{L}_{J+K}(\rho)}.
\end{multline}

We first consider the two cases below.
\begin{itemize}
    \item \textbf{Case 1:} The Lindbladian $\mathcal{L}_{i}$ consists of the Hamiltonian $c_{i}\sigma_{i}$ where $i \in \{1, \ldots, J\}$. This corresponds to the iteration with program state~$\sigma_{i}$ in the second and third loops of Algorithm~2. The output of this iteration is:
    \begin{align}
    &\operatorname{Tr}_{2}\!\left[\e^{\mathcal{N}\Delta'_{i}}(\rho \otimes \sigma_{i}) \right] \notag\\
    & = \rho + \operatorname{Tr}_{2}\!\left[ \mathcal{N}\left( \rho \otimes \sigma_{i}\right)\right]\Delta'_{i} + \frac{1}{2} \operatorname{Tr}_{2}\!\left[ (\mathcal{N} \circ \mathcal{N})\left( \rho \otimes \sigma_{i}\right)\right]\Delta'^2_{i} +  O(\Delta'^3_{i})\label{eq:app-pf-case-1-1st}\\
    & = \rho + \frac{\mathcal{L}_{i}(\rho)}{|c_{i}|} \Delta'_{i} + \frac{1}{2} \operatorname{Tr}_{2}\!\left[ (\mathcal{N} \circ \mathcal{N})\left( \rho \otimes \sigma_{i}\right)\right]\Delta'^2_{i} +  O(\Delta'^3_{i})\\
    & = \rho + \mathcal{L}'_{i} (\rho)\Delta + \frac{1}{2} \mathcal{J}_{i}(\rho)\Delta^2 + O(\Delta^3)\label{eq:thm-p-trotter-hamil}.
    \end{align}
    For intuition about the second equality, see the tensor-network diagrams in Figure~7 of \cite{Patel2023}, which provide a visual representation. The final equality follows by substituting
    \begin{align}
        \Delta'_{i} & = \frac{|c_{i}|t}{\left \Vert \mathcal{L} \right \Vert_{\max} n} = \frac{|c_{i}|}{\left \Vert \mathcal{L} \right \Vert_{\max} }\Delta,
    \end{align}
    where we set $\Delta \coloneqq t/n$, and defining
    \begin{align}
        \mathcal{L}'_{i} (\rho) & \coloneqq \frac{\mathcal{L}_{i} (\rho)}{\left \Vert \mathcal{L} \right \Vert_{\max}},\\
        \mathcal{J}_{i}(\rho) & \coloneqq  \operatorname{Tr}_{2}\!\left[ (\mathcal{N} \circ \mathcal{N})\left( \rho \otimes \sigma_{i}\right)\right] \frac{|c_{i}|^{2}}{\left \Vert \mathcal{L} \right \Vert_{\max}^2}. 
    \end{align}
    
    \item \textbf{Case 2:} The Lindbladian $\mathcal{L}_{i}$ consists of the Lindblad operator $L_{i - J}$ where $i \in \{J+1, \ldots, J+K\}$. This corresponds to the iteration with program state $\psi_{i-J}$ in the first and fourth loops of Algorithm~2. The output of this iteration is:
    \begin{align}
    &\operatorname{Tr}_{23}\!\left[\e^{\mathcal{M}\Delta_{i-J}}(\rho \otimes \psi_{i-J}) \right] \notag\\
    & = \rho + \operatorname{Tr}_{23}\!\left[ \mathcal{M}\left( \rho \otimes \psi_{i-J}\right)\right]\Delta_{i-J} + \frac{1}{2} \operatorname{Tr}_{23}\!\left[ (\mathcal{M} \circ \mathcal{M})\left( \rho \otimes \psi_{i-J}\right)\right]\Delta^2_{i-J} \notag\\
    & \hspace{8cm}+ O(\Delta^3_{i-J})\label{eq:app-pf-case-2-1st}\\
    & = \rho + \frac{\mathcal{L}_{i} (\rho)}{\left\Vert L_{i-J} \right \Vert_{2}^2}\Delta_{i-J} + \frac{1}{2} \operatorname{Tr}_{23}\!\left[ (\mathcal{M} \circ \mathcal{M})\left( \rho \otimes \psi_{i-J}\right)\right]\Delta^2_{i-J} \notag\\
    & \hspace{8cm}+ O(\Delta^3_{i-J})\\
    & = \rho + \mathcal{L}'_{i} (\rho)\Delta + \frac{1}{2} \mathcal{J}_{i}(\rho)\Delta^2 + O(\Delta^3)\label{eq:thm-p-trotter-lindblad}.
    \end{align}
    For understanding the second equality, refer to the tensor-network diagrams in Figures~3, 5, and 6 of \cite{Patel2023}. The final equality follows by substituting
    \begin{align}
        \Delta_{i-J} & = \frac{\left \Vert L_{i-J} \right \Vert_{2}^{2}t}{\left \Vert \mathcal{L} \right \Vert_{\max} n} = \frac{\left \Vert L_{i-J} \right \Vert_{2}^{2}}{\left \Vert \mathcal{L} \right \Vert_{\max} }\Delta,
    \end{align}
    and defining
    \begin{align}
        \mathcal{L}'_{i} (\rho) & \coloneqq \frac{\mathcal{L}_{i} (\rho)}{\left \Vert \mathcal{L} \right \Vert_{\max}},\\
        \mathcal{J}_{i}(\rho) & \coloneqq  \operatorname{Tr}_{23}\!\left[( \mathcal{M} \circ \mathcal{M})\left( \rho \otimes \psi_{i-J}\right)\right] \frac{\left \Vert L_{i-J} \right \Vert_{2}^{4}}{\left \Vert \mathcal{L} \right \Vert_{\max}^2}. 
        \label{eq:error-an-lindblad-prime}
    \end{align}
\end{itemize}

In what follows and based on the above, we adopt the same notation $\mathcal{L}'_i$ and $\mathcal{J}_i$ for each case, with the index $i \in \{1, \ldots, J+K \}$ keeping track of which case actually applies. From \eqref{eq:thm-p-trotter-hamil}, \eqref{eq:thm-p-trotter-lindblad}, and the flow of Algorithm~2, we can now see that Algorithm~2 applies $\e^{\mathcal{L}_{1}'\Delta}$ approximately in the end, preceded by applying $\e^{\mathcal{L}_{2}'\Delta}$ approximately, which is preceded by applying $\e^{\mathcal{L}_{3}'\Delta}$ approximately, etc. Then this sequential process is reversed and repeated $n$~times.

Let us begin with some initial observations, and we first analyze the approach under the assumption that the order is $r, \ldots, 1, r, \ldots, 1, \ldots, r, \ldots, 1$, for a total of $rn$ steps (this choice of ordering amounts to excluding Steps~3 and 4 in Algorithm~2). After presenting the error analysis under this assumption in \eqref{eq:1st-step-big-proof}--\eqref{eq:last-step-big-proof} below, we then consider the order $r, \ldots, 1, 1, \ldots, r,$ $r, \ldots, 1, 1, \ldots, r,$ $ \ldots, r,\ldots, 1, 1, \ldots, r$, for a total of $2rn$ steps, as presented in Algorithm~2 in the main text, and we observe how the extra symmetry of this order results in a significant cancellation of terms, thus leading to the claimed performance stated in Theorem~\ref{thm:trotter-approach}.

Before beginning the analysis, consider the following general lemma:
\begin{lemma}
\label{lem:trotter-induction-analysis}
Suppose that $\rho$ is an initial state and $f:\{1, \ldots, m\} \to \{1, \ldots, r\}$ is an arbitrary function. Suppose further that $m$ steps of Case~1 or~2 above are applied in the order $f(m), f(m-1), \ldots, f(2), f(1)$, with corresponding superoperators $(\mathcal{L}'_{f(i)})_{i=1}^m$ and $(\mathcal{J}_{f(i)})_{i=1}^m$. Then the resulting state is as follows:
\begin{multline}
    \rho+\sum_{i=1}^{m}\mathcal{L}'_{f(i)}(\rho)\Delta+\frac{1}{2}%
\sum_{i=1}^{m}\mathcal{J}_{f(i)}(\rho)\Delta^{2}\\
+\sum_{i,j=1:i<j}^{m}\left(
\mathcal{L}'_{f(i)}\circ\mathcal{L}'_{f(j)}\right)  (\rho)\Delta^{2}+O(\Delta^{3}).
\label{eq:trotter-induction-analysis}
\end{multline}
\end{lemma}

\begin{proof}
We employ a proof by induction, starting with the base step. Suppose that the state at step~$k$ is $\rho^{(k)} \approx e^{\mathcal{L}_{f(1)}'\Delta}(\rho^{k-1})$. Then it follows from \eqref{eq:app-pf-case-1-1st}--\eqref{eq:thm-p-trotter-hamil} and \eqref{eq:app-pf-case-2-1st}--\eqref{eq:thm-p-trotter-lindblad} that
\begin{equation}
\rho^{(k)}=\rho^{(k-1)}+\mathcal{L}'_{f(1)}(\rho^{(k-1)})\Delta+\frac{1}%
{2}\mathcal{J}_{f(1)}(\rho^{(k-1)})\Delta^{2}+O(\Delta^{3}).
\label{eq:1st-step-big-proof}
\end{equation}
Then considering that%
\begin{equation}
\rho^{(k-1)}=\rho^{(k-2)}+\mathcal{L}'_{f(2)}(\rho^{(k-2)})\Delta+\frac{1}%
{2}\mathcal{J}_{f(2)}(\rho^{(k-2)})\Delta^{2}+O(\Delta^{3}),
\end{equation}
we find that%
\begin{multline}
    \rho^{(k)}=\rho^{(k-2)}+\sum_{i=1}^{2}\mathcal{L}'_{f(i)}(\rho^{(k-2)}%
)\Delta+\frac{1}{2}\sum_{i=1}^{2}\mathcal{J}_{f(i)}(\rho^{(k-2)})2\Delta
^{2}\\ +\left(  \mathcal{L}'_{f(1)}\circ\mathcal{L}'_{f(2)}\right)  (\rho^{(k-2)}%
)\Delta^{2}+O(\Delta^{3}),
\end{multline}
because
\begin{align}
& \rho^{(k)} \notag \\
&  =\rho^{(k-1)}+\mathcal{L}'_{f(1)}(\rho^{(k-1)})\Delta+\frac{1}%
{2}\mathcal{J}_{f(1)}(\rho^{(k-1)})\Delta^{2}+O(\Delta^{3})\\
&  =\rho^{(k-2)}+\mathcal{L}'_{f(2)}(\rho^{(k-2)})\Delta+\frac{1}{2}%
\mathcal{J}_{f(2)}(\rho^{(k-2)})\Delta^{2}+O(\Delta^{3})\notag\\
&  \qquad+\mathcal{L}'_{f(1)}\left(  \rho^{(k-2)}+\mathcal{L}'_{f(2)}(\rho
^{(k-2)})\Delta+\frac{1}{2}\mathcal{J}_{f(2)}(\rho^{(k-2)})\Delta^{2}%
+O(\Delta^{3})\right)  \Delta\notag\\
&  \qquad+\frac{1}{2}\mathcal{J}_{f(1)}\left(  \rho^{(k-2)}+\mathcal{L}'_{f(2)}%
(\rho^{(k-2)})\Delta+\frac{1}{2}\mathcal{J}_{f(2)}(\rho^{(k-2)})\Delta
^{2}+O(\Delta^{3})\right)  \Delta^{2}\\
&  =\rho^{(k-2)}+\mathcal{L}'_{f(2)}(\rho^{(k-2)})\Delta+\frac{1}{2}%
\mathcal{J}_{f(2)}(\rho^{(k-2)})\Delta^{2}\notag\\
&  \qquad+\mathcal{L}'_{f(1)}\left(  \rho^{(k-2)}\right)  \Delta+(\mathcal{L}'%
_{f(1)}\circ\mathcal{L}'_{f(2)})(\rho^{(k-2)})\Delta^{2}\notag \\
& \qquad +\frac{1}{2}\mathcal{J}%
_{f(1)}\left(  \rho^{(k-2)}\right)  \Delta^{2}+O(\Delta^{3})\\
&  =\rho^{(k-2)}+\sum_{i=1}^{2}\mathcal{L}'_{f(i)}(\rho^{(k-2)})\Delta+\frac{1}%
{2}\sum_{i=1}^{2}\mathcal{J}_{f(i)}(\rho^{(k-2)})\Delta^{2}
\notag\\
& \hspace{4.6cm}+(\mathcal{L}'_{f(1)}%
\circ\mathcal{L}'_{f(2)})(\rho^{(k-2)})\Delta^{2}+O(\Delta^{3}).
\end{align}
We can repeat this again to find that%
\begin{multline}
\rho^{(k)}=\rho^{(k-3)}+\sum_{i=1}^{3}\mathcal{L}'_{f(i)}(\rho^{(k-3)}%
)\Delta+\frac{1}{2}\sum_{i=1}^{2}\mathcal{J}_{f(i)}(\rho^{(k-3)})\Delta^{2}%
\\
+\sum_{i,j=1:i<j}^{3}\left(  \mathcal{L}'_{f(i)}\circ\mathcal{L}'_{f(j)}\right)
(\rho^{(k-3)})\Delta^{2}+O(\Delta^{3}),
\end{multline}
because%
\begin{align}
& \rho^{(k)} \notag \\
&  =\rho^{(k-2)}+\sum_{i=1}^{2}\mathcal{L}'_{f(i)}(\rho^{(k-2)}%
)\Delta+\frac{1}{2}\sum_{i=1}^{2}\mathcal{J}_{f(i)}(\rho^{(k-2)})\Delta
^{2}\notag \\
& \hspace{4cm} +\left(  \mathcal{L}'_{f(1)}\circ\mathcal{L}'_{f(2)}\right)  (\rho^{(k-2)}%
)\Delta^{2}+O(\Delta^{3})\\
&  =\rho^{(k-3)}+\mathcal{L}'_{f(3)}(\rho^{(k-3)})\Delta+\frac{1}{2}%
\mathcal{J}_{f(3)}(\rho^{(k-3)})\Delta^{2}+O(\Delta^{3})\notag\\
&  \quad+\sum_{i=1}^{2}\mathcal{L}'_{f(i)}\left(  \rho^{(k-3)}+\mathcal{L}'%
_{f(3)}(\rho^{(k-3)})\Delta+\frac{1}{2}\mathcal{J}_{f(3)}(\rho^{(k-3)})\Delta
^{2}+O(\Delta^{3})\right)  \Delta\notag\\
&  \quad+\frac{1}{2}\sum_{i=1}^{2}\mathcal{J}_{f(i)}\left(  \rho^{(k-3)}%
+\mathcal{L}'_{f(3)}(\rho^{(k-3)})\Delta+\frac{1}{2}\mathcal{J}_{f(3)}(\rho
^{(k-3)})\Delta^{2}+O(\Delta^{3})\right)  \Delta^{2}\notag\\
&  \quad+\left(  \mathcal{L}'_{f(1)}\circ\mathcal{L}'_{f(2)}\right)  \left(
\rho^{(k-3)}+\mathcal{L}'_{f(3)}(\rho^{(k-3)})\Delta+\frac{1}{2}\mathcal{J}%
_{f(3)}(\rho^{(k-3)})\Delta^{2}+O(\Delta^{3})\right)  \Delta^{2}\\
&  =\rho^{(k-3)}+\mathcal{L}'_{f(3)}(\rho^{(k-3)})\Delta+\frac{1}{2}%
\mathcal{J}_{f(3)}(\rho^{(k-3)})\Delta^{2}\notag\\
&  \quad+\sum_{i=1}^{2}\mathcal{L}'_{f(i)}\left(  \rho^{(k-3)}\right)
\Delta+\sum_{i=1}^{2}\left(  \mathcal{L}'_{f(i)}\circ\mathcal{L}'_{f(3)}\right)
(\rho^{(k-3)})\Delta^{2}\notag\\
&  \quad+\frac{1}{2}\sum_{i=1}^{2}\mathcal{J}_{f(i)}\left(  \rho^{(k-3)}\right)
\Delta^{2}+\left(  \mathcal{L}'_{f(1)}\circ\mathcal{L}'_{f(2)}\right)  \left(
\rho^{(k-3)}\right)  \Delta^{2}+O(\Delta^{3})\\
&  =\rho^{(k-3)}+\sum_{i=1}^{3}\mathcal{L}'_{f(i)}(\rho^{(k-3)})\Delta+\frac{1}%
{2}\sum_{i=1}^{2}\mathcal{J}_{f(i)}(\rho^{(k-3)})\Delta^{2}\notag\\
& \hspace{3cm}+\sum_{i,j=1:i<j}%
^{3}\left(  \mathcal{L}'_{f(i)}\circ\mathcal{L}'_{f(j)}\right)  (\rho^{(k-3)}%
)\Delta^{2}+O(\Delta^{3}).
\end{align}

Having established the base step, we now consider the inductive step and claim that after applying in the order $f(m), f(m-1), \ldots, f(2), f(1)$ (as in the lemma statement), the state after $m$
steps is given by 
\begin{multline}
\rho^{(k)}=\rho^{(k-m)}+\sum_{i=1}^{m}\mathcal{L}'_{f(i)}(\rho^{(k-m)}%
)\Delta+\frac{1}{2}\sum_{i=1}^{m}\mathcal{J}_{f(i)}(\rho^{(k-m)})\Delta^{2}%
\\ +\sum_{i,j=1:i<j}^{m}\left(  \mathcal{L}'_{f(i)}\circ\mathcal{L}'_{f(j)}\right)
(\rho^{(k-m)})\Delta^{2}+O(\Delta^{3})
\label{eq:inductive-step-to-prove}.
\end{multline}
To establish the inductive step, suppose that
\begin{equation}
\rho^{(k-m)}=\rho^{(k-m-1)}+\mathcal{L}'_{f(m+1)}(\rho^{(k-m-1)})\Delta+\frac
{1}{2}\mathcal{J}_{f(m+1)}(\rho^{(k-m-1)})\Delta^{2}+O(\Delta^{3}).
\end{equation}
Then, by starting with \eqref{eq:inductive-step-to-prove}, we find that
\begin{align}
& \rho^{(k)}   \notag \\
& =\rho^{(k-m)}+\sum_{i=1}^{m}\mathcal{L}'_{f(i)}(\rho^{(k-m)}%
)\Delta+\frac{1}{2}\sum_{i=1}^{m}\mathcal{J}_{f(i)}(\rho^{(k-m)})\Delta^{2}%
\notag \\
& \quad +\sum_{i,j=1:i<j}^{m}\left(  \mathcal{L}'_{f(i)}\circ\mathcal{L}'_{f(j)}\right)
(\rho^{(k-m)})\Delta^{2} +O(\Delta^{3}) \\
&  =\rho^{(k-m-1)}+\mathcal{L}'_{f(m+1)}(\rho^{(k-m-1)})\Delta+\frac{1}%
{2}\mathcal{J}_{f(m+1)}(\rho^{(k-m-1)})\Delta^{2} +O(\Delta^{3})\notag\\
&  \quad+\sum_{i=1}^{m}\mathcal{L}'_{f(i)}\left(  \rho^{(k-m-1)}+\mathcal{L}'%
_{f(m+1)}(\rho^{(k-m-1)})\Delta+\frac{1}{2}\mathcal{J}_{f(m+1)}(\rho^{(k-m-1)}%
)\Delta^{2}\right)  \Delta \notag\\
&  \quad+\frac{1}{2}\sum_{i=1}^{m}\mathcal{J}_{f(i)}\left(  \rho^{(k-m-1)}%
+\mathcal{L}'_{f(m+1)}(\rho^{(k-m-1)})\Delta+\frac{1}{2}\mathcal{J}_{f(m+1)}%
(\rho^{(k-m-1)})\Delta^{2}\right)  \Delta^{2}\notag\\
&  \,  +\sum_{\substack{i,j=1:\\i<j}}^{m}\left(  \mathcal{L}'_{f(i)}\circ\mathcal{L}''%
_{j}\right)  \left(  \rho^{(k-m-1)}+\mathcal{L}'_{f(m+1)}(\rho^{(k-m-1)}%
)\Delta+\frac{1}{2}\mathcal{J}_{f(m+1)}(\rho^{(k-m-1)})\Delta^{2}\right)
\Delta^{2} \\
&  =\rho^{(k-m-1)}+\mathcal{L}'_{f(m+1)}(\rho^{(k-m-1)})\Delta+\frac{1}%
{2}\mathcal{J}_{f(m+1)}(\rho^{(k-m-1)})\Delta^{2}\notag\\
&  \quad+\sum_{i=1}^{m}\mathcal{L}'_{f(i)}\left(  \rho^{(k-m-1)}\right)
\Delta+\sum_{i=1}^{m}\left(  \mathcal{L}'_{f(i)}\circ\mathcal{L}'_{f(m+1)}\right)
(\rho^{(k-m-1)})\Delta^{2}\notag\\
&  \quad+\frac{1}{2}\sum_{i=1}^{m}\mathcal{J}_{f(i)}\left(  \rho^{(k-m-1)}%
\right)  \Delta^{2}\notag \\
& \quad +\sum_{i,j=1:i<j}^{m}\left(  \mathcal{L}'_{f(i)}\circ
\mathcal{L}'_{f(j)}\right)  \left(  \rho^{(k-m-1)}\right)  \Delta^{2}+O(\Delta
^{3})\\
&  =\rho^{(k-m-1)}+\sum_{i=1}^{m+1}\mathcal{L}'_{f(i)}(\rho^{(k-m)})\Delta
+\frac{1}{2}\sum_{i=1}^{m+1}\mathcal{J}_{f(i)}(\rho^{(k-m)})\Delta^{2}%
\notag \\
& \hspace{3cm} +\sum_{i,j=1:i<j}^{m+1}\left(  \mathcal{L}'_{f(i)}\circ\mathcal{L}'_{f(j)}\right)
(\rho^{(k-m)})\Delta^{2}+O(\Delta^{3}).
\end{align}
This concludes the proof of the inductive step.
Setting $k=m$ and $\rho^{(0)}=\rho$ in \eqref{eq:inductive-step-to-prove} gives the claim in \eqref{eq:trotter-induction-analysis}.
\end{proof}
\medskip

Now we apply Lemma~\ref{lem:trotter-induction-analysis} to the simulation order $r,\ldots,1,r,\ldots
,1,r,\ldots,1$, with the goal of approximating the following evolution:
\begin{equation}
\e^{\mathcal{L}'n\Delta}(\rho)=\rho+\sum_{i=1}^{r}\mathcal{L}'_{i}(\rho
)n\Delta+\frac{1}{2}\sum_{i,j=1}^{r}\left(  \mathcal{L}'_{i}\circ
\mathcal{L}'_{j}\right)  \left(  \rho\right)  n^{2}\Delta^{2}+\ldots \, .
\label{eq:trotter-sim-target-pf}
\end{equation}
Let us suppose that $L_{0} \coloneqq L_{r}$ for convenience. 
Also, let $\mathcal{L}''_{i} \coloneqq \mathcal{L}'_{i\!\Mod{r}} $ and $\mathcal{J}''_{i} \coloneqq \mathcal{J}_{i\!\Mod{r}} $ for all $i \in \{1, \ldots, rn\}$.
Then by applying Lemma~\ref{lem:trotter-induction-analysis}, we find that the state after
$rn$ steps of the simulation is given by
\begin{equation}
\rho+\sum_{i=1}^{rn}\mathcal{L}''_{i}(\rho)\Delta+\frac{1}{2}%
\sum_{i=1}^{rn}\mathcal{J}''_{i}(\rho)\Delta^{2}+\sum_{i,j=1:i<j}^{rn}\left(
\mathcal{L}''_{i}\circ\mathcal{L}''_{j}\right)  (\rho)\Delta^{2}+O(\Delta^{3}).
\label{eq:expansion-rho-n-app}
\end{equation}
As such, the zeroth and first order terms in \eqref{eq:trotter-sim-target-pf} and \eqref{eq:expansion-rho-n-app} are equal and thus cancel when subtracting \eqref{eq:expansion-rho-n-app} from \eqref{eq:trotter-sim-target-pf}.
To determine what happens when subtracting the second-order terms, we note that
\begin{equation}
    \sum_{i=1}^{rn}\mathcal{J}''_{i}(\rho)\Delta^{2} = O(rn\Delta^2),
    \label{eq:J-term-order-pf}
\end{equation}
and thus it remains to analyze the term $\sum_{i,j=1:i<j}^{rn}
\mathcal{L}''_{i}\circ\mathcal{L}''_{j} $ in \eqref{eq:expansion-rho-n-app}, with this choice of ordering, and compare it to the second-order term in \eqref{eq:trotter-sim-target-pf}.
In what follows, we prove that
\begin{equation}
\sum_{i,j=1:i<j}^{rn}\left(  \mathcal{L}''_{i}\circ\mathcal{L}''_{j}\right)
(\rho)\Delta^{2}-\left(  \frac{1}{2}\sum_{i,j=1}^{r}\left(  \mathcal{L}'%
_{i}\circ\mathcal{L}'_{j}\right)  \left(  \rho\right)  n^{2}\Delta^{2}\right)
=O(r^2n\Delta^{2})\label{eq:i<j-n2-terms}.
\end{equation}

After substituting $\mathcal{L}''_{i} = \mathcal{L}'_{i\!\Mod{r}} $ for all $i \in \{1, \ldots, rn\}$ in the first term of the above equation, we decompose it in the following way:
\begin{multline}
     \sum_{i,j=1:i<j}^{rn}\left(  \mathcal{L}'_{i\!\Mod{r}}\circ\mathcal{L}'_{j\!\Mod{r}}\right)
(\rho)\Delta^{2}\\
 = \sum_{i,j=1}^{rn}\left(  \mathcal{L}'_{i\!\Mod{r}}\circ\mathcal{L}'_{j\!\Mod{r}}\right)
(\rho)\Delta^{2} - \sum_{i,j=1:i>j}^{rn}\left(  \mathcal{L}'_{i\!\Mod{r}}\circ\mathcal{L}'_{j\!\Mod{r}}\right)
(\rho)\Delta^{2} \\ 
 - \sum_{i,j=1:i=j}^{rn}\left( \mathcal{L}'_{i\!\Mod{r}}\circ\mathcal{L}'_{j\!\Mod{r}}\right)
(\rho)\Delta^{2}.\label{eq:error-gen-decompose}
\end{multline}
The first term on the right-hand side of the above equation can be rewritten as:
\begin{align}
    & \sum_{i,j=1}^{rn}\left(  \mathcal{L}'_{i\!\Mod{r}}\circ\mathcal{L}'_{j\!\Mod{r}}\right)
(\rho)\Delta^{2}\notag \\
& = \sum_{i=1}^{r} \Bigg( \sum_{j=1}^{r} \left(  \mathcal{L}'_{i\!\Mod{r}}\circ\mathcal{L}'_{j\!\Mod{r}}\right)(\rho)\Delta^{2}  + \cdots \\
& \hspace{4cm}+ \sum_{j=r(n-1)+1}^{rn} \left(  \mathcal{L}'_{i\!\Mod{r}}\circ\mathcal{L}'_{j\!\Mod{r}}\right)(\rho)\Delta^{2} \Bigg) + \cdots \notag \\
& \qquad  + \sum_{i=r(n-1)+1}^{rn} \Bigg( \sum_{j=1}^{r} \left(  \mathcal{L}'_{i\!\Mod{r}}\circ\mathcal{L}'_{j\!\Mod{r}}\right) (\rho)\Delta^{2}  + \cdots \notag \\
& \hspace{4cm }+ \sum_{j=r(n-1)+1}^{rn} \left(  \mathcal{L}'_{i\!\Mod{r}}\circ\mathcal{L}'_{j\!\Mod{r}}\right) (\rho)\Delta^{2} \Bigg)\\
& = \sum_{i=1}^{r} \left(n \sum_{j=1}^{r} \left(  \mathcal{L}'_{i\!\Mod{r}}\circ\mathcal{L}'_{j\!\Mod{r}}\right) (\rho)\Delta^{2} \right) + \cdots \notag \\
& \hspace{2.2cm}  + \sum_{i=r(n-1)+1}^{rn} \left(n \sum_{j=1}^{r} \left(  \mathcal{L}'_{i\!\Mod{r}}\circ\mathcal{L}'_{j\!\Mod{r}}\right) (\rho)\Delta^{2} \right)\\
& =  n \sum_{j=1}^{r} \Bigg( \sum_{i=1}^{r} \left(  \mathcal{L}'_{i\!\Mod{r}}\circ\mathcal{L}'_{j\!\Mod{r}}\right) (\rho)\Delta^{2} + \cdots \notag \\
& \hspace{4cm}+  \sum_{i=r(n-1)+1}^{rn}  \left( \mathcal{L}'_{i\!\Mod{r}}\circ\mathcal{L}'_{j\!\Mod{r}}\right) (\rho)\Delta^{2} \Bigg)\\
& =  n \sum_{j=1}^{r} \left( n \sum_{i=1}^{r} \left(  \mathcal{L}'_{i\!\Mod{r}}\circ\mathcal{L}'_{j\!\Mod{r}}\right) (\rho)\Delta^{2} \right)\\
& =   \sum_{i, j=1}^{r}\left(  \mathcal{L}'_{i\!\Mod{r}}\circ\mathcal{L}'_{j\!\Mod{r}}\right) (\rho) n^2\Delta^{2}\label{eq:first-term-simplify} .
\end{align}
Now, the second term on the right-hand side of \eqref{eq:error-gen-decompose} and the expression on its left-hand side has the following relation:
\begin{align}
    & \sum_{\substack{i,j=1:\\i<j}}^{rn}\left(  \mathcal{L}'_{i\!\Mod{r}}\circ\mathcal{L}'_{j\!\Mod{r}}\right)
(\rho)\Delta^{2} +\sum_{\substack{i,j=1 :\\(j,i)\in \mathcal{T}}}^{rn} \left(  \mathcal{L}'_{i\!\Mod{r}}\circ\mathcal{L}'_{j\!\Mod{r}}\right)
(\rho)\Delta^{2} \notag \\
& = \sum_{\substack{i,j=1:\\i>j}}^{rn}\left(  \mathcal{L}'_{i\!\Mod{r}}\circ\mathcal{L}'_{j\!\Mod{r}}\right)
(\rho)\Delta^{2} + \sum_{\substack{i,j=1:\\(i,j)\in \mathcal{T}}}^{rn}\left(  \mathcal{L}'_{i\!\Mod{r}}\circ\mathcal{L}'_{j\!\Mod{r}}\right)
(\rho)\Delta^{2}  ,
\label{eq:key-eq-proof-app}
\end{align}
where
\begin{equation}
\mathcal{T}  \coloneqq \{(i,j):  0\leq j - i < r \wedge
\lfloor (i-1)/r \rfloor = \lfloor (j-1)/r \rfloor \}.
\end{equation}
Figure~\ref{fig:proof-idea} depicts an example with $r=3$ and $n=4$ in order to illustrate the main idea behind the equality in \eqref{eq:key-eq-proof-app}. 
\begin{figure}
    \centering
    \includegraphics[scale=0.5]{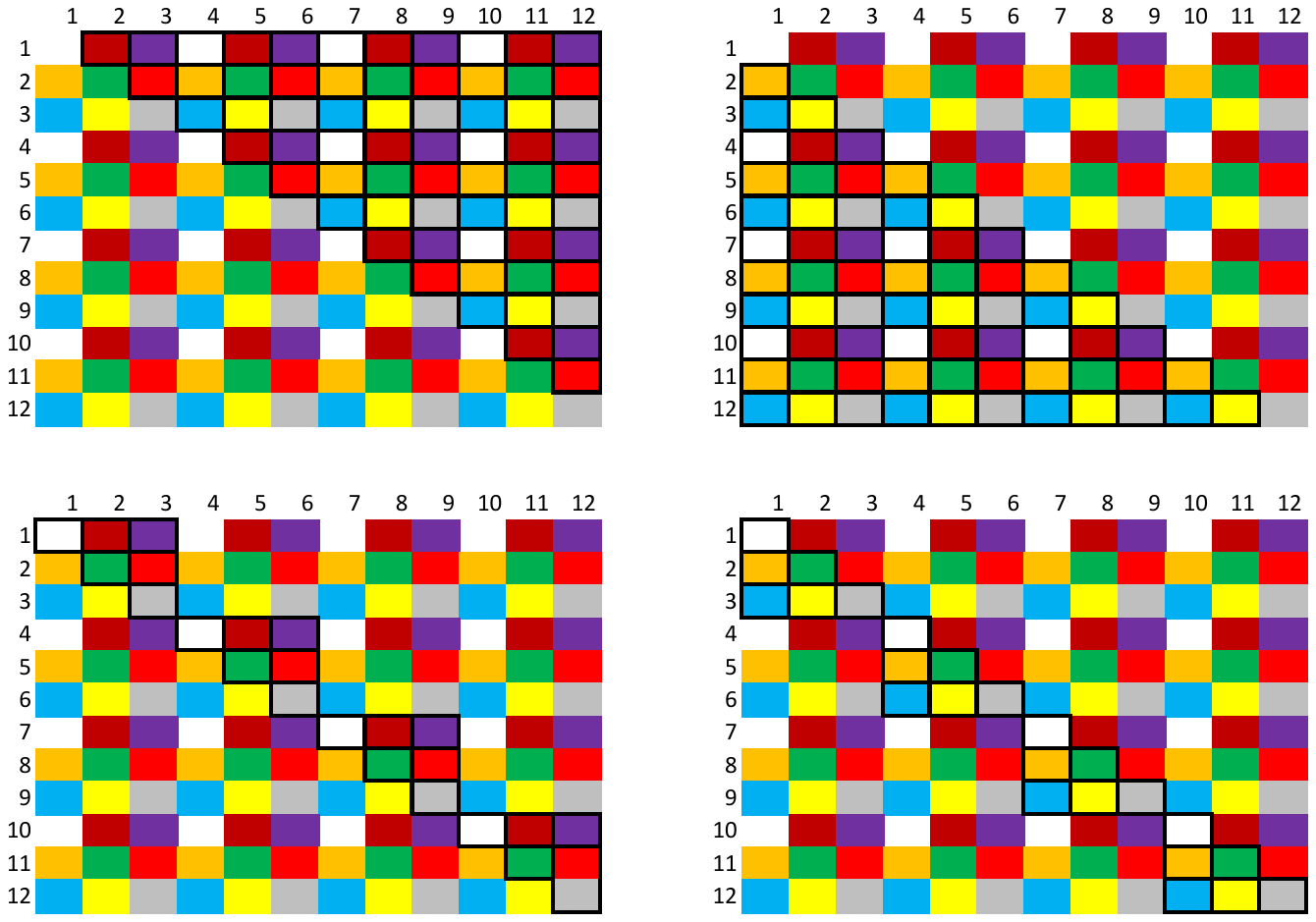}
    \caption{Depiction of the main idea behind the equality in \eqref{eq:key-eq-proof-app} for the case $r=3$ and $n=4$. A unique color is assigned to each pair $(i\!\Mod{r}, j\!\Mod{r})$ for all values of $i, j \in \{1, \ldots, rn\}$. The top left diagram depicts the first term in \eqref{eq:key-eq-proof-app}, with all terms included in that sum depicted with thick borders. Similarly, the bottom right diagram depicts the second term in \eqref{eq:key-eq-proof-app}, the top right diagram depicts the third term in \eqref{eq:key-eq-proof-app}, and the bottom left diagram depicts the fourth term in \eqref{eq:key-eq-proof-app}.}
    \label{fig:proof-idea}
\end{figure}
This implies that
\begin{align}
    & \sum_{\substack{i,j=1:\\i<j}}^{rn}\left(  \mathcal{L}'_{i\!\Mod{r}}\circ\mathcal{L}'_{j\!\Mod{r}}\right)
(\rho)\Delta^{2} - \sum_{\substack{i,j=1:\\(i,j)\in \mathcal{T}}}^{rn}\left(  \mathcal{L}'_{i\!\Mod{r}}\circ\mathcal{L}'_{j\!\Mod{r}}\right)
(\rho)\Delta^{2} \notag \\
& = \sum_{\substack{i,j=1:\\i>j}}^{rn}\left(  \mathcal{L}'_{i\!\Mod{r}}\circ\mathcal{L}'_{j\!\Mod{r}}\right)
(\rho)\Delta^{2} - \sum_{\substack{i,j=1:\\(j,i)\in \mathcal{T}}}^{rn}\left(  \mathcal{L}'_{i\!\Mod{r}}\circ\mathcal{L}'_{j\!\Mod{r}}\right)
(\rho)\Delta^{2} .
\end{align}
Using the above equation and \eqref{eq:first-term-simplify}, we can rewrite \eqref{eq:error-gen-decompose} in the following way:
\begin{multline}
      \sum_{i,j=1:i<j}^{rn}\left(  \mathcal{L}'_{i\!\Mod{r}}\circ\mathcal{L}'_{j\!\Mod{r}}\right)
(\rho)\Delta^{2}  \\
   \!\!\!\! = \frac{1}{2}\Bigg(\sum_{i, j=1}^{r}\left( \mathcal{L}'_{i\!\Mod{r}}\circ\mathcal{L}'_{j\!\Mod{r}}\right) (\rho) n^2\Delta^{2} 
 \\
 \qquad + \sum_{\substack{i,j=1:(i,j)\in \mathcal{T}}}^{rn}\left(  \mathcal{L}'_{i\!\Mod{r}}\circ\mathcal{L}'_{j\!\Mod{r}}\right)
(\rho)\Delta^{2} \\
 \qquad - \sum_{\substack{i,j=1:(j,i)\in \mathcal{T}}}^{rn}\left(  \mathcal{L}'_{i\!\Mod{r}}\circ\mathcal{L}'_{j\!\Mod{r}}\right)
(\rho)\Delta^{2} \\
- \sum_{i,j=1:i=j}^{rn}\left( \mathcal{L}'_{i\!\Mod{r}}\circ\mathcal{L}'_{j\!\Mod{r}}\right)
(\rho)\Delta^{2}\Bigg).\label{eq:n2-n-terms}
\end{multline}
Now, let us simplify the second, third, and fourth terms of the right-hand side of the preceding equation, one by one. The second term can be simplified in the following way: 
\begin{align}
    & \sum_{i,j=1:(i,j)\in \mathcal{T}}^{rn}\left(  \mathcal{L}'_{i\!\Mod{r}}\circ\mathcal{L}'_{j\!\Mod{r}}\right)
(\rho)\Delta^{2} \\
& = \underbrace{ \sum_{i=1}^{r} \left(\sum_{j=1:(i,j)\in \mathcal{T}}^{r}\mathcal{L}'_{i\!\Mod{r}}\circ\mathcal{L}'_{j\!\Mod{r}}\right) (\rho)\Delta^{2}}_\text{$ \frac{r(r+1)}{2} $ terms}  \notag \\
&\hspace{1cm} + \cdots + \underbrace{\sum_{i=r(n-1)+1}^{rn} \left(\sum_{j=r(n-1)+1:(i,j)\in \mathcal{T}}^{rn} \mathcal{L}'_{i\!\Mod{r}}\circ\mathcal{L}'_{j\!\Mod{r}}\right) (\rho)\Delta^{2} }_\text{$ \frac{r(r+1)}{2} $ terms}\\
& = \underbrace{ \sum_{i=1}^{r}\left(\sum_{j=1:(i,j)\in \mathcal{T}}^{r} \mathcal{L}'_{i\!\Mod{r}}\circ\mathcal{L}'_{j\!\Mod{r}} \right) (\rho)}_\text{$ \frac{r(r+1)}{2} $ terms} n\Delta^{2}
\label{eq:sec-n2-n-terms}
\end{align}
Similarly, the third term on the right-hand side of \eqref{eq:n2-n-terms} can be simplified and written as 
\begin{multline}
    \sum_{i,j=1:(j,i)\in \mathcal{T}}^{rn}\left(  \mathcal{L}'_{i\!\Mod{r}}\circ\mathcal{L}'_{j\!\Mod{r}}\right)
(\rho)\Delta^{2} \\
= \underbrace{\sum_{i=1}^{r} \left(\sum_{j=1:(j,i)\in \mathcal{T}}^{r} \mathcal{L}'_{i\!\Mod{r}}\circ\mathcal{L}'_{j\!\Mod{r}}\right)(\rho)}_\text{$ \frac{r(r+1)}{2} $ terms} n\Delta^{2}  \label{eq:third-n2-n-terms}
\end{multline}
All that remains is to simplify the fourth term of \eqref{eq:n2-n-terms}. As we can see, it is possible to write it like this:
\begin{equation}
    \sum_{i,j=1:i=j}^{rn}\left( \mathcal{L}'_{i\!\Mod{r}}\circ\mathcal{L}'_{j\!\Mod{r}}\right)
(\rho)\Delta^{2} = \sum_{i=1}^{r}\left( \mathcal{L}'_{i\!\Mod{r}}\circ\mathcal{L}'_{i\!\Mod{r}}\right)
(\rho)n \Delta^{2}.
\label{eq:fourth-n2-n-terms}
\end{equation}
Substituting \eqref{eq:sec-n2-n-terms}, \eqref{eq:third-n2-n-terms}, and \eqref{eq:fourth-n2-n-terms} into \eqref{eq:n2-n-terms}, we get
\begin{align}
     & \sum_{i,j=1:i<j}^{rn}\left(  \mathcal{L}'_{i\!\Mod{r}}\circ\mathcal{L}'_{j\!\Mod{r}}\right)
(\rho)\Delta^{2}\notag \\
 & = \frac{1}{2}\Bigg(\sum_{i, j=1}^{r}\left( \mathcal{L}'_{i\!\Mod{r}}\circ\mathcal{L}'_{j\!\Mod{r}}\right) (\rho) n^2\Delta^{2} \notag \\
 &\qquad + \sum_{i=1}^{r} \left(\sum_{j=1:(i,j)\in \mathcal{T}}^{r}\mathcal{L}'_{i\!\Mod{r}}\circ\mathcal{L}'_{j\!\Mod{r}} \right) (\rho)n\Delta^{2} \notag \\ \notag
 & \qquad - \sum_{i=1}^{r} \left(\sum_{j=1:(j,i)\in \mathcal{T}}^{r} \mathcal{L}'_{i\!\Mod{r}}\circ\mathcal{L}'_{j\!\Mod{r}} \right) (\rho)n\Delta^{2} \notag \\
 & \qquad - \sum_{i=1}^{r}\left( \mathcal{L}'_{i\!\Mod{r}}\circ\mathcal{L}'_{i\!\Mod{r}}\right)
(\rho)n \Delta^{2}\Bigg).
\label{eq:n2-n-terms-mod}
\end{align}
Furthermore, substituting the above equation into \eqref{eq:i<j-n2-terms}, we finally obtain the following: 
\begin{align}
    & \sum_{i,j=1:i<j}^{rn}\left(  \mathcal{L}'_{i\!\Mod{r}}\circ\mathcal{L}'_{j\!\Mod{r}}\right)
(\rho)\Delta^{2}-\left(  \frac{1}{2}\sum_{i,j=1}^{r}\left(  \mathcal{L}'%
_{i}\circ\mathcal{L}'_{j}\right)  \left(  \rho\right)  n^{2}\Delta^{2}\right) \notag  \\
& \notag = \frac{1}{2}\Bigg(\sum_{i, j=1}^{r}\left( \mathcal{L}'_{i}\circ\mathcal{L}'_{j}\right) (\rho) n^2\Delta^{2} + \sum_{i=1}^{r} \left(\sum_{j=1:(i,j)\in \mathcal{T}}^{r} \mathcal{L}'_{i}\circ\mathcal{L}'_{j} \right) (\rho)n\Delta^{2}\\ \notag
 &\hspace{2cm} - \sum_{i=1}^{r} \left(\sum_{j=1:(j,i)\in\mathcal{T}}^{r} \mathcal{L}'_{i}\circ\mathcal{L}'_{j} \right) (\rho)n\Delta^{2} - \sum_{i=1}^{r}\left( \mathcal{L}'_{i}\circ\mathcal{L}'_{i}\right)
(\rho)n \Delta^{2}\Bigg)\\
&\hspace{5.5cm} - \left(  \frac{1}{2}\sum_{i,j=1}^{r}\left(  \mathcal{L}'%
_{i}\circ\mathcal{L}'_{j}\right)  \left(  \rho\right)  n^{2}\Delta^{2}\right) \\
& = \frac{1}{2}\Bigg( \sum_{i=1}^{r} \left(\sum_{j=1:(i,j)\in \mathcal{T}}^r \mathcal{L}'_{i}\circ\mathcal{L}'_{j} \right) (\rho)n\Delta^{2} - \sum_{i=1}^{r} \left(\sum_{j=1:(j,i)\in \mathcal{T}}^r \mathcal{L}'_{i}\circ\mathcal{L}'_{j} \right) (\rho)n\Delta^{2} \notag \\ 
 &\hspace{7cm}  - \sum_{i=1}^{r}\left( \mathcal{L}'_{i}\circ\mathcal{L}'_{i}\right)
(\rho)n \Delta^{2}\Bigg)
\label{eq:penultimate-step-big-proof}\\
& = O(r^2 n\Delta^2). 
\label{eq:last-step-big-proof}
\end{align}

Substituting $\Delta=t/n$ and repeating Algorithm~2 (but without Steps~3 and 4 there)
\begin{equation}
n = O(r^2t^2/\varepsilon) =O((J+K)^2t^2/\varepsilon)     
\end{equation}
times approximates the channel $e^{\mathcal{L}'t}$ with an approximation error $O(\varepsilon)$.
In order to simulate our desired channel $e^{\mathcal{L}t}$, we simulate the above channel for time $t' = \left \Vert \mathcal{L} \right \Vert_{\max} t$, so that
\begin{align}
    e^{\mathcal{L}'t'} & = e^{\frac{\mathcal{L}}{\left \Vert \mathcal{L} \right \Vert_{\max}} \left \Vert \mathcal{L} \right \Vert_{\max}t} = e^{\mathcal{L}t}.
\end{align}
This implies that we need $n = O((J+K)^2\left \Vert \mathcal{L} \right \Vert_{\max}^2 t^2/\varepsilon)$ repetitions of Algorithm~2 (without its Steps~3 and 4), and for this, we need $n = O((J+K)^2\left \Vert \mathcal{L} \right \Vert_{\max}^2 t^2/\varepsilon)$ copies of each program state in  $\{\sigma_{1}, \ldots, \sigma_{J}, \psi_{1}, \ldots, \psi_{K}\}$.
The total number of program states used is then $O\!\left((J+K)^3\left \Vert \mathcal{L} \right \Vert_{\max}^2 t^2/\varepsilon\right )$.
As the aforementioned argument holds for an arbitrary input state $\rho$, the error bound holds more generally for the diamond distance.

As mentioned at the beginning of this proof before \eqref{eq:1st-step-big-proof}, we can instead adopt the order $r, \ldots, 1, 1, \ldots, r,$ $r, \ldots, 1, 1, \ldots, r,$ $ \ldots, r,\ldots, 1, 1, \ldots, r$, for a total of $2rn$ steps. This implies that the set $\mathcal{T}$ (appropriately modified and enlarged for this modified ordering) has the symmetry $(i,j) \in \mathcal{T}$ if and only if $(j,i) \in \mathcal{T}$. As such, by inspecting the first line of \eqref{eq:penultimate-step-big-proof} above, the difference of the two terms there is equal to zero for this modified ordering, implying that the resulting error term is $O(rn\Delta^2)$ and matches the order of the term in~\eqref{eq:J-term-order-pf}. Thus, repeating the argument in the last paragraph above for this case, we conclude that we can set $n = O(r t^2/\varepsilon) =O((J+K)t^2/\varepsilon) $    
to  approximate the channel $e^{\mathcal{L}'t}$ with an error $O(\varepsilon)$. The total number of program states used with this modified ordering is then $O\!\left((J+K)^2\left \Vert \mathcal{L} \right \Vert_{\max}^2 t^2/\varepsilon\right )$. Figure~\ref{fig:proof-idea-other} depicts the same example from Figure~\ref{fig:proof-idea} but with this modified ordering. It demonstrates how the second and fourth terms of \eqref{eq:key-eq-proof-app} contain the same terms, ultimately leading to the cancellation in the first line of~\eqref{eq:penultimate-step-big-proof}.

\begin{figure}
    \centering
    \includegraphics[scale=0.4]{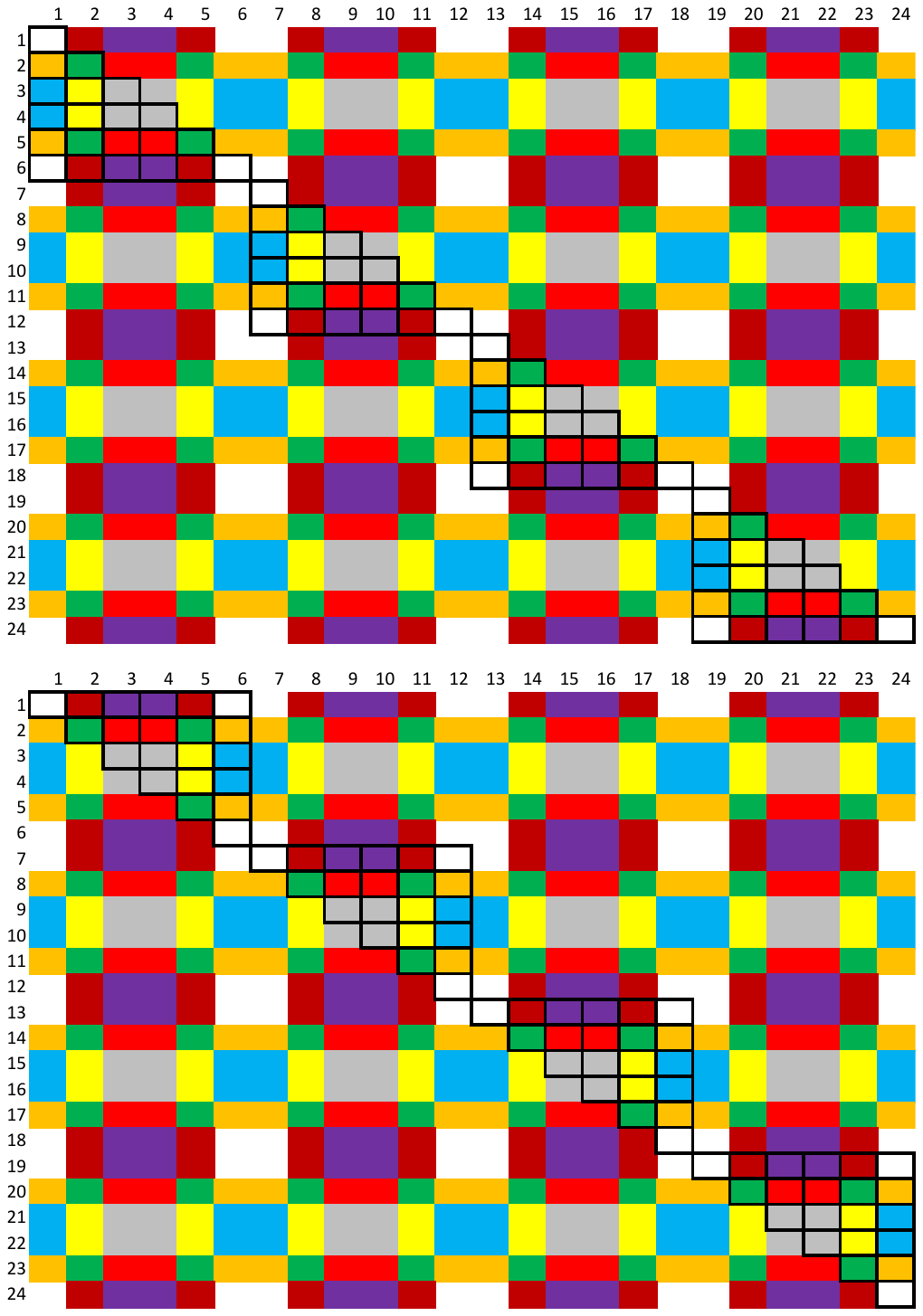}
    \caption{For the case $r=3$ and $n=4$, depiction of how the second and fourth terms of \eqref{eq:key-eq-proof-app} contain the same terms, when using the order $r, \ldots, 1, 1, \ldots, r,$ $r, \ldots, 1, 1, \ldots, r,$ $ \ldots, r,\ldots, 1, 1, \ldots, r$ for Trotter steps. A unique color is assigned to each pair for all values of $i, j \in \{1, \ldots, 2rn\}$. The top diagram depicts the second term in \eqref{eq:key-eq-proof-app}, using the modified ordering, with all terms included in that sum depicted with thick borders. The bottom diagram depicts the fourth term in \eqref{eq:key-eq-proof-app}, using the modified ordering, with all terms included in that sum depicted with thick borders.}
    \label{fig:proof-idea-other}
\end{figure}

\section{Lemma~\ref{lem:cyclic-swap-lem} Statement and Proof}

\label{app:cyclic-swap-lem}

For the purposes of our development in this appendix, we suppose that the system of
interest is labeled by $0$, and the program systems are labeled by
$1,\ldots,2D$, with system $1$ entangled with $D$, system $2$ entangled with
$D+1$, \ldots, and system $D$ entangled with $2D$. We also employ the shorthand $\Gamma_{m,n}\equiv |\Gamma\rangle\!\langle\Gamma|_{m,n}$ throughout to denote the maximally entangled vector for systems $m$ and $n$.

\begin{lemma}
\label{lem:cyclic-swap-lem}
For
\begin{equation}
    M  \coloneqq \frac{1}{d^{D/2}}\left(  I_{0}\otimes\bigotimes\limits_{\ell=1}%
^{D}|\Gamma\rangle\!\langle\Gamma|_{\ell,D+\ell}\right)  \left(  \mathsf{CYCSWAP}%
_{0,1,\ldots,D}\otimes I_{D+1,\ldots,2D}\right),
\end{equation}
the following equalities hold:
\begin{multline}
\operatorname{Tr}_{1,\ldots,2D}\!\left[M^{\dag}M\left(  \rho_{0}\otimes L_{1}%
\Gamma_{1,D+1}L_{1}^{\prime}\otimes L_{2}\Gamma_{2,D+2}L_{2}^{\prime}%
\otimes\cdots\otimes L_{D}\Gamma_{D,2D}L_{D}^{\prime}\right)
\right]\label{eq:long-proof-cycswap-1}\\
=\left(  L_{D}^{\prime}L_{D-1}^{\prime}\cdots L_{3}^{\prime}L_{2}^{\prime
}L_{1}^{\prime}L_{1}L_{2}L_{3}\cdots L_{D-1}L_{D}\rho\right)  _{0}\ ,
\end{multline}
\begin{multline}
\operatorname{Tr}_{1,\ldots,2D}\!\left[\left(  \rho_{0}\otimes L_{1}\Gamma
_{1,D+1}L_{1}^{\prime}\otimes L_{2}\Gamma_{2,D+2}L_{2}^{\prime}\otimes
\cdots\otimes L_{D}\Gamma_{D,2D}L_{D}^{\prime}\right)  M^{\dag}M\right]\\
=\left(  \rho L_{D}^{\prime}L_{D-1}^{\prime}\cdots L_{3}^{\prime}L_{2}%
^{\prime}L_{1}^{\prime}L_{1}L_{2}L_{3}\cdots L_{D-1}L_{D}\right)  _{0}\ ,
\label{eq:long-proof-cycswap-2}
\end{multline}
\begin{multline}
\operatorname{Tr}_{1,\ldots,2D}\!\left[  M\left(  \rho_{0}\otimes L_{1}%
\Gamma_{1,D+1}L_{1}^{\prime}\otimes L_{2}\Gamma_{2,D+2}L_{2}^{\prime}%
\otimes\cdots\otimes L_{D}\Gamma_{D,2D}L_{D}^{\prime}\right)  M^{\dag}\right]
\label{eq:long-proof-cycswap-other}\\
=\left(  L_{1}L_{2}\cdots L_{D-1}L_{D}\rho L_{D}^{\prime}L_{D-1}^{\prime
}\cdots L_{2}^{\prime}L_{1}^{\prime}\right)  _{0},
\end{multline}
where $\rho_0$ is a density operator acting on system $0$, the operators $L_1, \ldots, L_D$, $L_1', \ldots, L_D'$ are arbitrary square linear operators, and identity operators for systems $D+1, \ldots, 2D$ are implicit in each of the first lines above.
\end{lemma}

\begin{proof}
Let us begin with the
following observation:%
\begin{align}
M  & =\frac{1}{d^{D/2}}\left(  I_{0}\otimes\bigotimes\limits_{\ell=1}%
^{D}|\Gamma\rangle\!\langle\Gamma|_{\ell,D+\ell}\right)  \left(  \mathsf{CYCSWAP}%
_{0,1,\ldots,D}\otimes I_{D+1,\ldots,2D}\right)  \\
& =\frac{1}{d^{D/2}}\left(
\begin{array}
[c]{c}%
\sum_{k}|k\rangle\!\langle k|_{0}\otimes\sum_{\substack{i_{1},j_{1}%
,\\i_{2},j_{2},\\\ldots,i_{D},j_{D}}}|i_{1}\rangle\!\langle j_{1}|_{1}%
\otimes|i_{2}\rangle\!\langle j_{2}|_{2}\otimes\cdots\otimes|i_{D}\rangle\!\langle
j_{D}|_{D}\otimes\\
|i_{1}\rangle\!\langle j_{1}|_{D+1}\otimes|i_{2}\rangle\!\langle j_{2}%
|_{D+2}\otimes\cdots\otimes|i_{D}\rangle\!\langle j_{D}|_{2D}%
\end{array}
\right)  \nonumber\\
& \qquad\times\left(  \mathsf{CYCSWAP}_{0,1,\ldots,D}\otimes I_{D+1,\ldots
,2D}\right)  \\
& =\frac{1}{d^{D/2}}\left(
\begin{array}
[c]{c}%
\sum_{\substack{k,i_{1},j_{1},\\i_{2},j_{2},\\\ldots,i_{D},j_{D}}%
}|k\rangle\!\langle j_{D}|_{0}\otimes|i_{1}\rangle\!\langle k|_{1}\otimes
|i_{2}\rangle\!\langle j_{1}|_{2}\otimes\cdots\otimes|i_{D-1}\rangle\!\langle
j_{D-2}|_{D-1}\\
\otimes|i_{D}\rangle\!\langle j_{D-1}|_{D}\otimes|i_{1}\rangle\!\langle
j_{1}|_{D+1}\otimes|i_{2}\rangle\!\langle j_{2}|_{D+2}\otimes\cdots\otimes
|i_{D}\rangle\!\langle j_{D}|_{2D}%
\end{array}
\right)  .
\end{align}
Then we find that%
\begin{align}
& M^{\dag}M\nonumber\\
& =\frac{1}{d^{D/2}}\left(
\begin{array}
[c]{c}%
\sum_{\substack{k^{\prime},i_{1}^{\prime},j_{1}^{\prime},\\i_{2}^{\prime
},j_{2}^{\prime},\\\ldots,i_{D}^{\prime},j_{D}^{\prime}}}|j_{D}^{\prime
}\rangle\!\langle k^{\prime}|_{0}\otimes|k^{\prime}\rangle\!\langle i_{1}^{\prime
}|_{1}\otimes|j_{1}^{\prime}\rangle\!\langle i_{2}^{\prime}|_{2}\otimes
\cdots\otimes|j_{D-2}^{\prime}\rangle\!\langle i_{D-1}^{\prime}|_{D-1}\\
\otimes|j_{D-1}^{\prime}\rangle\!\langle i_{D}^{\prime}|_{D}\otimes
|j_{1}^{\prime}\rangle\!\langle i_{1}^{\prime}|_{D+1}\otimes|j_{2}^{\prime
}\rangle\!\langle i_{2}^{\prime}|_{D+2}\otimes\cdots\otimes|j_{D}^{\prime
}\rangle\!\langle i_{D}^{\prime}|_{2D}%
\end{array}
\right)  \nonumber\\
& \quad\times\frac{1}{d^{D/2}}\left(
\begin{array}
[c]{c}%
\sum_{\substack{k,i_{1},j_{1},\\i_{2},j_{2},\\\ldots,i_{D},j_{D}}%
}|k\rangle\!\langle j_{D}|_{0}\otimes|i_{1}\rangle\!\langle k|_{1}\otimes
|i_{2}\rangle\!\langle j_{1}|_{2}\otimes\cdots\otimes|i_{D-1}\rangle\!\langle
j_{D-2}|_{D-1}\\
\otimes|i_{D}\rangle\!\langle j_{D-1}|_{D}\otimes|i_{1}\rangle\!\langle
j_{1}|_{D+1}\otimes|i_{2}\rangle\!\langle j_{2}|_{D+2}\otimes\cdots\otimes
|i_{D}\rangle\!\langle j_{D}|_{2D}%
\end{array}
\right)  \\
& =\frac{1}{d^{D}}\sum_{\substack{k^{\prime},i_{1}^{\prime},j_{1}^{\prime
},\\i_{2}^{\prime},j_{2}^{\prime},\\\ldots,i_{D}^{\prime},j_{D}^{\prime}}%
}\sum_{\substack{k,i_{1},j_{1},\\i_{2},j_{2},\\\ldots,i_{D},j_{D}}%
}|j_{D}^{\prime}\rangle\!\langle k^{\prime}|k\rangle\!\langle j_{D}|_{0}%
\otimes|k^{\prime}\rangle\!\langle i_{1}^{\prime}|i_{1}\rangle\!\langle
k|_{1}\otimes|j_{1}^{\prime}\rangle\!\langle i_{2}^{\prime}|i_{2}\rangle\!\langle
j_{1}|_{2}\otimes\cdots\nonumber\\
& \qquad\otimes|j_{D-2}^{\prime}\rangle\!\langle i_{D-1}^{\prime}|i_{D-1}%
\rangle\!\langle j_{D-2}|_{D-1}\otimes|j_{D-1}^{\prime}\rangle\!\langle
i_{D}^{\prime}|i_{D}\rangle\!\langle j_{D-1}|_{D}\nonumber\\
& \qquad\otimes|j_{1}^{\prime}\rangle\!\langle i_{1}^{\prime}|i_{1}%
\rangle\!\langle j_{1}|_{D+1}\otimes|j_{2}^{\prime}\rangle\!\langle i_{2}^{\prime
}|i_{2}\rangle\!\langle j_{2}|_{D+2}\otimes\cdots\otimes|j_{D}^{\prime}%
\rangle\!\langle i_{D}^{\prime}|i_{D}\rangle\!\langle j_{D}|_{2D}\\
& =\sum_{\substack{j_{1}^{\prime},i_{2}^{\prime},j_{2}^{\prime},\\\ldots
,i_{D}^{\prime},j_{D}^{\prime}}}\sum_{\substack{k,i_{1},j_{1},i_{2}%
,j_{2},\\\ldots,i_{D},j_{D}}}|j_{D}^{\prime}\rangle\!\langle j_{D}|_{0}%
\otimes|k\rangle\!\langle k|_{1}\otimes|j_{1}^{\prime}\rangle\!\langle j_{1}%
|_{2}\otimes\cdots\otimes|j_{D-2}^{\prime}\rangle\!\langle j_{D-2}%
|_{D-1}\nonumber\\
& \qquad\otimes|j_{D-1}^{\prime}\rangle\!\langle j_{D-1}|_{D}\otimes
|j_{1}^{\prime}\rangle\!\langle j_{1}|_{D+1}\otimes|j_{2}^{\prime}\rangle\!\langle
j_{2}|_{D+2}\otimes\cdots\otimes|j_{D}^{\prime}\rangle\!\langle j_{D}|_{2D}\\
& =I_{1}\otimes\Gamma_{0,2D}\otimes\Gamma_{2,D+1}\otimes\Gamma_{3,D+2}%
\otimes\cdots\otimes\Gamma_{D-1,2D-2}\otimes\Gamma_{D,2D-1}.
\end{align}

We first prove \eqref{eq:long-proof-cycswap-1}. To this end, consider that%
\begin{align}
& M^{\dag}M\left(  \rho_{0}\otimes L_{1}\Gamma_{1,D+1}L_{1}^{\prime}\otimes
L_{2}\Gamma_{2,D+2}L_{2}^{\prime}\otimes\cdots\otimes L_{D}\Gamma_{D,2D}%
L_{D}^{\prime}\right)  \nonumber\\
& =\left[
\begin{array}
[c]{c}%
\left(  I_{1}\otimes\Gamma_{0,2D}\otimes\Gamma_{2,D+1}\otimes\Gamma
_{3,D+2}\otimes\cdots\otimes\Gamma_{D-1,2D-2}\otimes\Gamma_{D,2D-1}\right)  \\
\left(  \rho_{0}\otimes L_{1}\Gamma_{1,D+1}L_{1}^{\prime}\otimes L_{2}%
\Gamma_{2,D+2}L_{2}^{\prime}\otimes\cdots\otimes L_{D}\Gamma_{D,2D}%
L_{D}^{\prime}\right)
\end{array}
\right]  \\
& =\sum_{\substack{i_{1},j_{1},i_{2},j_{2},\\i_{3},j_{3},i_{D-1}%
,j_{D-1},\\i_{D},j_{D},k_{1},\ell_{1},\\k_{2},\ell_{2},\ldots,k_{D},\ell_{D}%
}}\left[
\begin{array}
[c]{c}%
\left(
\begin{array}
[c]{c}%
I_{1}\otimes|i_{1}\rangle\!\langle j_{1}|_{0}\otimes|i_{1}\rangle\!\langle
j_{1}|_{2D}\otimes|i_{2}\rangle\!\langle j_{2}|_{2}\otimes|i_{2}\rangle\!\langle
j_{2}|_{D+1}\\
\otimes|i_{3}\rangle\!\langle j_{3}|_{3}\otimes|i_{3}\rangle\!\langle j_{3}%
|_{D+2}\otimes\cdots\otimes|i_{D-1}\rangle\!\langle j_{D-1}|_{D-1}\\
\otimes|i_{D-1}\rangle\!\langle j_{D-1}|_{2D-2}\otimes|i_{D}\rangle\!\langle
j_{D}|_{D}\otimes|i_{D}\rangle\!\langle j_{D}|_{2D-1}%
\end{array}
\right)  \\
\left(
\begin{array}
[c]{c}%
\rho_{0}\otimes L_{1}|k_{1}\rangle\!\langle\ell_{1}|_{1}L_{1}^{\prime}%
\otimes|k_{1}\rangle\!\langle\ell_{1}|_{D+1}\otimes L_{2}|k_{2}\rangle
\langle\ell_{2}|_{2}L_{2}^{\prime}\\
\otimes|k_{2}\rangle\!\langle\ell_{2}|_{D+2}\otimes\cdots\otimes L_{D}%
|k_{D}\rangle\!\langle\ell_{D}|_{D}L_{D}^{\prime}\otimes|k_{D}\rangle\!\langle
\ell_{D}|_{2D}%
\end{array}
\right)
\end{array}
\right]  \\
& =\sum_{\substack{i_{1},j_{1},i_{2},j_{2},\\i_{3},j_{3},i_{D-1}%
,j_{D-1},\\i_{D},j_{D},k_{1},\ell_{1},\\k_{2},\ell_{2},k_{3},\ell_{3}%
,\ldots,\\k_{D-2},\ell_{D-2},k_{D-1},\\\ell_{D-1},k_{D},\ell_{D}}}\left(
\begin{array}
[c]{c}%
L_{1}|k_{1}\rangle\!\langle\ell_{1}|_{1}L_{1}^{\prime}\otimes|i_{1}%
\rangle\!\langle j_{1}|_{0}\rho_{0}\otimes|i_{1}\rangle\!\langle j_{1}|_{2D}%
|k_{D}\rangle\!\langle\ell_{D}|_{2D}\\
\otimes|i_{2}\rangle\!\langle j_{2}|_{2}L_{2}|k_{2}\rangle\!\langle\ell_{2}%
|_{2}L_{2}^{\prime}\otimes|i_{2}\rangle\!\langle j_{2}|_{D+1}|k_{1}%
\rangle\!\langle\ell_{1}|_{D+1}\\
\otimes|i_{3}\rangle\!\langle j_{3}|_{3}L_{3}|k_{3}\rangle\!\langle\ell_{3}%
|_{3}L_{3}^{\prime}\otimes|i_{3}\rangle\!\langle j_{3}|_{D+2}|k_{2}%
\rangle\!\langle\ell_{2}|_{D+2}\otimes\cdots\\
\otimes|i_{D-1}\rangle\!\langle j_{D-1}|_{D-1}L_{D-1}|k_{D-1}\rangle\!\langle
\ell_{D-1}|_{D}L_{D-1}^{\prime}\\
\otimes|i_{D-1}\rangle\!\langle j_{D-1}|_{2D-2}|k_{D-2}\rangle\!\langle\ell
_{D-2}|_{2D-2}\\
\otimes|i_{D}\rangle\!\langle j_{D}|_{D}L_{D}|k_{D}\rangle\!\langle\ell_{D}%
|_{D}L_{D}^{\prime}\\
\otimes|i_{D}\rangle\!\langle j_{D}|_{2D-1}|k_{D-1}\rangle\!\langle\ell
_{D-1}|_{2D-1}%
\end{array}
\right)  .
\end{align}
Taking the partial trace over systems $1,\ldots,2D$ of the last line above
gives%
\begin{align}
& \sum_{\substack{i_{1},j_{1},i_{2},j_{2},\\i_{3},j_{3},i_{D-1},j_{D-1}%
,\\i_{D},j_{D},k_{1},\ell_{1},\\k_{2},\ell_{2},k_{3},\ell_{3},\ldots
,\\k_{D-2},\ell_{D-2},k_{D-1},\\\ell_{D-1},k_{D},\ell_{D}}}\left(
\begin{array}
[c]{c}%
|i_{1}\rangle\!\langle j_{1}|_{0}\rho_{0}~\langle\ell_{1}|_{1}L_{1}^{\prime
}L_{1}|k_{1}\rangle\ \langle j_{1}|k_{D}\rangle\!\langle\ell_{D}|i_{1}\rangle\\
\langle j_{2}|_{2}L_{2}|k_{2}\rangle\!\langle\ell_{2}|L_{2}^{\prime}%
|i_{2}\rangle\ \langle j_{2}|k_{1}\rangle\!\langle\ell_{1}|i_{2}\rangle\\
\langle j_{3}|_{3}L_{3}|k_{3}\rangle\!\langle\ell_{3}|_{3}L_{3}^{\prime}%
|i_{3}\rangle\ \langle j_{3}|k_{2}\rangle\!\langle\ell_{2}|i_{3}\rangle\cdots\\
\langle j_{D-1}|_{D-1}L_{D-1}|k_{D-1}\rangle\!\langle\ell_{D-1}|_{D}%
L_{D-1}^{\prime}|i_{D-1}\rangle\\
\ \langle j_{D-1}|k_{D-2}\rangle\!\langle\ell_{D-2}|i_{D-1}\rangle\\
\langle j_{D}|_{D}L_{D}|k_{D}\rangle\!\langle\ell_{D}|_{D}L_{D}^{\prime}%
|i_{D}\rangle\ \langle j_{D}|k_{D-1}\rangle\!\langle\ell_{D-1}|i_{D}\rangle
\end{array}
\right)  \nonumber\\
& =\sum_{\substack{i_{1},j_{1},i_{2},j_{2},\\i_{3},j_{3},i_{D-1}%
,j_{D-1},\\i_{D},j_{D}}}\left(
\begin{array}
[c]{c}%
|i_{1}\rangle\!\langle j_{1}|_{0}\rho_{0}~\langle i_{2}|_{1}L_{1}^{\prime}%
L_{1}|j_{2}\rangle\ \\
\langle j_{2}|_{2}L_{2}|j_{3}\rangle\!\langle i_{3}|L_{2}^{\prime}|i_{2}%
\rangle\ \\
\langle j_{3}|_{3}L_{3}|j_{4}\rangle\!\langle i_{4}|_{3}L_{3}^{\prime}%
|i_{3}\rangle\ \cdots\\
\langle j_{D-1}|_{D-1}L_{D-1}|j_{D}\rangle\!\langle i_{D}|_{D}L_{D-1}^{\prime
}|i_{D-1}\rangle\ \\
\langle j_{D}|_{D}L_{D}|j_{1}\rangle\!\langle i_{1}|_{D}L_{D}^{\prime}%
|i_{D}\rangle
\end{array}
\right)  \\
& =\sum_{\substack{i_{1},j_{1},i_{2},j_{2},\\i_{3},j_{3},i_{D-1}%
,j_{D-1},\\i_{D},j_{D}}}\left(
\begin{array}
[c]{c}%
|i_{1}\rangle\!\langle j_{1}|_{0}\rho_{0}~\langle i_{2}|_{1}L_{1}^{\prime}%
L_{1}|j_{2}\rangle\ \langle j_{2}|_{2}L_{2}|j_{3}\rangle\!\langle i_{3}%
|L_{2}^{\prime}|i_{2}\rangle\ \\
\langle j_{3}|_{3}L_{3}|j_{4}\rangle\!\langle i_{4}|_{3}L_{3}^{\prime}%
|i_{3}\rangle\ \cdots\langle j_{D-1}|_{D-1}L_{D-1}|j_{D}\rangle\\
\langle i_{D}|_{D}L_{D-1}^{\prime}|i_{D-1}\rangle\ \langle j_{D}|_{D}%
L_{D}|j_{1}\rangle\!\langle i_{1}|_{D}L_{D}^{\prime}|i_{D}\rangle
\end{array}
\right)  \\
& =\sum_{\substack{i_{1},j_{1},i_{2},j_{2},\\i_{3},j_{3},i_{D-1}%
,j_{D-1},\\i_{D},j_{D}}}\left(
\begin{array}
[c]{c}%
|i_{1}\rangle\!\langle i_{1}|_{D}L_{D}^{\prime}|i_{D}\rangle\!\langle i_{D}%
|_{D}L_{D-1}^{\prime}|i_{D-1}\rangle\cdots\langle i_{4}|_{3}L_{3}^{\prime
}|i_{3}\rangle\\
\langle i_{3}|L_{2}^{\prime}|i_{2}\rangle\!\langle i_{2}|_{1}L_{1}^{\prime}%
L_{1}|j_{2}\rangle\!\langle j_{2}|_{2}L_{2}|j_{3}\rangle\\
\langle j_{3}|_{3}L_{3}|j_{4}\rangle\ \cdots\langle j_{D-1}|_{D-1}%
L_{D-1}|j_{D}\rangle\\
\langle j_{D}|_{D}L_{D}|j_{1}\rangle\ \langle j_{1}|_{0}\rho_{0}%
\end{array}
\right)  \\
& =\left(  L_{D}^{\prime}L_{D-1}^{\prime}\cdots L_{3}^{\prime}L_{2}^{\prime
}L_{1}^{\prime}L_{1}L_{2}L_{3}\cdots L_{D-1}L_{D}\rho\right)  _{0}\ .
\end{align}
This completes the proof of \eqref{eq:long-proof-cycswap-1}.

Using the Hermitian conjugate of \eqref{eq:long-proof-cycswap-1} along with
some substitutions then gives \eqref{eq:long-proof-cycswap-2}.

Now we finally prove \eqref{eq:long-proof-cycswap-other}.
To this end, consider that%
\begin{align}
& M\left(  \rho_{0}\otimes L_{1}\Gamma_{1,D+1}L_{1}^{\prime}\otimes
L_{2}\Gamma_{2,D+2}L_{2}^{\prime}\otimes\cdots\otimes L_{D}\Gamma_{D,2D}%
L_{D}^{\prime}\right)  M^{\dag}\nonumber\\
& =\frac{1}{d^{D}}\left(
\begin{array}
[c]{c}%
\sum_{\substack{k,i_{1},j_{1},\\i_{2},j_{2},\\\ldots,i_{D},j_{D}}%
}|k\rangle\!\langle j_{D}|_{0}\otimes|i_{1}\rangle\!\langle k|_{1}\otimes
|i_{2}\rangle\!\langle j_{1}|_{2}\otimes\cdots\otimes|i_{D-1}\rangle\!\langle
j_{D-2}|_{D-1}\\
\otimes|i_{D}\rangle\!\langle j_{D-1}|_{D}\otimes|i_{1}\rangle\!\langle
j_{1}|_{D+1}\otimes|i_{2}\rangle\!\langle j_{2}|_{D+2}\otimes\cdots\otimes
|i_{D}\rangle\!\langle j_{D}|_{2D}%
\end{array}
\right)  \nonumber\\
& \quad\times\left(  \rho_{0}\otimes L_{1}\Gamma_{1,D+1}L_{1}^{\prime}\otimes
L_{2}\Gamma_{2,D+2}L_{2}^{\prime}\otimes\cdots\otimes L_{D}\Gamma_{D,2D}%
L_{D}^{\prime}\right)  \nonumber\\
& \quad\times\left(
\begin{array}
[c]{c}%
\sum_{\substack{k^{\prime},i_{1}^{\prime},j_{1}^{\prime},\\i_{2}^{\prime
},j_{2}^{\prime},\\\ldots,i_{D}^{\prime},j_{D}^{\prime}}}|j_{D}^{\prime
}\rangle\!\langle k^{\prime}|_{0}\otimes|k^{\prime}\rangle\!\langle i_{1}^{\prime
}|_{1}\otimes|j_{1}^{\prime}\rangle\!\langle i_{2}^{\prime}|_{2}\otimes
\cdots\otimes|j_{D-2}^{\prime}\rangle\!\langle i_{D-1}^{\prime}|_{D-1}\\
\otimes|j_{D-1}^{\prime}\rangle\!\langle i_{D}^{\prime}|_{D}\otimes
|j_{1}^{\prime}\rangle\!\langle i_{1}^{\prime}|_{D+1}\otimes|j_{2}^{\prime
}\rangle\!\langle i_{2}^{\prime}|_{D+2}\otimes\cdots\otimes|j_{D}^{\prime
}\rangle\!\langle i_{D}^{\prime}|_{2D}%
\end{array}
\right)  \\
& =\frac{1}{d^{D}}\sum_{\substack{k,i_{1},j_{1},\\i_{2},j_{2},\\\ldots
,i_{D},j_{D}}}\sum_{\substack{k^{\prime},i_{1}^{\prime},j_{1}^{\prime}%
,\\i_{2}^{\prime},j_{2}^{\prime},\\\ldots,i_{D}^{\prime},j_{D}^{\prime}%
}}|k\rangle\!\langle j_{D}|_{0}\rho_{0}|j_{D}^{\prime}\rangle\!\langle k^{\prime
}|_{0}\nonumber\\
& \quad\otimes\left(  |i_{1}\rangle\!\langle k|_{1}\otimes|i_{1}\rangle\!\langle
j_{1}|_{D+1}\right)  L_{1}\Gamma_{1,D+1}L_{1}^{\prime}\left(  |k^{\prime
}\rangle\!\langle i_{1}^{\prime}|_{1}\otimes|j_{1}^{\prime}\rangle\!\langle
i_{1}^{\prime}|_{D+1}\right)  \nonumber\\
& \quad\otimes\left(  |i_{2}\rangle\!\langle j_{1}|_{2}\otimes|i_{2}%
\rangle\!\langle j_{2}|_{D+2}\right)  L_{2}\Gamma_{2,D+2}L_{2}^{\prime}\left(
|j_{1}^{\prime}\rangle\!\langle i_{2}^{\prime}|_{2}\otimes|j_{2}^{\prime}%
\rangle\!\langle i_{2}^{\prime}|_{D+2}\right)  \nonumber\\
& \quad\otimes\cdots\otimes\nonumber\\
& \quad\left(  |i_{D-1}\rangle\!\langle j_{D-2}|_{D-1}\otimes|i_{D-1}%
\rangle\!\langle j_{D-1}|_{2D-1}\right)  L_{D-1}\Gamma_{D-1,2D-1}\times
\nonumber\\
& \quad L_{D-1}^{\prime}\left(  |j_{D-2}^{\prime}\rangle\!\langle i_{D-1}%
^{\prime}|_{D-1}\otimes|j_{D-1}^{\prime}\rangle\!\langle i_{D-1}^{\prime
}|_{2D-1}\right)  \nonumber\\
& \quad\otimes\left(  |i_{D}\rangle\!\langle j_{D-1}|_{D}\otimes|i_{D}%
\rangle\!\langle j_{D}|_{2D}\right)  L_{D}\Gamma_{D,2D}L_{D}^{\prime}\left(
|j_{D-1}^{\prime}\rangle\!\langle i_{D}^{\prime}|_{D}\otimes|j_{D}^{\prime
}\rangle\!\langle i_{D}^{\prime}|_{2D}\right)  .
\end{align}
Now taking a partial trace over systems $1,\ldots,2D$ in the last line gives
the following:
\begin{align}
& \frac{1}{d^{D}}\sum_{\substack{k,i_{1},j_{1},\\i_{2},j_{2},\\\ldots
,i_{D},j_{D}}}\sum_{\substack{k^{\prime},i_{1}^{\prime},j_{1}^{\prime}%
,\\i_{2}^{\prime},j_{2}^{\prime},\\\ldots,i_{D}^{\prime},j_{D}^{\prime}%
}}|k\rangle\!\langle j_{D}|_{0}\rho_{0}|j_{D}^{\prime}\rangle\!\langle k^{\prime
}|_{0}\nonumber\\
& \quad\times\left(  \langle k|_{1}\otimes\langle j_{1}|_{D+1}\right)
L_{1}\Gamma_{1,D+1}L_{1}^{\prime}\left(  |k^{\prime}\rangle_{1}\langle
i_{1}^{\prime}|i_{1}\rangle\otimes|j_{1}^{\prime}\rangle_{D+1}\langle
i_{1}^{\prime}|i_{1}\rangle\right)  \nonumber\\
& \quad\times\left(  \langle j_{1}|_{2}\otimes\langle j_{2}|_{D+2}\right)
L_{2}\Gamma_{2,D+2}L_{2}^{\prime}\left(  |j_{1}^{\prime}\rangle_{2}\langle
i_{2}^{\prime}|i_{2}\rangle\otimes|j_{2}^{\prime}\rangle_{D+2}\langle
i_{2}^{\prime}|i_{2}\rangle\right)  \cdots\nonumber\\
& \quad\times\left(  \langle j_{D-2}|_{D-1}\otimes\langle j_{D-1}%
|_{2D-1}\right)  L_{D-1}\Gamma_{D-1,2D-1}\nonumber\\
& \quad\times L_{D-1}^{\prime}\left(  |j_{D-2}^{\prime}\rangle_{D-1}\langle
i_{D-1}^{\prime}|i_{D-1}\rangle\otimes|j_{D-1}^{\prime}\rangle_{2D-1}\langle
i_{D-1}^{\prime}|i_{D-1}\rangle\right)  \nonumber\\
& \quad\times\left(  \langle j_{D-1}|_{D}\otimes\langle j_{D}|_{2D}\right)
L_{D}\Gamma_{D,2D}L_{D}^{\prime}\left(  |j_{D-1}^{\prime}\rangle_{D}\langle
i_{D}^{\prime}|i_{D}\rangle\otimes|j_{D}^{\prime}\rangle_{2D}\langle
i_{D}^{\prime}|i_{D}\rangle\right)  \nonumber\\
& =\sum_{\substack{k,j_{1},j_{2},\\ \ldots,j_{D}}%
}\sum_{\substack{k^{\prime},j_{1}^{\prime},j_{2}^{\prime},\\\ldots
,j_{D}^{\prime}}}|k\rangle\!\langle j_{D}|_{0}\rho_{0}|j_{D}^{\prime}%
\rangle\!\langle k^{\prime}|_{0}\nonumber\\
& \quad\times\left(  \langle k|_{1}\otimes\langle j_{1}|_{D+1}\right)
L_{1}\Gamma_{1,D+1}L_{1}^{\prime}\left(  |k^{\prime}\rangle_{1}\otimes
|j_{1}^{\prime}\rangle_{D+1}\right)  \nonumber\\
& \quad\times\left(  \langle j_{1}|_{2}\otimes\langle j_{2}|_{D+2}\right)
L_{2}\Gamma_{2,D+2}L_{2}^{\prime}\left(  |j_{1}^{\prime}\rangle_{2}%
\otimes|j_{2}^{\prime}\rangle_{D+2}\right)  \cdots\nonumber\\
& \quad\times\left(  \langle j_{D-2}|_{D-1}\otimes\langle j_{D-1}%
|_{2D-1}\right)  L_{D-1}\Gamma_{D-1,2D-1}L_{D-1}^{\prime}\left(
|j_{D-2}^{\prime}\rangle_{D-1}\otimes|j_{D-1}^{\prime}\rangle_{2D-1}\right)
\nonumber\\
& \quad\times\left(  \langle j_{D-1}|_{D}\otimes\langle j_{D}|_{2D}\right)
L_{D}\Gamma_{D,2D}L_{D}^{\prime}\left(  |j_{D-1}^{\prime}\rangle_{D}%
\otimes|j_{D}^{\prime}\rangle_{2D}\right)  \\
& =\sum_{\substack{k,j_{1},j_{2},\\\ldots,j_{D}}%
}\sum_{\substack{k^{\prime},j_{1}^{\prime},j_{2}^{\prime},\\\ldots
,j_{D}^{\prime}}}|k\rangle\!\langle j_{D}|\rho_{0}|j_{D}^{\prime}\rangle\!\langle
k^{\prime}|_{0}\nonumber\\
& \quad\times\langle k|L_{1}|j_{1}\rangle\ \langle j_{1}^{\prime}%
|L_{1}^{\prime}|k^{\prime}\rangle\nonumber\\
& \quad\times\langle j_{1}|L_{2}|j_{2}\rangle\ \langle j_{2}^{\prime}%
|L_{2}^{\prime}|j_{1}^{\prime}\rangle\ \cdots\nonumber\\
& \quad\times\langle j_{D-2}|L_{D-1}|j_{D-1}\rangle\ \langle j_{D-1}^{\prime
}|L_{D-1}^{\prime}|j_{D-2}^{\prime}\rangle\nonumber\\
& \quad\times\langle j_{D-1}|L_{D}|j_{D}\rangle\ \langle j_{D}^{\prime}%
|L_{D}^{\prime}|j_{D-1}^{\prime}\rangle\\
& =\sum_{\substack{k,j_{1},j_{2},\\\ldots,j_{D}}%
}\sum_{\substack{k^{\prime},j_{1}^{\prime},j_{2}^{\prime},\\\ldots
,j_{D}^{\prime}}}|k\rangle\!\langle k|L_{1}|j_{1}\rangle\!\langle j_{1}%
|L_{2}|j_{2}\rangle\cdots\nonumber\\
& \quad\times\langle j_{D-2}|L_{D-1}|j_{D-1}\rangle\!\langle j_{D-1}|L_{D}%
|j_{D}\rangle\!\langle j_{D}|\rho_{0}|j_{D}^{\prime}\rangle\nonumber\\
& \quad\times\langle j_{D}^{\prime}|L_{D}^{\prime}|j_{D-1}^{\prime}%
\rangle\!\langle j_{D-1}^{\prime}|L_{D-1}^{\prime}|j_{D-2}^{\prime}\rangle
\cdots\langle j_{2}^{\prime}|L_{2}^{\prime}|j_{1}^{\prime}\rangle\!\langle
j_{1}^{\prime}|L_{1}^{\prime}|k^{\prime}\rangle\!\langle k^{\prime}|\\
& =\left(  L_{1}L_{2}\cdots L_{D-1}L_{D}\rho L_{D}^{\prime}L_{D-1}^{\prime
}\cdots L_{2}^{\prime}L_{1}^{\prime}\right)  _{0}\ .
\end{align}
This completes the proof of \eqref{eq:long-proof-cycswap-other}.
\end{proof}

\section{Proof of Theorem~\ref{thm:comp-main-thm}}

\label{app:D}

To begin with, we have
\begin{align}
    & \left\Vert \e^{\tilde{\mathcal{L}}t} -  \e^{\mathcal{L}t}\right\Vert_{\diamond} \notag \\
    & = \lim_{r \rightarrow \infty} \left\Vert (\e^{\tilde{\mathcal{L}}t/r})^{r} -  (\e^{\mathcal{L}t/r})^{r}\right\Vert_{\diamond}
    \label{eq:first-step-proof-app}\\
    & \leq  \lim_{r \rightarrow \infty} r \left\Vert \e^{\tilde{\mathcal{L}}t/r} -  \e^{\mathcal{L}t/r} \right\Vert_{\diamond}\\
    & = \lim_{r \rightarrow \infty} r \sup_{\omega \in \mathcal{D}(\mathcal{H}_{RS})} \left\Vert \e^{\tilde{\mathcal{L}}t/r}(\omega) -  \e^{\mathcal{L}t/r} (\omega)\right\Vert_{1}\\
    & \leq \lim_{r \rightarrow \infty} r \left( \sup_{\omega \in \mathcal{D}(\mathcal{H}_{RS})} \left\Vert \omega + \mathcal{L}(\omega) t/r -  \omega - \tilde{\mathcal{L}}(\omega)t/r \right\Vert_{1} + O(t^{2}/r^{2}) \right)\\
    & = \lim_{r \rightarrow \infty} r \left (\sup_{\omega \in \mathcal{D}(\mathcal{H}_{RS})} \left\Vert \mathcal{L}(\omega)- \tilde{\mathcal{L}}(\omega) \right\Vert_{1}t/r + O(t^{2}/r^{2}) \right)\\
    & = \lim_{r \rightarrow \infty} \sup_{\omega \in \mathcal{D}(\mathcal{H}_{RS})} \left\Vert \mathcal{L}(\omega)- \tilde{\mathcal{L}}(\omega) \right\Vert_{1}t + O(t^{2}/r) \\
    & = \sup_{\omega \in \mathcal{D}(\mathcal{H}_{RS})} \left\Vert \mathcal{L}(\omega)- \tilde{\mathcal{L}}(\omega) \right\Vert_{1}t \\
    & =  \left\Vert \mathcal{L}- \tilde{\mathcal{L}} \right\Vert_{\diamond} t.\label{eq:limsuptomo}
\end{align}
The first inequality employs the subadditivity of diamond distance under channel composition (the proof of this latter statement employs the triangle inequality, in a way similar to the proof of \cite[Eq.~(4.63)]{nielsen2010quantum}. For clarity, the identity channel's action is implicit in the second equality and thereafter. The second inequality applies a Taylor expansion and the triangle inequality. 

Now, for bounding $\left\Vert \mathcal{L}- \tilde{\mathcal{L}} \right\Vert_{\diamond}$, let $\omega \in \mathcal{D}(\mathcal{H}_{RS})$ and, again with implicit identities acting on $R$, consider that
\begin{align}
    & \left \Vert\mathcal{L}(\omega) - \tilde{\mathcal{L}}(\omega) \right \Vert_{1} \notag \\
    & = \left \Vert L\omega L^{\dagger} - \frac{1}{2} \left \{L^{\dagger}L, \omega \right\} - \tilde{L}\omega \tilde{L}^{\dagger} + \frac{1}{2} \left \{\tilde{L}^{\dagger}\tilde{L}, \omega \right\} \right \Vert_{1} \\
    & = \left \Vert L\omega L^{\dagger} - \frac{1}{2} \left (L^{\dagger}L \omega +  \omega L^{\dagger}L  \right) - \tilde{L}\omega \tilde{L}^{\dagger} + \frac{1}{2} \left (\tilde{L}^{\dagger}\tilde{L} \omega +  \omega \tilde{L}^{\dagger}\tilde{L}  \right) \right \Vert_{1} \\
    & = \left \Vert L\omega L^{\dagger} - \tilde{L}\omega \tilde{L}^{\dagger} - \frac{1}{2} \left (L^{\dagger}L \omega - \tilde{L}^{\dagger}\tilde{L} \omega  \right)  - \frac{1}{2} \left (\omega L^{\dagger}L  -  \omega \tilde{L}^{\dagger}\tilde{L}  \right) \right \Vert_{1} \\
    & \leq \left \Vert L\omega L^{\dagger} - \tilde{L}\omega \tilde{L}^{\dagger} \right \Vert_{1} + \frac{1}{2} \left \Vert  L^{\dagger}L \omega - \tilde{L}^{\dagger}\tilde{L} \omega   \right \Vert_{1}  + \frac{1}{2} \left \Vert \omega L^{\dagger}L  -  \omega \tilde{L}^{\dagger}\tilde{L}   \right \Vert_{1} \\
    & = \left \Vert L\omega L^{\dagger} - \tilde{L}\omega L^{\dagger} + \tilde{L}\omega L^{\dagger} - \tilde{L}\omega \tilde{L}^{\dagger} \right \Vert_{1} \notag \\
    & \qquad + \frac{1}{2} \left \Vert  L^{\dagger}L \omega - \tilde{L}^{\dagger}L \omega + \tilde{L}^{\dagger}L \omega - \tilde{L}^{\dagger}\tilde{L} \omega   \right \Vert_{1} \notag \\
    & \qquad + \frac{1}{2} \left \Vert \omega L^{\dagger}L - \omega \tilde{L}^{\dagger}L + \omega \tilde{L}^{\dagger}L -  \omega \tilde{L}^{\dagger}\tilde{L}   \right \Vert_{1}\\
    & \leq \left \Vert L\omega L^{\dagger} - \tilde{L}\omega L^{\dagger} \right \Vert_{1} + \left \Vert\tilde{L}\omega L^{\dagger} - \tilde{L}\omega \tilde{L}^{\dagger} \right \Vert_{1} \notag\\
    & \qquad + \frac{1}{2} \left \Vert  L^{\dagger}L \omega - \tilde{L}^{\dagger}L \omega \right \Vert_{1} + \frac{1}{2} \left \Vert \tilde{L}^{\dagger}L \omega - \tilde{L}^{\dagger}\tilde{L} \omega   \right \Vert_{1} \notag \\
    & \qquad + \frac{1}{2} \left \Vert   \omega L^{\dagger}L -  \omega\tilde{L}^{\dagger}L \right \Vert_{1}  + \frac{1}{2}  \left \Vert \omega \tilde{L}^{\dagger}L  -  \omega \tilde{L}^{\dagger}\tilde{L}  \right \Vert_{1} \\
    & \leq \left \Vert L - \tilde{L} \right \Vert_{2} \left \Vert \omega \right \Vert \left \Vert L^{\dagger} \right \Vert_{2}  +   \left \Vert \tilde{L} \right \Vert_{2} \left \Vert \omega \right \Vert \left \Vert L^{\dagger} - \tilde{L}^{\dagger} \right \Vert_{2}\notag \\
    & \qquad + \left \Vert \omega\right \Vert \left \Vert L \right \Vert_{2} \left \Vert L^{\dagger} - \tilde{L}^{\dagger} \right \Vert_{2} + \left \Vert \omega\right \Vert \left \Vert \tilde{L}^{\dagger} \right \Vert_{2} \left \Vert L - \tilde{L} \right \Vert_{2}\\
    & = O(\varepsilon/t) .
    \label{eq:lem2-first-term}
\end{align}
The first and second inequalities follow from the triangle inequality. The  last inequality follows from the generalized H\"older's inequality (see, e.g., \cite[Eq.~(8)]{beigi2013sandwiched}). The final equality follows from the assumption in the theorem statement. We thus conclude that
\begin{equation}
    \left\Vert \mathcal{L}- \tilde{\mathcal{L}} \right\Vert_{\diamond} = O(\varepsilon/t),
\end{equation}
and finally, by combining with \eqref{eq:first-step-proof-app}--\eqref{eq:limsuptomo} that
\begin{equation}
    \frac{1}{2}\left\Vert \e^{\tilde{\mathcal{L}}t} -  \e^{\mathcal{L}t}\right\Vert_{\diamond} \leq \left\Vert \mathcal{L}- \tilde{\mathcal{L}} \right\Vert_{\diamond} t = O(\varepsilon/t) t = O(\varepsilon).
\end{equation}

% Therefore, we can now rewrite \eqref{eq:limsuptomo} in the following way:
% \begin{align}
%     \frac{1}{2}\left\Vert \e^{\tilde{\mathcal{L}}t} -  \e^{\mathcal{L}t}\right\Vert_{\diamond} & \leq \lim_{r \rightarrow \infty} \sup_{\omega \in \mathcal{D}(\mathcal{H}_{S})} O(\varepsilon/t) t + O(t^{2}/r) \\
%     & = O(\varepsilon).
% \end{align}
% This concludes the proof.

\section{Proof of Lemma~\ref{lem:HS-distance-lem}}

\label{app:HS-distance-lem}

The steps are essentially the same as in the proof of \cite[Lemma~I.4]{dupuis2010decoupling}, but here we are considering bipartite vectors and their associated operators. We give the proof for convenience. Consider the following:
\begin{align}
    & \!\!\!\!\frac{1}{2}\left \Vert |\tilde \psi\rangle\!\langle \tilde \psi| - |\psi\rangle\!\langle \psi| \right \Vert_{1} \notag \\
    & = \frac{1}{2} \left \Vert (\tilde L \otimes I )|\Gamma\rangle\!\langle\Gamma| (\tilde L^\dagger \otimes I ) - (L \otimes I )|\Gamma\rangle\!\langle\Gamma| (L^\dagger \otimes I ) \right\Vert_{1}\\
    &  = \sqrt{1 - \left |\langle\Gamma| \left ( \tilde L^\dagger  L \otimes I \right) |\Gamma\rangle\right |^2 }\\
    & = \sqrt{1 - \left |\sum_{i, j}\langle i| \langle i| \left ( \tilde L^\dagger  L \otimes I \right) |j\rangle |j\rangle\right |^2 }\\
    & = \sqrt{1 - \left |\sum_{i, j}\langle i| \left ( \tilde L^\dagger  L  \right)|j\rangle  \otimes \langle i |j\rangle\right |^2 }\\
    & = \sqrt{1 - \left |\sum_{i}\langle i| \left ( \tilde L^\dagger  L  \right)|i\rangle \right |^2 }\\
    & = \sqrt{1 - \left |\operatorname{Tr}\!\left [\tilde L^\dagger L \right] \right|^2 }\\
    & = \sqrt{1 + \left |\operatorname{Tr}\!\left [\tilde L^\dagger L \right] \right| } \sqrt{1 - \left |\operatorname{Tr}\!\left [\tilde L^\dagger L \right] \right| }\\
    & \leq \sqrt{2 \left (1 - \left |\operatorname{Tr}\!\left [\tilde L^\dagger L \right] \right| \right)}\\
    & \leq \sqrt{2 \left (1 - \operatorname{Re}\!\left [\operatorname{Tr}\!\left [\tilde L^\dagger L \right] \right] \right)}\\
    & = \left \Vert \tilde{L} - L \right \Vert_{2}\\
    & \leq \delta.
\end{align}
The second equality follows from the fact that the trace distance between two pure quantum states can be expressed in terms of their inner product (see proof of \cite[Theorem~9.3.1]{wilde_2017}); i.e., for quantum states $|\psi\rangle$ and $|\phi\rangle$, we have 
\begin{equation}
    \frac{1}{2}\left \Vert |\phi\rangle\!\langle\phi| - |\psi\rangle\!\langle \psi| \right \Vert_{1} = \sqrt{1 - \left |\langle\psi|\phi\rangle\right |^2 }.
\end{equation}
The third-to-last inequality follows because $\left |\operatorname{Tr}\!\left [\tilde L^\dagger L \right] \right| \leq 1$, which is a consequence of the Cauchy--Schwarz inequality for the Hilbert--Schmidt inner product and the assumption that $\left \Vert L \right \Vert_{2} = \left \Vert \tilde L \right \Vert_{2} = 1$.
The last equality comes from the following chain of equalities:
\begin{align}
    \left \Vert \tilde{L} - L \right \Vert_{2} & = \sqrt{\operatorname{Tr}\!\left[\left ( \tilde{L}^{\dagger} - L^{\dagger} \right) \left( \tilde{L} - L \right)\right]}\\
    & = \sqrt{\operatorname{Tr}\!\left[\tilde{L}^{\dagger} \tilde L - \tilde{L}^{\dagger} L - L^{\dagger} \tilde L + L^{\dagger} L \right]}\\
    & = \sqrt{\operatorname{Tr}\!\left[\tilde{L}^{\dagger} \tilde L \right] - \operatorname{Tr}\!\left[ \tilde{L}^{\dagger} L \right ] - \operatorname{Tr}\!\left[L^{\dagger} \tilde L \right] + \operatorname{Tr}\!\left[L^{\dagger} L \right]}\\
   & = \sqrt{\left \Vert \tilde L \right \Vert_{2}^{2} - 2\operatorname{Re}\!\left [\operatorname{Tr}\!\left [\tilde L^\dagger L \right] \right] + \left \Vert L \right \Vert_{2}^{2}}\\
   & =\sqrt{2 \left (1 - \operatorname{Re}\!\left [\operatorname{Tr}\!\left [\tilde L^\dagger L \right] \right] \right)},
\end{align}
where the first equality follows from the definition of Hilbert--Schmidt norm (see \eqref{eq:schatten-norm-def}).

\bibliography{references}
\bibliographystyle{plain}

\end{document}